\documentclass[envcountsame]{llncs}

\usepackage{amsmath,amssymb,stmaryrd,xspace,pdfsync,mathrsfs,graphicx,refcount,verbatim}
\usepackage{microtype}
\usepackage{lmodern}
\usepackage{tikz}
\usepackage{enumerate}
\usepackage[mathscr]{euscript}
\usetikzlibrary{decorations.text,arrows,snakes,shapes,fit,shadows,calc,patterns,intersections}

\usepackage[textwidth=4cm]{todonotes}
\newcounter{sauvegarde}
\newcommand\adjustc[1]{%
  \protect{\setcounter{sauvegarde}{\thetheorem}
  \setcounterref{theorem}{#1}
  \addtocounter{theorem}{-1}
}}

\newcommand\restorec{
\setcounter{theorem}{\thesauvegarde}
}

\newcommand{\mynote}[3][]{\todo[caption={\sf #3}, color={%
    \ifnum#2=0 green!20
    \else\ifnum#2=1 orange!30
    \else\ifnum#2=2 yellow!20
    \else\ifnum#2=3 cyan!20
    \else magenta!20\fi\fi\fi\fi}, size=\tiny, #1]{\renewcommand{\baselinestretch}{1}\selectfont\sf#3}\xspace}

\definecolor{my1}{cmyk}{0,.6,0,0}
\definecolor{my2}{cmyk}{.3,.0,.0,.0}
\newcommand*{\swap}[2]{#2#1}
\newcommand{\efgame}{Ehrenfeucht-Fra\"iss\'e\xspace}
\newcommand\nat{\ensuremath{\mathbb{N}}\xspace}

\newcommand\Cs{\ensuremath{\mathcal{C}}\xspace}
\newcommand\Csgen[3]{\ensuremath{\Cs_{#1,#3}^{#2}}\xspace}
\newcommand\Cslev[1]{\ensuremath{\Cs_{#1}}\xspace}
\newcommand\Cslevk[2]{\ensuremath{\Cs_{#1}^{#2}}\xspace}
\newcommand\Csik{\ensuremath{\Cs_i^k}\xspace}
\newcommand\Cstwo{\ensuremath{\Cs_2}\xspace}
\newcommand\Cstwotwo{\ensuremath{\Cstwolen2}\xspace}
\newcommand\Csi{\ensuremath{\Cs_i}\xspace}
\newcommand\Csikn{\ensuremath{\Cs_{i,n}^k}\xspace}
\newcommand\Csitwo{\ensuremath{\Cs_{i,2}}\xspace}
\newcommand\Cstwolen[1]{\ensuremath{\Cs_{2,#1}}\xspace}
\newcommand\Cstwon{\ensuremath{\Cstwolen{n}}\xspace}
\newcommand\Csin{\ensuremath{\Cs_{i,n}}\xspace}
\newcommand\fCgen[3]{\ensuremath{\fC_{#1,#3}^{#2}}\xspace}
\newcommand\fCik{\ensuremath{\fC_i^k}\xspace}
\newcommand\fCi{\ensuremath{\fC_i}\xspace}
\newcommand\fCtwo{\ensuremath{\fC_2}\xspace}
\newcommand\fCtwotwo{\ensuremath{\fCtwolen2}\xspace}
\newcommand\fCikn{\ensuremath{\fC_{i,n}^k}\xspace}
\newcommand\fCin{\ensuremath{\fC_{i,n}}\xspace}
\newcommand\fCtwolen[1]{\ensuremath{\fC_{2,#1}}\xspace}
\newcommand\fCtwon{\ensuremath{\fC_{2,n}}\xspace}

\newcommand\Gs{\ensuremath{\mathcal{G}}\xspace}

\newcommand\Ss{\ensuremath{\mathcal{S}}\xspace}
\newcommand\Ts{\ensuremath{\mathcal{T}}\xspace}
\newcommand\Rs{\ensuremath{\mathcal{R}}\xspace}

\newcommand\ct{\ensuremath{\mathbb{T}}\xspace}
\newcommand\cs{\ensuremath{\mathbb{S}}\xspace}
\newcommand\crr{\ensuremath{\mathbb{U}}\xspace}

\newcommand\mat{\ensuremath{\mathscr{M}}\xspace}
\newcommand\mnat{\ensuremath{\mathscr{N}}\xspace}

\newcommand\pat{\ensuremath{\mathscr{P}}\xspace}

\newcommand{\dec}[1]{\ensuremath{\Delta_{#1}}\xspace}
\newcommand{\dew}[1]{\ensuremath{\Delta_{#1}(<)}\xspace}

\newcommand{\sic}[1]{\ensuremath{\Sigma_{#1}}\xspace}
\newcommand{\siw}[1]{\ensuremath{\Sigma_{#1}(<)}\xspace}
\newcommand{\pic}[1]{\ensuremath{\Pi_{#1}}\xspace}
\newcommand{\piw}[1]{\ensuremath{\Pi_{#1}(<)}\xspace}
\newcommand{\bsc}[1]{\ensuremath{\mathcal{B}\Sigma_{#1}}\xspace}
\newcommand{\bsw}[1]{\ensuremath{\mathcal{B}\Sigma_{#1}(<)}\xspace}

\newcommand{\decu}{\ensuremath{\Delta_{1}}\xspace}
\newcommand{\dewu}{\ensuremath{\Delta_{1}(<)}\xspace}

\newcommand{\sicu}{\ensuremath{\Sigma_{1}}\xspace}
\newcommand{\siwu}{\ensuremath{\Sigma_{1}(<)}\xspace}

\newcommand{\picu}{\ensuremath{\Pi_{1}}\xspace}
\newcommand{\piwu}{\ensuremath{\Pi_{1}(<)}\xspace}

\newcommand{\bscu}{\ensuremath{\mathcal{B}\Sigma_{1}}\xspace}
\newcommand{\bswu}{\ensuremath{\mathcal{B}\Sigma_{1}(<)}\xspace}

\newcommand{\decd}{\ensuremath{\Delta_{2}}\xspace}
\newcommand{\dewd}{\ensuremath{\Delta_{2}(<)}\xspace}

\newcommand{\sicd}{\ensuremath{\Sigma_{2}}\xspace}
\newcommand{\siwd}{\ensuremath{\Sigma_{2}(<)}\xspace}

\newcommand{\picd}{\ensuremath{\Pi_{2}}\xspace}
\newcommand{\piwd}{\ensuremath{\Pi_{2}(<)}\xspace}

\newcommand{\bscd}{\ensuremath{\mathcal{B}\Sigma_{2}}\xspace}
\newcommand{\bswd}{\ensuremath{\mathcal{B}\Sigma_{2}(<)}\xspace}

\newcommand{\dect}{\ensuremath{\Delta_{3}}\xspace}
\newcommand{\dewt}{\ensuremath{\Delta_{3}(<)}\xspace}

\newcommand{\sict}{\ensuremath{\Sigma_{3}}\xspace}
\newcommand{\siwt}{\ensuremath{\Sigma_{3}(<)}\xspace}

\newcommand{\pict}{\ensuremath{\Pi_{3}}\xspace}
\newcommand{\piwt}{\ensuremath{\Pi_{3}(<)}\xspace}

\newcommand{\bsct}{\ensuremath{\mathcal{B}\Sigma_{3}}\xspace}
\newcommand{\bswt}{\ensuremath{\mathcal{B}\Sigma_{3}(<)}\xspace}

\newcommand{\bspd}{\ensuremath{\mathcal{B}\Sigma_{2}(<,+1)}\xspace}

\newcommand{\sipt}{\ensuremath{\Sigma_{3}(<,+1)}\xspace}

\newcommand{\deci}{\ensuremath{\Delta_{i}}\xspace}
\newcommand{\dewi}{\ensuremath{\Delta_{i}(<)}\xspace}

\newcommand{\sici}{\ensuremath{\Sigma_{i}}\xspace}
\newcommand{\siwi}{\ensuremath{\Sigma_{i}(<)}\xspace}

\newcommand{\pici}{\ensuremath{\Pi_{i}}\xspace}
\newcommand{\piwi}{\ensuremath{\Pi_{i}(<)}\xspace}

\newcommand{\bsci}{\ensuremath{\mathcal{B}\Sigma_{i}}\xspace}
\newcommand{\bswi}{\ensuremath{\mathcal{B}\Sigma_{i}(<)}\xspace}

\newcommand{\mso}{\ensuremath{\textup{MSO}}\xspace}
\newcommand{\fo}{\ensuremath{\textup{FO}}\xspace}
\newcommand{\fow}{\ensuremath{\textup{FO}(<)}\xspace}

\newcommand{\savenotation}{
\let\dews\dew
\let\siws\siw
\let\piws\piw
\let\bsws\bsw
\let\dewus\dewu
\let\siwus\siwu
\let\piwus\piwu
\let\bswus\bswu
\let\dewds\dewd
\let\siwds\siwd
\let\piwds\piwd
\let\bswds\bswd
\let\dewts\dewt
\let\siwts\siwt
\let\piwts\piwt
\let\bswts\bswt
\let\dewis\dewi
\let\siwis\siwi
\let\piwis\piwi
\let\bswis\bswi
\let\fows\fow
}
\savenotation

\newcommand{\lightennotation}{
\let\dew\dec
\let\siw\sic
\let\piw\pic
\let\bsw\bsc
\let\dewu\decu
\let\siwu\sicu
\let\piwu\picu
\let\bswu\bscu
\let\dewd\decd
\let\siwd\sicd
\let\piwd\picd
\let\bswd\bscd
\let\dewt\dect
\let\siwt\sict
\let\piwt\pict
\let\bswt\bsct
\let\dewi\deci
\let\siwi\sici
\let\piwi\pici
\let\bswi\bsci
\let\fow\fo  
}
\lightennotation

\newcommand{\restorenotation}{
\let\dew\dews
\let\siw\siws
\let\piw\piws
\let\bsw\bsws
\let\dewu\dewus
\let\siwu\siwus
\let\piwu\piwus
\let\bswu\bswus
\let\dewd\dewds
\let\siwd\siwds
\let\piwd\piwds
\let\bswd\bswds
\let\dewt\dewts
\let\siwt\siwts
\let\piwt\piwts
\let\bswt\bswts
\let\dewi\dewis
\let\siwi\siwis
\let\piwi\piwis
\let\bswi\bswis
\let\fow\fows
}

\newcommand\sieq[2]{\ensuremath{\lesssim^{#1}_{#2}}\xspace}
\newcommand\ksieq[1]{\sieq{k}{#1}}

\newcommand\bceq[2]{\ensuremath{\cong^{#1}_{#2}}\xspace}
\newcommand\kbceq[1]{\bceq{k}{#1}}

\newcommand\gmo{\ensuremath{\geqslant}\xspace}
\newcommand\lmo{\ensuremath{\leqslant}\xspace}

\let\leq\leqslant

\let\geq\geqslant

\newcommand\Sep{\ensuremath{\mathsf{Sep}}\xspace}

\newcommand\content[1]{\ensuremath{\contentmorphism(#1)}}

\newcommand\contentmorphism{\ensuremath{\textsf{alph}}}

\newcommand\val[1]{\ensuremath{\textsf{val}(#1)\xspace}}

\newcommand\cval[1]{\ensuremath{\textsf{cval}(#1)\xspace}}

\newcommand\chain{chain\xspace}
\newcommand\qchain[1]{\ensuremath{\sic{#1}}-chain\xspace}
\newcommand\chains{chains\xspace}
\newcommand\qchains[1]{\ensuremath{\sic{#1}}-chains\xspace}
\newcommand\Chain{Chain\xspace}

\newcommand\Chains{Chains\xspace}
\newcommand\qChains[1]{\ensuremath{\sic{#1}}-Chains\xspace}

\newcommand\qpchains[2]{\ensuremath{\sic{#1}[#2]}-chains\xspace}

\newcommand\ichain{\qchain{i}}

\newcommand\dchain{\qchain{2}}
\newcommand\ichains{\qchains{i}}

\newcommand\dchains{\qchains{2}}

\newcommand\iChains{\qChains{i}}

\newcommand\dChains{\qChains{2}}

\newcommand\ikchains{\qpchains{i}{k}}

\newcommand\fI{\ensuremath{\mathfrak I}\xspace}
\newcommand\fM{\ensuremath{\mathfrak M}\xspace}

\newcommand\fC{\ensuremath{\mathfrak C}\xspace}
\newcommand\fS{\ensuremath{\mathfrak S}\xspace}

\newcommand\fO{\ensuremath{\mathfrak O}\xspace}
\newcommand\fT{\ensuremath{\mathfrak T}\xspace}

\DeclareMathOperator{\downclos}{\downarrow}

\tikzstyle{nor}=[minimum size=0.35cm,draw,rectangle,inner sep=2pt]
\tikzstyle{nod}=[minimum size=0.35cm,draw,circle,inner sep=2pt]
\tikzstyle{nof}=[minimum size=0.35cm,draw,circle,double,double
distance=1pt]

\tikzstyle{nol}=[minimum size=0.35cm,draw,rectangle,inner sep=1pt,rotate=90]

\tikzstyle{ar}=[line width=0.5pt,->,double]
\tikzstyle{siar}=[line width=1.5pt,->]

\newtheorem{fact}[theorem]{Fact}

\usepackage{enumitem}
\usepackage{hyperref}

\title{Going higher in the First-order\\ Quantifier Alternation Hierarchy on Words\thanks{Supported by ANR 2010 BLAN 0202 01 FREC}}
\author{Thomas~Place and Marc~Zeitoun}
\institute{LaBRI, Universit\'e de Bordeaux, France}

\begin{document}
\maketitle

\begin{abstract}
  We investigate the quantifier alternation hierarchy in first-order
  logic on finite words. Levels in this hierarchy are defined by counting the
  number of quantifier alternations in formulas. We prove that one can
  decide membership of a regular language to the levels \bscd (boolean
  combination of formulas having only $1$ alternation) and \sict
  (formulas having only $2$ alternations beginning with an existential
  block). Our proof works by considering a deeper problem, called
  separation, which, once solved for lower levels, allows us to solve
  membership for higher~levels.
\end{abstract}

\label{sec:intro}
The connection between logic and automata theory is well known and has a
fruitful history in computer science. It was first observed when B\"uchi,
Elgot and Trakhtenbrot proved independently that the regular languages are
exactly those that can be defined using a monadic second-order logic (\mso)
formula. 
Since then, many efforts have been made to investigate and understand
the expressive power of relevant fragments of \mso. In this field, the
yardstick result is often to prove \emph{decidable
  characterizations}, \emph{i.e.}, to design an algorithm which, given as input a
regular language, decides whether it can be defined in the fragment under
investigation. More than the algorithm itself, the main motivation is
the insight given by its proof. Indeed, in order to prove a decidable
characterization, one has to consider and understand \emph{all}
properties that can be expressed in the~fragment.

The most prominent fragment of \mso is first-order logic (\fo)
equipped with a predicate "$<$" for the linear-order. The expressive
power of \fo is now well-understood over words and a decidable
characterization has been obtained. The result, Sch\"utzenberger's
Theorem~\cite{sfo,mnpfo}, states that a regular language is definable
in \fo if and only if its syntactic monoid is aperiodic. The syntactic
monoid is a finite algebraic structure that can effectively be
computed from any representation of the language. Moreover,
aperiodicity can be rephrased as an equation that needs to be
satisfied by all elements of the monoid. Therefore, Sch\"utzenberger's
Theorem can indeed be used to decide definability in \fo.

In this paper, we investigate an important hierarchy inside \fo, obtained by
classifying formulas according to the number of quantifier alternations in
their prenex normal form. More precisely, an \fo formula is \sici if its
prenex normal form has at most $(i-1)$ quantifier alternations and starts with
a block of existential quantifiers. The hierarchy also involves the classes
\bsci of boolean combinations of \sici formulas, and the classes \deci of
languages that can be defined by both a \sici and the negation of a \sici
formula. The quantifier alternation hierarchy was proved to be
strict~\cite{BroKnaStrict,ThomStrict}: $\deci \subsetneq \sici \subsetneq
{\bsci} \subsetneq \dec{i+1}$. In the
literature, 
many efforts have been made to find decidable characterizations of levels of
this well-known~hierarchy.

Despite these efforts, only the lower levels are known to be
decidable. The~class \bscu consists exactly of all piecewise testable
languages, \emph{i.e.}, such that membership of a word only depends on its
subwords up to a fixed~size. These languages were characterized by
Simon~\cite{simon75} as those whose syntactic monoid is
$\mathcal{J}$-trivial. A decidable characterization of \sicd (and hence of
\dewd as well) was proven in~\cite{arfi87}.  For~\dewd, the literature is very
rich~\cite{Tesson02diamondsare}. For example, these are exactly the languages
definable by the two variable restriction of \fo~\cite{twfodeux}. These are also those whose syntactic monoid
is in the class~$\textup{\sf DA}$~\cite{pwdelta}. For higher levels in the hierarchy, getting
decidable characterizations remained an important open problem. In particular,
the case of \bscd has a very rich history and a series of combinatorial,
logical, and algebraic conjectures have been proposed over the years.  We refer 
to~\cite{Pin-ThemeVar2011,AK2010,pinbridges,pin-straubing:upper} for an
exhaustive bibliography. So far, the only known effective result was partial,
working only when the alphabet is of size $2$~\cite{StrauDD2}. One of the main
motivations for investigating this class in formal language theory is its ties
with two other famous hierarchies defined in terms of regular expressions. In
the first one, the \emph{Straubing-Th\'erien
  hierarchy}~\cite{StrauConcat,TheConcat}, level $i$ corresponds exactly to
the class \bsci~\cite{Thom82}. In the second one, the \emph{dot-depth
  hierarchy}~\cite{BrzoDot}, level $i$ corresponds to adding a predicate for
the successor relation in \bsci~\cite{Thom82}. Proving decidability for \bscd
immediately proves decidability of level $2$ in the Straubing-Th\'erien
hierarchy, but also in the dot-depth hierarchy using a reduction by
Straubing~\cite{StrauVD}.

In this paper, we prove decidability for \bscd, \dect and \sict. These new
results are based on a deeper decision problem than decidable
characterizations: the separation problem. Fix a class \Sep of languages. The
\Sep-separation problem amounts to decide whether, given two input regular
languages, there exists a third language in \Sep containing the first language
while being disjoint from the second one. This problem generalizes
decidable characterizations. Indeed, since regular languages are closed under
complement, testing membership in \Sep can be achieved by testing whether the
input is \Sep-separable from its complement. Historically, the separation problem
was first investigated as a special case of a deep problem in semigroup
theory, see~\cite{MR1709911}. This line of research gave solutions to the
problem for several classes. 
However, the motivations are disconnected from our own, and the proofs rely on
deep, purely algebraic arguments. Recently, a research effort has been made to
investigate this problem from a different perspective, with the aim of finding
new and self-contained proofs relying on elementary ideas and notions from
language theory only~\cite{martens,pvzmfcs13,pzfo,pvzltt}. This paper is a
continuation of this effort: we solve the separation problem for \sicd, and
use our solution as a basis to obtain decidable characterizations for \bscd,
\dect and~\sict.

Our solution works as follows: given two regular languages, one can
easily construct a monoid morphism $\alpha: A^* \rightarrow M$ that
recognizes both of them. We then design an algorithm that computes, inside the
monoid $M$,
enough \sicd-related information to answer the \sicd-separation
question for \emph{any} pair of languages that are recognized by
$\alpha$. It turns out that it is also possible (though much more
difficult) to use this information to obtain decidability of \bscd,
\dect and \sict. This information amounts to the notion of \dchain,
our main tool in the paper. A \dchain is an \emph{ordered
  sequence} $s_1,\dots,s_n \in M$ that witnesses a property of
$\alpha$ wrt.\ \sicd. 
Let us give some intuition in the case $n = 2$ -- which is enough to make the
link with \sicd-separation. A sequence $s_1,s_2$ is a \dchain if
any \sicd language containing all words in $\alpha^{-1}(s_1)$ 
intersects $\alpha^{-1}(s_2)$. In terms of separation, this means that
$\alpha^{-1}(s_1)$ is \emph{not} separable from $\alpha^{-1}(s_2)$ by a
$\sicd$ definable language.

This paper contains three main separate and difficult new results:
(1) an algorithm to compute \dchains\ -- hence \sicd-separability is decidable (2)
decidability of \sict (decidability of \dect is an immediate
consequence), and (3) decidability of~\bscd. Computing \dchains is
achieved using a fixpoint algorithm that starts with trivial \dchains
such as $s,s,\dots,s$, and iteratively computes more \dchains until a
fixpoint is reached. Note that its completeness proof relies on
the Factorization Forest Theorem of Simon~\cite{simonfacto}. This is
not surprising, as the link between this theorem and the quantifier
alternation hierarchy was already observed in~\cite{pwdelta,bfacto}.


For \sict, we prove a decidable characterization via an
equation on the syntactic monoid of the language. This equation is  
parametrized by the set of \dchains of length $2$. In other words,
we use \dchains to abstract an infinite set of equations into a
single one. The proof relies again on the Factorization Forest Theorem
of Simon~\cite{simonfacto} and is actually generic to all levels in
the hierarchy. This means that for any $i$, we define a notion of
\ichain and characterize \sic{i+1} using an equation parametrized by
\ichains of length $2$. However, decidability of \sic{i+1}
depends on our ability to compute the \ichains of length $2$, which we
can only do for $i =2$.

Our decidable characterization of \bscd is the most difficult result of the
paper. As for \sict, it is presented by two equations parametrized by \dchains
(of length $2$ and $3$). However, the characterization is this time specific
to the case $i = 2$. This is because most of our proof relies on a deep
analysis of our algorithm that computes \dchains, which only works for $i
=2$. The equations share surprising similarities with the ones used
in~\cite{bpopen} to characterize a totally different formalism: boolean
combination of open sets of infinite trees. In~\cite{bpopen} also, the authors
present their characterization as a set of equations parametrized by a notion
of ``\chain'' for open sets of infinite trees (although their ``\chains'' are
not explicitly identified as a separation relation). Since the formalisms are
of different nature, the way these \chains and our \dchains are constructed
are completely independent, which means that the proofs are also mostly
independent. However, once the construction analysis of \chains has been done,
several combinatorial arguments used to make the link with equations are
analogous. In particular, we reuse and adapt definitions from~\cite{bpopen} to
present these combinatorial arguments in our proof. One could say that the
proofs are both (very different) setups to apply similar combinatorial
arguments in the end.

\noindent
{\bf Organization.} We present definitions on languages
and logic in Sections~\ref{sec:words} and~\ref{sec:logic}
respectively. Section~\ref{sec:chains} is devoted to the presentation
of our main tool: \ichains. In Section~\ref{sec:comput}, we give our
algorithm computing \dchains. The two remaining sections present
our decidable characterizations, for \sict and \dect in
Section~\ref{sec:caracsi} and for \bscd in
Section~\ref{sec:caracbc}. Due to lack of space, proofs
can be found in~\cite{pz:qalt:2014}.


\section{Words and Algebra}
\label{sec:words}
\newcommand\one{\textup{1}}

\medskip
\noindent
{\textbf{Words and Languages.}} We fix a finite alphabet $A$ and we
denote by $A^{*}$ the set of all words over $A$. If $u,v$ are words, we 
denote by $u \cdot v$ or $uv$ the word obtained by concatenation of $u$ and
$v$. If $u \in A^*$ we denote by \content{u} its alphabet, \emph{i.e.}, the
smallest subset $B$ of $A$ such that $u\in B^*$. A \emph{language} is a subset of $A^*$. In
this paper we consider regular languages: these are languages
definable by \emph{nondeterministic finite automata}, or
equivalently by \emph{finite monoids}. In the paper, we only work with
the monoid representation of regular languages.

\medskip
\noindent
{\textbf{Monoids.}} A \emph{semigroup} is a set $S$ equipped with an
associative multiplication denoted by '$\cdot$'. A \emph{monoid} $M$
is a semigroup in which there exists a neutral element denoted
$1_M$. In the paper, we investigate classes of languages, such as
\siwi, that are not closed under complement. For such classes, it is
known that one needs to use \emph{ordered monoids}.  An ordered
monoid is a monoid endowed with a partial order '$\lmo$' which is
compatible with multiplication: $s\lmo t$ and $s'\lmo t'$
imply~$ss'\lmo tt'$.
Given any finite semigroup $S$, it is well known
that there is a number $\omega(S)$ (denoted by $\omega$ when $S$ is
understood from the context) such that for each element $s$ of $S$,
$s^\omega$ is an idempotent: $s^\omega = s^\omega \cdot s^\omega$. 

Let $L$ be a language and $M$ be a monoid. We say that \emph{$L$ is
  recognized by $M$} if there exists a monoid morphism $\alpha : A^*
\rightarrow M$ and an \emph{accepting set} $F \subseteq M$ such that
$L=\alpha^{-1}(F)$. It is well known that a language is regular if and only if it
can be recognized by a \emph{finite monoid}.


\medskip
\noindent {\textbf{Syntactic Ordered Monoid of a Language.}} 
The \emph{syntactic preorder} $\lmo_L$ of a language $L$ is defined as follows
on pairs of words in $A^*$: $w \lmo_L w'$ if for all $u,v \in A^*$, $uwv \in L
\Rightarrow uw'v \in L$. Similarly, we define $\equiv_L$, the \emph{syntactic
  equivalence} of $L$ as follows: $w \equiv_L w'$ if $w \lmo_L w'$ and $w'
\lmo_L w$. One can verify that $\lmo_L$ and $\equiv_L$ are compatible with
multiplication. Therefore, the quotient $M_L$ of $A^*$ by $\equiv_L$ is an
ordered monoid for the partial order induced by the preorder~$\lmo_L$.  It is
well known that $M_L$ can be effectively computed from $L$. Moreover, $M_L$
recognizes $L$. We call $M_L$ the \emph{syntactic ordered monoid of $L$} and
the associated morphism the \emph{syntactic morphism}.

\medskip\noindent
{\bf Separation.} Given three languages $L,L_0,L_1$, we say that $L$
\emph{separates} $L_0$ from $L_1$ if $L_0 \subseteq L \text{ and } L_1
\cap L = \emptyset$. Set $X$ as a class of languages, we say that
$L_0$ is \emph{$X$-separable} from $L_1$ if some language in $X$
separates $L_0$ from $L_1$. Observe that when $X$ is not closed under
complement, the definition is not symmetrical: $L_0$ could be
$X$-separable from $L_1$ while $L_1$ is not $X$-separable from $L_0$.

When working on separation, we consider as input two regular languages
$L_0,L_1$. It will be convenient to have a \emph{single} monoid recognizing
both of them, rather than having to deal with two objects. Let $M_0,
M_1$ be monoids recognizing $L_0,L_1$ together with the morphisms
$\alpha_0,\alpha_1$, respectively. Then, $M_0 \times M_1$ equipped
with the componentwise multiplication $(s_0,s_1) \cdot (t_0,t_1)=(s_0
t_0,s_1 t_1)$ is a monoid that recognizes both $L_0$ and $L_1$ with
the morphism $\alpha : w \mapsto (\alpha_0(w),\alpha_1(w))$. From now
on, we work with such a single monoid recognizing both languages.


\medskip
\noindent
{\bf \Chains and Sets of \Chains.} Set $M$ as a finite monoid. A
\emph{\chain } for $M$ is a word over the alphabet $M$, \emph{i.e.}, an
element of $M^*$. A remark about notation is in order here. A word is
usually denoted as the concatenation of its letters. Since $M$ is a
monoid, this would be ambiguous here since $st$ could either mean a
word with 2 letters $s$ and $t$, or the product of $s$ and $t$ in
$M$. To avoid confusion, we will write $(s_1,\dots,s_n)$ a \chain of
length $n$ on the alphabet $M$.

In the paper, we will consider both sets of \chains (denoted by
$\Ts,\Ss,\dots$) and sets of sets of \chains (denoted by $\fT,\fS,
\dots$). In particular, if $\fT$ is a set of sets of \chains, we define
$\downclos \fT$, the \emph{downset} of $\fT$, as the set:
\[
\downclos \fT=\{ \Ts \mid\exists \Ss \in \fT,\ \Ts \subseteq \Ss\}.
\]
We will often restrict ourselves to considering only \chains of a
given fixed length. For $n \in \nat$, observe that $M^n$, the set of
\chains of length $n$, is a monoid when equipped with the
componentwise multiplication. Similarly the set $2^{M^n}$ of sets of \chains of
length $n$ is a monoid for the operation: $\Ss \cdot \Ts =
\{\bar{s}\bar{t} \in M^n \mid \bar{s} \in \Ss \quad \bar{t} \in \Ts\}$.

\section{First-Order Logic and Quantifier Alternation Hierarchy}
\label{sec:logic}
We view words as logical structures made of a sequence of positions labeled
over~$A$. We denote by $<$ the linear order over the positions. We work with
first-order logic \fow using unary predicates $P_a$ for all $a \in A$ that
select positions labeled with an $a$, as well as a binary predicate for the
linear order $<$. The \emph{quantifier rank} of an \fow formula is the length
of its longest sequence of nested quantifiers.

One can classify first-order formulas by counting the number of
alternations between $\exists$ and $\forall$ quantifiers in the prenex
normal form of the formula. Set $i \in \nat$, a formula is said to be
\siw{i} (resp.\ \piw{i}) if its prenex normal form has $i -1$
quantifier alternations (\emph{i.e.}, $i$ blocks of quantifiers) and starts
with an $\exists$  (resp.\ $\forall$) quantification. For example, a
formula whose prenex normal form is
\[
\forall x_1 \forall x_2 \exists x_3 \forall x_4
\ \varphi(x_1,x_2,x_3,x_4) \quad \text{(with $\varphi$ quantifier-free)}
\]
\noindent
is \piwt. Observe that a \piw{i} formula is by definition the negation of a
\siw{i} formula. Finally, a \bsw{i} formula is a boolean combination of
\siw{i} formulas. For $X = \fow,\siw{i},\piw{i}$ or $\bsw{i}$, we say that a
language $L$ is $X$-definable if it can be defined by an $X$-formula. Finally,
we say that a language is \dew{i}-definable if it can be defined by
\emph{both} a \siw{i} and a \piw{i} formula. It is known that this gives a
strict infinite hierarchy of classes of languages as represented in
Figure~\ref{fig:hiera}.

\tikzstyle{non}=[inner sep=1pt]
\tikzstyle{tag}=[draw,fill=white,sloped,circle,inner sep=1pt]
\begin{figure}[h]
  \begin{center}
    \begin{tikzpicture}

      \node[non] (d1) at (0.0,0.0) {\decu};
      \node[non] (s1) at ($(d1)+(1.0,-0.8)$) {\sicu};
      \node[non] (p1) at ($(d1)+(1.0,0.8)$) {\picu};
      \node[non] (b1) at ($(d1)+(2.0,0.0)$) {\bscu};

      \node[non] (d2) at (3.5,0.0) {\decd};
      \node[non] (s2) at ($(d2)+(1.0,-0.8)$) {\sicd};
      \node[non] (p2) at ($(d2)+(1.0,0.8)$) {\picd};
      \node[non] (b2) at ($(d2)+(2.0,0.0)$) {\bscd};

      \node[non] (d3) at (7.0,0.0) {\dect};
      \node[non] (s3) at ($(d3)+(1.0,-0.8)$) {\sict};
      \node[non] (p3) at ($(d3)+(1.0,0.8)$) {\pict};
      \node[non] (b3) at ($(d3)+(2.0,0.0)$) {\bsct};

      \node[non] (d4) at (10.5,0.0) {\dec{4}};

      \draw[thick] (d1.south) to [out=-90,in=180] node[tag] {\scriptsize
        $\subsetneq$} (s1.west);
      \draw[thick] (d1.north) to [out=90,in=-180] node[tag] {\scriptsize
        $\subsetneq$} (p1.west);
      \draw[thick] (s1.east) to [out=0,in=-90] node[tag] {\scriptsize
        $\subsetneq$} (b1.south);
      \draw[thick] (p1.east) to [out=0,in=90] node[tag] {\scriptsize
        $\subsetneq$} (b1.north);
      \draw[thick] (b1.east) to [out=0,in=180] node[tag] {\scriptsize
        $\subsetneq$} (d2.west);

      \draw[thick] (d2.south) to [out=-90,in=180] node[tag] {\scriptsize
        $\subsetneq$} (s2.west);
      \draw[thick] (d2.north) to [out=90,in=-180] node[tag] {\scriptsize
        $\subsetneq$} (p2.west);
      \draw[thick] (s2.east) to [out=0,in=-90] node[tag] {\scriptsize
        $\subsetneq$} (b2.south);
      \draw[thick] (p2.east) to [out=0,in=90] node[tag] {\scriptsize
        $\subsetneq$} (b2.north);
      \draw[thick] (b2.east) to [out=0,in=180] node[tag] {\scriptsize
        $\subsetneq$} (d3.west);

      \draw[thick] (d3.south) to [out=-90,in=180] node[tag] {\scriptsize
        $\subsetneq$} (s3.west);
      \draw[thick] (d3.north) to [out=90,in=-180] node[tag] {\scriptsize
        $\subsetneq$} (p3.west);
      \draw[thick] (s3.east) to [out=0,in=-90] node[tag] {\scriptsize
        $\subsetneq$} (b3.south);
      \draw[thick] (p3.east) to [out=0,in=90] node[tag] {\scriptsize
        $\subsetneq$} (b3.north);
      \draw[thick] (b3.east) to [out=0,in=180] node[tag] {\scriptsize
        $\subsetneq$} (d4.west);

      \draw[thick,dotted] ($(d4.east)+(0.1,0.0)$) to
      ($(d4.east)+(0.7,0.0)$);

    \end{tikzpicture}
  \end{center}
  \caption{Quantifier Alternation Hierarchy}
  \label{fig:hiera}
\end{figure}
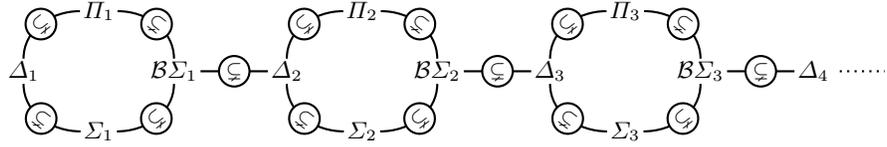


\medskip
\noindent {\bf Preorder for \siw{i}.} Let $w,w' \in A^*$ and $k,i \in
\nat$. We write $w \ksieq{i} w'$ if any \siw{i} formula of quantifier rank $k$
satisfied by $w$ is also satisfied by~$w'$. Observe that since a \piw{i}
formula is the negation of a \siw{i} formula, we have $w
\ksieq{i} w'$ iff any \piw{i} formula of quantifier rank $k$ satisfied by $w'$
is also satisfied by~$w$. One can verify that $\ksieq{i}$ is a preorder for
all $k,i$. Moreover, by definition, a language $L$ can be defined by a \siw{i}
formula of rank $k$ iff $L$ is saturated by $\ksieq{i}$, \emph{i.e.}, for all $w \in
L$ and all $w'$ such that $w \ksieq{i} w'$, we have $w' \in L$.

\section{\iChains}
\label{sec:chains}
We now introduce the main tool of this paper:
\emph{\ichains.} Fix a level~$i$ in the quantifier alternation
hierarchy and $\alpha : A^* \rightarrow M$ a monoid morphism. A
\emph{\ichain } for $\alpha$ is a \chain $(s_1,\dots,s_n) \in M^*$ such
that for arbitrarily large $k \in \nat$, there exist words $w_1
\ksieq{i} \cdots \ksieq{i} w_n$ mapped respectively to $s_1, \dots,
s_n$ by~$\alpha$. Intuitively, this contains information about the
limits of the expressive power of the logic $\siw{i}$ with respect to
$\alpha$. For example, if $(s_1,s_2)$ is a \ichain, then any $\siw{i}$
language that contains all words of image $s_1$ must also contain at
least one word of image~$s_2$.

\smallskip
In this section, we first give all definitions related to \ichains. We then
present an immediate application of this notion: solving the separation problem
for \siw{i} can be reduced to computing the \ichains of length $2$.

\subsection{Definitions}

\noindent
{\bf \iChains.} Fix $i$ a level in the hierarchy, $k \in \nat$ and $B
\subseteq A$. We define 
$\Csik[\alpha]$ (resp.\ $\Csik[\alpha,B]$) as
the \emph{set of \ikchains for $\alpha$} (resp.\ for $(\alpha,B)$) and
$\Csi[\alpha]$ (resp.\ $\Csi[\alpha,B]$) as
the \emph{set of \ichains for $\alpha$} (resp.\ for $(\alpha,B)$). For $i = 0$, we set 
$\Csi[\alpha] = \Csik[\alpha] = M^*$. Otherwise,
let $\bar{s} = (s_1,\dots,s_n) \in M^*$. We let

\begin{itemize}
\item $\bar{s} \in \Csik[\alpha]$ if there exist $w_1,\dots,w_n \in A^*$
  verifying $w_1 \ksieq{i} w_2 \ksieq{i} \cdots \ksieq{i} w_n$ and for all
  $j$, we have $\alpha(w_j)=s_j$. Moreover, $\bar{s} \in \Csik[\alpha,B]$ if
  the words $w_j$ can be chosen so that they satisfy additionally
  $\content{w_j}=B$ for all $j$.
\item $\bar{s} \in \Csi[\alpha]$ if for all $k$, we have $\bar{s} \in
  \Csik[\alpha]$. That is, $\Csi[\alpha]=\bigcap_k\Csik[\alpha]$. In the same
  way, $\Csi[\alpha,B]=\bigcap_k\Csik[\alpha,B]$.
\end{itemize}

One can check that if $i\geq2$, then
$\Csik[\alpha] =
\bigcup_{B\subseteq A} \Csik[\alpha,B]$, since  the fragment $\siwi$ can
detect the alphabet (\emph{i.e.}, for $i\geq2$, $w\ksieq{i}w'$
implies $\content{w}=\content{w'}$). Similarly for $i\geq2$, the set
of \ichains for $\alpha$ is $\Csi[\alpha]=\bigcup_{B\subseteq A}
\Csi[\alpha,B]$. Observe that all these sets 
are closed under subwords. Therefore, by Higman's lemma, we get the
following fact.

\begin{fact} \label{fct:high}
  For all $i,k \in \nat$ and $B \subseteq A$, $\Csi[\alpha,B]$ and
  $\Csik[\alpha,B]$ are regular languages.
\end{fact}

Fact~\ref{fct:high} is interesting but essentially useless in our
argument, as Higman's lemma provides no way for actually computing a
recognizing device for $\Csi[\alpha,B]$.

For any fixed $n \in \nat$, we let  $\Csikn[\alpha,B]$ be the
set of \ikchains of length~$n$ for $\alpha,B$,
\emph{i.e.}, $\Csikn[\alpha,B]=\Csik[\alpha,B] \cap M^n$. We define
$\Csin[\alpha,B],\Csikn[\alpha]$ and $\Csin[\alpha]$ similarly. The
following fact is immediate.

\begin{fact} \label{fct:chaincomp} If $B,C\subseteq A$, then
  $\Csikn[\alpha,B]\cdot\Csikn[\alpha,C] \subseteq\Csikn[\alpha,B\cup C]$.  In
  particular, $\Csikn[\alpha]$ and $\Csin[\alpha]$ (resp. $\Csikn[\alpha,B]$
  and $\Csin[\alpha,B]$) are submonoids (resp. subsemigroups) of~$M^n$.
\end{fact}



This ends the definition of \ichains. However, in order to define our
algorithm for computing \dchains and state our decidable
characterization of \bswd, we will need a slightly refined notion:
\emph{compatible sets of \chains}.

\medskip
\noindent
{\bf Compatible Sets of \iChains.} In some cases, it will be useful to
know that several \ichains with the same first element can be
`synchronized'. For example take two \ichains $(s,t_1)$ and $(s,t_2)$
of length $2$. By definition, for all $k$ there exist words
$w_1,w'_1,w_2,w'_2 $ whose images under $\alpha$ are $s,t_1,s,t_2$ respectively, and such that $w_1
\ksieq{i} w'_1$ and $w_2 \ksieq{i} w'_2$. In some cases (but not all),
it will be possible to choose $w_1 = w_2$ for all $k$. The goal of the
notion of compatible sets of \chains is to record the cases in which
this is true.

Fix $i$ a level in the hierarchy, $k \in \nat$ and $B \subseteq A$.
We define two sets of sets of~\chains:
$\fCik[\alpha,B]$, the \emph{set of compatible sets of \ikchains for}
($\alpha,B$), and
$\fCi[\alpha,B]$, the \emph{set of compatible sets of
  \ichains for} ($\alpha,B$).
Let \Ts be a set of \chains, all having the same
length 
$n$ and the same
first element~$s_1$.
\begin{itemize}
  \itemsep1ex
\item $\Ts \in \fCik[\alpha,B]$ if there exists $w \in A^*$ such
  that $\content{w} = B$, $\alpha(w) = s_1$, and for all \chains
  $(s_1,\dots,s_n) \in \Ts$, there exist $w_2,\dots,w_n \in A^*$
  verifying $w \ksieq{i} w_2 \ksieq{i} \cdots \ksieq{i} w_n$, and for all $j=2,\dots,n$,
  $\alpha(w_j)=s_j$, and $\content{w_j}=B$.
\item $\Ts \in \fCi[\alpha,B]$ if $\Ts \in
  \fCik[\alpha,B]$  for all $k$. 
\end{itemize}

As before we set $\fCik[\alpha]$ and $\fCi[\alpha]$ as the union
of these sets for all $B \subseteq A$. Moreover, we denote by
$\fCikn[\alpha,B],\fCin[\alpha,B],\fCikn[\alpha]$ and
$\fCin[\alpha]$ the restriction of these sets to sets of \chains of
length $n$ (\emph{i.e.}, subsets of $2^{M^n}$).

\begin{fact} \label{fct:setcomp} If $B,C\subseteq A$, then
  $\fCikn[\alpha,B]\cdot\fCikn[\alpha,C] \subseteq\fCikn[\alpha,B\cup C]$.  In
  particular, $\fCikn[\alpha]$ and $\fCin[\alpha]$ (resp. $\fCikn[\alpha,B]$
  and $\fCin[\alpha,B]$) are submonoids (resp. subsemigroups) of $2^{M^n}$.
\end{fact}

\subsection{\iChains and Separation}
We now state a reduction from the separation problem by $\siw{i}$ and by
$\piw{i}$-definable languages to the computation of \ichains of length 2.

\begin{theorem} \label{thm:sep}
  Let $L_1,L_2$ be regular languages and $\alpha: A^* \rightarrow M$ be
  a morphism into a finite monoid recognizing both languages with
  accepting sets $F_1,F_2 \subseteq M$. Set $i \in \nat$. Then the
  following properties hold:
  \begin{enumerate}
  \item $L_1$ is $\siwi$-separable from $L_2$ iff for all $s_1,s_2
    \in F_1,F_2$, $(s_1,s_2) \not\in \Csi[\alpha]$.
  \item $L_1$ is $\piwi$-separable from $L_2$ iff for all $s_1,s_2
    \in F_1,F_2$, $(s_2,s_1) \not\in \Csi[\alpha]$.
  \end{enumerate}
\end{theorem}



The proof of Theorem~\ref{thm:sep}, which is parametrized by \ichains, is
standard and identical to the corresponding theorems in previous separation
papers, see \emph{e.g.,}~\cite{pzfo}. In Section~\ref{sec:comput}, we present
an algorithm computing \ichains of length~2 at level $i = 2$ of the
alternation hierarchy (in fact, our algorithm needs to compute the more
general notion of sets of compatible \dchains). This makes
Theorem~\ref{thm:sep} effective for \siwd and \piwd.


\section{Computing \dChains}
\label{sec:comput}
In this section, we give an algorithm for computing all \dchains and
sets of compatible \dchains of a given fixed length. We already know
by Theorem~\ref{thm:sep} that achieving this for length $2$ suffices
to solve the separation problem for \siwd and \piwd. Moreover, we will
see in Sections~\ref{sec:caracsi} and~\ref{sec:caracbc} that this
algorithm can be used to obtain decidable characterizations for 
\siwt, \piwt, \dewt and~\bswd. Note that in this section, we only provide the
algorithm and intuition on its correctness.

For the remainder of this section, we fix a morphism $\alpha: A^*
\rightarrow M$ into a finite monoid $M$. For any fixed $n \in \nat$
and $B \subseteq A$, we need to compute the following:

\begin{enumerate}
  \itemsep1ex
\item the sets $\Cstwon[\alpha,B]$ of \dchains of length $n$ for
  $\alpha$. 
\item the sets $\fCtwon[\alpha,B]$ of compatible subsets of
  $\Cstwon[\alpha,B]$.
\end{enumerate}

Our algorithm directly computes the second item, \emph{i.e.},
$\fCtwon[\alpha,B]$. More precisely, we compute the map $B \mapsto
\fCtwon[\alpha,B]$. Observe that this is enough to obtain
the first item since by definition, $\bar{s} \in \Cstwon[\alpha,B]$
iff $\{\bar{s}\} \in \fCtwon[\alpha,B]$. Note that going through
compatible subsets is necessary for the technique to work, even if we
are only interested in computing the map $B \mapsto \Cstwon[\alpha,B]$.

\medskip
\noindent {\bf Outline.} We begin by explaining what our algorithm does. For
this outline, assume $n = 2$. Observe that for all $w \in A^*$ such that
$\content{w} = B$, we have $\bigl\{(\alpha(w),\alpha(w))\bigr\} \in
\fCtwon[\alpha,B]$. The algorithm starts from these trivially compatible sets,
and then saturates them with two operations that preserve membership
in $\fCtwon[\alpha,B]$. Let us describe these two operations. The first one is
multiplication: if $\Ss \in \fCtwon[\alpha,B]$ and $\Ts \in \fCtwon[\alpha,C]$
then $\Ss \cdot \Ts \in \fCtwon[\alpha,B\cup C]$ by
Fact~\ref{fct:setcomp}. The main idea behind the second operation is to
exploit the following property of~\siwd:
\[
\forall k\ \exists \ell \text{~~~} w \ksieq{2} u, w \ksieq{2} u' \text{ and } \content{w'} =
\content{w} ~~~\Longrightarrow~~~ w^{2\ell} \ksieq{2} u^\ell w' u'^{\ell}.
\]
This is why compatible sets are needed: in order to use this property,
we need to have a single word $w$ such that $w \ksieq{2} u$ and $w
\ksieq{2} u'$, which is information that is not provided by \dchains. This
yields an operation that states that whenever $\Ss$ belongs to $\fCtwon[\alpha,B]$,
then so does $\Ss^\omega \cdot \Ts \cdot \Ss^\omega$, where $\Ts$ is the set of
\chains $(1_M,\alpha(w'))$ with $\content{w'} = B$. Let us now formalize
this procedure and generalize it to arbitrary length.

\medskip
\noindent
{\bf Algorithm.} As we explained, our algorithm works by fixpoint,
starting from trivial compatible sets. For all $n \in \nat$ and $B
\subseteq A$, we let $\fI_n[B]$ be the set $\fI_n[B] =
\bigl\{\{(\alpha(w),\dots,\alpha(w))\} \mid \content{w} = B\bigr\} \subseteq
2^{M_n}$. Our algorithm will start from the function $f_0:2^A\to2^{2^{M^n}}$
that maps any $C\subseteq A$ to $\fI_n[C]$.

Our algorithm is defined for any fixed length $n \gmo 1$.  We use a procedure
$Sat_n$ taking as input a mapping $f:2^A\to2^{2^{M^n}}$ and producing another such
mapping.  The algorithm starts from $f_0$ and iterates  $Sat_n$
until a fixpoint is reached.

When $n \gmo 2$, the procedure $Sat_n$ is parametrized by
$\Cstwolen{n-1}[\alpha,B]$, the sets of \dchains of length $n-1$, for $B
\subseteq A$. This means that in order to use $Sat_n$, one needs to have
previously computed the \dchains of length $n-1$ with~$Sat_{n-1}$.

\medskip
We now define the procedure $Sat_n$. If \Ss is a set of \chains of
length $n-1$ and $s \in M$, we write $(s,\Ss)$ for the set
$\{(s,s_1,\dots,s_{n-1}) \mid (s_1,\dots,s_{n-1}) \in \Ss\}$, which consists of
\chains of length $n$. Let $f:2^A\to2^{2^{M^n}}$ be a mapping, written $f=(C\mapsto\fT_C)$. For all $B \subseteq A$, we define a set
$Sat_n[B](f)$ in $2^{M^n}$. That is, $B \mapsto Sat_n[B](f)$ is again a mapping from
$2^A$ to $2^{2^{M^n}}$.
Observe that when $n = 1$,
there is no computation to do since for all $B$, $\fCtwolen1[\alpha,B] =
\fI_1[B]$ by definition. Therefore, we simply set $Sat_1[B](C \mapsto
\fT_C) = \fT_B$. When $n \gmo 2$, we define $Sat_n[B](C \mapsto
\fT_C)$ as the set $\fT_B \cup \fM_B \cup \fO_B$ with
\begin{eqnarray} 
  \fM_B & = & \bigcup_{C \cup D = B} (\fT_C \cdot \fT_D)\label{eq:mul} \\
  \fO_B & = & \big\{\Ts^\omega \cdot (1_M,\Cstwolen{n-1}[\alpha,B]) \cdot \Ts^{\omega} \mid \Ts \in
  \fT_B\big\}\label{eq:oper}
\end{eqnarray}
This ends the description of the procedure $Sat_n$.
We now formalize how to iterate~it. For any mapping $f: 2^A \rightarrow 2^{M^n}$ and any $B \subseteq A$ ,
we set $Sat^0_n[B](f) = f(B)$. For all $j \geq 1$, we set
$Sat^{j}_{n}[B](f) = Sat_n[B](C \mapsto Sat^{j-1}_n[C](f))$. By
definition of $Sat_n$, for all $j \gmo 0$ and $B \subseteq A$, we have
$Sat^j_n(f)[B] \subseteq Sat^{j+1}_n(f)[B] \subseteq
2^{M^n}$. Therefore, there exists $j$ such that $Sat^{j}_n[B](f) =
Sat^{j+1}_n[B](f)$. We denote by $Sat^{*}_n[B](f)$ this set. This
finishes the definition of the algorithm. Its correctness 
and completeness are stated in the following proposition. 

\begin{proposition} \label{prop:compu}
  Let $n \gmo 1$, $B \subseteq A$ and $\ell \gmo 3|M| \cdot
  2^{|A|}\cdot n \cdot 2^{2^{2|M|^n}}$. Then $$\fCtwon[\alpha,B] = 
  \fCgen{2}{\ell}{n}[\alpha,B] = \downclos Sat^{*}_n[B](C \mapsto \fI_n[C]).$$
\end{proposition} 
\noindent
Proposition~\ref{prop:compu} states correctness of the algorithm (the set
$\downclos Sat^{*}_n[B](C \mapsto \fI_n[C])$ \emph{only} consists of
compatible sets of \dchains) and completeness (this set contains \emph{all}
such sets). It also establishes a bound $\ell$. This bound is a byproduct of
the proof of the algorithm. It is of particular interest for separation and
Theorem~\ref{thm:sep}. Indeed, one can prove that for any two languages that are
\siwd-separable and recognized by $\alpha$, the separator can be chosen with
quantifier rank $\ell$ (for $n = 2$).

We will see in Sections~\ref{sec:caracsi} and~\ref{sec:caracbc} how to use Proposition~\ref{prop:compu}
to get decidable characterizations of \siwt, \piwt, \dewt and  \bswd. We
already state the following corollary as a consequence of
Theorem~\ref{thm:sep}.

\begin{corollary} \label{cor:decidsep}
  Given as input two regular languages $L_1,L_2$ it is decidable to test
  whether $L_1$ can be $\siw{2}$-separated
  (resp. $\piw{2}$-separated) from $L_2$.
\end{corollary}

\section{Decidable Characterizations of \siwt, \piwt, \dewt}
\label{sec:caracsi}
In this section we present our decidable characterizations for \dewt,
\siwt and \piwt. We actually give characterizations for all classes
\dew{i}, \siw{i} and \piw{i} in the quantifier alternation hierarchy.
The characterizations are all stated in terms of equations on
the syntactic monoid of the language. However, these equations are
parametrized by the \qchains{i-1} of length $2$. Therefore, getting
\emph{decidable} characterizations depends on our ability to compute the set
of \qchains{i-1} of length $2$, which we are only able to do for
$i \lmo 3$.  We begin by stating
our characterization for \siw{i}, and the characterizations for \piw{i}
and \dew{i} will then be simple corollaries.

\begin{theorem} \label{thm:caracsig}
  Let $L$ be a regular language and $\alpha: A^* \rightarrow M$ be its
  syntactic morphism. For all $i \gmo 1$, $L$ is definable in \siw{i}
  iff $M$ satisfies the following property: 
  \begin{equation}
    s^{\omega} \lmo s^{\omega}ts^{\omega} \quad \text{for all $(t,s) \in \Cslev{i-1}[\alpha]$}. \label{eq:sig}
  \end{equation}
\end{theorem}

It follows from Theorem~\ref{thm:caracsig} that it suffices to compute
the \qchains{i-1} of length $2$ in order to decide whether a
language is definable in \siw{i}. Also observe that when $i=1$,
by definition we have $(t,1_M) \in \Cslev0[\alpha]$ for all $t \in
M$. Therefore, \eqref{eq:sig} can be rephrased as $1_M \lmo t$ for all
$t \in M$, which is the already known equation for \siwu, see~\cite{pwdelta}. Similarly, when $i=2$, \eqref{eq:sig} can be
rephrased as $s^{\omega} \lmo s^{\omega}ts^{\omega}$ whenever $t$ is a
`subword' of $s$, which is the previously known equation for \siwd
(see~\cite{pwdelta,bfacto}).

The proof of Theorem~\ref{thm:caracsig} is done using Simon's
Factorization Forest Theorem and is actually a generalization of a
proof of~\cite{bfacto} for the special case of \siwd. 
Here, we state characterizations of \piw{i}
and \dew{i} as immediate corollaries. Recall that a language is
\piw{i}-definable if its complement is \siw{i}-definable, and that
it is \dew{i}-definable if it is both \siw{i}-definable and
\piw{i}-definable. 

\begin{corollary} \label{cor:caracpi}
  Let $L$ be a regular language and let $\alpha: A^* \rightarrow M$ be its
  syntactic morphism. For all $i \gmo 1$, the following properties hold:
  \begin{itemize}
  \item $L$ is definable in \piw{i} iff $M$ satisfies $s^{\omega}
    \gmo s^{\omega}ts^{\omega}$ for all $(t,s) \in \Cslev{i-1}[\alpha]$.
  \item $L$ is definable in \dew{i} iff $M$ satisfies $s^{\omega}
    = s^{\omega}ts^{\omega}$ for all $(t,s) \in \Cslev{i-1}[\alpha]$.
  \end{itemize}
\end{corollary}

We finish the section by stating decidability for the case $i=3$.
Indeed by Proposition~\ref{prop:compu}, one can compute the \dchains
of length $2$ for any morphism. Therefore, we get the following
corollary.

\begin{corollary} \label{cor:decid}  
  Definability of a regular language in \dewt, \siwt or \piwt is decidable.
\end{corollary}

\section{Decidable Characterization of \bswd}
\label{sec:caracbc}
In this section we present our decidable characterization for
\bswd. In this case, unlike Theorem~\ref{thm:caracsig}, the
characterization is specific to the case $i=2$ and does not generalize
as a non-effective characterization for all levels. The main reason is
that both the intuition and the proof of the characterization rests on
a deep analysis of our algorithm for computing \dchains, which is
specific to level $i = 2$. The characterization is stated as two
equations that must be satisfied by the syntactic morphism of the
language. The first one is parametrized by \dchains of length $3$, and
the second one by sets of compatible \dchains of length $2$ through a
more involved relation that we define below.

\medskip
\noindent
{\bf Alternation Schema.} Let $\alpha: A^* \rightarrow M$ be a monoid
morphism and let $B \subseteq A$. A \emph{$B$-schema} for $\alpha$ is
a triple $(s_1,s_2,s'_2) \in M^3$ such that there exist $\Ts \in
\fCtwo[\alpha,B]$ and $r_1,r'_1 \in M$ verifying $s_1 = r_1r'_1$,
$(r_1,s_2) \in \Cstwo[\alpha,B] \cdot \Ts^\omega$ and $(r'_1,s'_2) \in
\Ts^\omega \cdot \Cstwo[\alpha,B]$. Intuitively, the purpose of
$B$-schemas is to abstract a well-known property of \siwd on elements of
$M$: one can prove that if $(s_1,s_2,s'_2)$ is a $B$-schema, then for
all $k \in \nat$, there exist $w_1,w_2,w'_2 \in A^*$, mapped respectively to $s_1,
s_2,s'_2$ under $\alpha$, and such that for all $u \in B^*$, $w_1 \ksieq{2}
w_2uw'_2$. 

\begin{theorem} \label{thm:caracbc}
  Let $L$ be a regular language and $\alpha: A^* \rightarrow M$ be its
  syntactic morphism. Then $L$ is definable in \bswd iff $M$ satisfies the
  following properties:
  \begin{equation}
    \begin{array}{rcl}
      s_1^{\omega}s_3^{\omega} & = & s_1^{\omega}s_2s_3^{\omega} \\
      s_3^{\omega}s_1^{\omega} & = & s_3^{\omega}s_2s_1^{\omega}
    \end{array} \quad \text{for $(s_1,s_2,s_3) \in \Cstwo[\alpha]$} \label{eq:bcs1}
  \end{equation}

  \begin{equation}
    \begin{array}{c}
      (s_2t_2)^{\omega}s_1(t'_2s'_2)^{\omega} = (s_2t_2)^{\omega}s_2t_1s'_2(t'_2s'_2)^{\omega} \\
      \text{for $(s_1,s_2,s'_2)$ and $(t_1,t_2,t'_2)$ $B$-schemas for some
        $B \subseteq A$}
    \end{array}\label{eq:bcs2}
  \end{equation}
\end{theorem}

The proof of Theorem~\ref{thm:caracbc} is far more involved than that
of Theorem~\ref{thm:caracsig}. 
However, a simple consequence is decidability of
definability in \bswd. Indeed, it suffices to compute \dchains of
length $3$ and the $B$-schemas for all $B \subseteq A$ to check validity of both
equations. Computing this information is possible by
Proposition~\ref{prop:compu}, and therefore, we get the following
corollary.

\begin{corollary} \label{cor:decid2} Definability of a regular language in
  \bswd is decidable.
\end{corollary}



\section{Conclusion}
\label{sec:conc}
We solved the separation problem for \siwd using the new notion of \dchains,
and we used our solution to prove decidable characterizations for \bswd,
\dewt, \siwt and \piwt. The main open problem in this field remains to lift up
these results to higher levels in the hierarchy. In particular, we proved that
for any natural~$i$, generalizing our separation solution to \siwi (\emph{i.e.}, being able
to compute the \ichains of length~$2$) would yield a decidable
characterization for \siw{i+1}, \piw{i+1} and \dew{i+1}.

Our algorithm for computing \dchains cannot be directly generalized
for higher levels. An obvious reason for this is the fact that it
considers \dchains parametrized by sub-alphabets. This parameter is
designed to take care of the alternation between levels $1$ and $2$,
but is not adequate for higher levels. However, this is unlikely
to be the only problem. In particular, we do have an algorithm that
avoids using the alphabet, but it remains difficult to generalize. We 
leave the presentation of this alternate algorithm for further work.

\restorenotation
Another open question is to generalize our results to logical
formulas that can use a binary predicate $+1$ for the successor
relation. In formal languages, this corresponds to the well-known
\emph{dot-depth hierarchy}~\cite{BrzoDot}. It was proved
in~\cite{StrauVD} and~\cite{pinweilVD} that decidability of \bspd and
\sipt is a consequence of our results for \bswd and \siwt.
However, while the reduction itself is simple, its proof rely on deep
algebraic arguments. We believe that our techniques can be generalized
to obtain direct proofs of the decidability of \bspd and~\sipt.

\bibliographystyle{abbrv}

\appendix
\newpage

\section*{Appendix}
We divide this appendix into several sections. In Appendix~\ref{app:facto}, we
define the main tools we will use for our proofs: \efgame games and
factorization forests.  In Appendix~\ref{app:algo}, we complete
Section~\ref{sec:comput} by proving the correctness and completeness of our
algorithm for computing \dchains. In Appendix~\ref{app:sig}, we prove
Theorem~\ref{thm:caracsig}, i.e. our characterization of \siw{i} (which is
decidable for $i \lmo 3$). The remaining appendices are then devoted to the
proof of Theorem~\ref{thm:caracbc}, i.e. our decidable characterization of
\bswd. In Appendix~\ref{app:ctrees} we define \emph{\Chains Trees} which are
our main tool for proving the difficult direction of the characterization. In
Appendix~\ref{app:bc} we give an outline of the proof. Finally,
Appendix~\ref{app:depth} and Appendix~\ref{app:width} are devoted to proving
the two most difficult propositions in the proof.

\section{Tools}
\label{app:facto}
In this appendix we define \efgame games and factorization
forests. Both notions are well-known and we will use them
several times in our proofs.

\subsection{\efgame Games}

It is well known that the expressive power of logics can be
expressed in terms of games. These games are called \efgame games. We
define here the game tailored to the quantifier alternation hierarchy.

Before we give the definition, a remark is in order. There are actually two
ways to define the class of \siw{i}-definable languages.  First, one can
consider all first-order formulas and say that a formula is \siw{i} if it has
at most $i$ blocks of quantifiers once rewritten in prenex normal form. This
is what we do. However, one can also restrict the set of allowed formulas to
those that are already in prenex form and have at most $i$ blocks of
quantifiers. While this does not change the class of \siw{i}-definable
languages as a whole, this changes the set of formulas of quantifier rank $k$
for a fixed $k$. Therefore, this changes the preorder $\ksieq{i}$. This means
that there is a version of the \efgame game for each definition. In this
paper, we use the version that corresponds to the definition given in the
main part of the paper (\emph{i.e.}, the one considering all first-order formulas).

\medskip
\noindent
{\bf \efgame games.} Set $i$ a level in the quantifier alternation
hierarchy. We define the game for \siw{i}. The board of the game
consists of two words $w,w'\in A^*$ and there are two players called
\emph{Spoiler and Duplicator}. Moreover, there exists a
distinguished word among $w,w'$ that we call the \emph{active
  word}. The game is set to last a predefined number $k$ of
rounds. When the game starts, both players have $k$ pebbles. Moreover,
there are two parameters that get updated during the game, the active
word and a counter $c$ called the \emph{alternation counter}. Initially, $c$ is
set to $0$.

At the start of each round $j$, Spoiler chooses a word, either $w$ or
$w'$. Spoiler can always choose the active word, in which case both $c$
and the active word remain unchanged. However, Spoiler can only choose the
word that is not active when $c < i - 1$, in which case the active
word is switched and $c$ is incremented by $1$ (in particular this
means that the active word can be switched at most $i-1$ times). If
Spoiler chooses $w$ (resp. $w'$), he puts a pebble on a position $x_j$ in
$w$ (resp. $x'_j$ in $w'$).

Duplicator must answer by putting a pebble at a position $x'_j$ in
$w'$ (resp. $x_j$ in $w$). Moreover, Duplicator must ensure that all
pebbles that have been placed up to this point verify the following
condition: for all  $\ell_1,\ell_2 \lmo j$, the labels at positions
$x^{}_{\ell_1},x'_{\ell_1}$ are the same, and $x_{\ell_1} < x_{\ell_2}$ iff $x'_{\ell_1} <
x'_{\ell_2}$.  

Duplicator wins if she manages to play for all $k$ rounds, and Spoiler
wins as soon as Duplicator is unable to play.

\begin{lemma}[Folklore] \label{lem:efgame}
  For all $k,i \in \nat$ and $w,w' \in A^{*}$, $w \ksieq{i} w'$ iff
  Duplicator has a winning strategy for playing $k$ rounds in the
  \siw{i} game played on $w,w'$ with $w$ as the initial active word.
\end{lemma}

Note that we will often use Lemma~\ref{lem:efgame} implicitly and
alternate between the original and the game definition of
$\ksieq{i}$. We now give a few classical lemmas on \efgame games that
we reuse several times in our proofs. We begin with a 
lemma stating that $\ksieq{i}$ is a pre-congruence, \emph{i.e.} that it is
compatible with the concatenation product.

\begin{lemma} \label{lem:efconcat}
  Let $i \in \nat$ and let $w^{}_1,w^{}_2,w'_1,w'_2 \in A^*$ such that $w^{}_1
  \ksieq{i} w^{}_2$ and $w'_1 \ksieq{i} w'_2$. Then $w^{}_1w'_1 \ksieq{i}
  w^{}_2w'_2$.
\end{lemma}

\begin{proof}
  By Lemma~\ref{lem:efgame}, Duplicator has winning strategies in the level
  $i$ games between $w^{}_1,w^{}_2$ and $w'_1,w'_2$, with $w^{}_1,w'_1$ as initial
  active words respectively. These strategies can be easily combined into a
  strategy for the level $i$ game between $w^{}_1w'_1$ and $w^{}_2w'_2$ with
  $w^{}_1w'_1$ as intial active word. We conclude that $w^{}_1w'_1 \ksieq{i}
  w^{}_2w'_2$. \qed
\end{proof}

The second property concerns full first-order logic.

\begin{lemma} \label{lem:aperiodic} Let $k,k_1,k_2 \in \nat$ be such that
  $k_1,k_2 \gmo 2^{k}-1$. Let $v \in A^*$. Then
  $$\forall i\in\nat,\quad v^{k_1} \ksieq{i} v^{k_2}.$$
\end{lemma}

\begin{proof}
  This is well known for full first-order logic
  (see~\cite{bookstraub} for details).\qed
\end{proof}

We finish with another classical property, which is this time specific
to \siw{i}.

\begin{lemma} \label{lem:siprop}
  Let $i \in \nat$, let $k,\ell,r,\ell',r' \in \nat$ be such
  that $\ell,r,\ell',r' \gmo 2^k$ and let $u,v \in
  A^*$ such that $u \ksieq{i} v$. Then we have:
  \[
  v^{\ell}v^{r} \ksieq{i+1} v^{\ell'}uv^{r'}.
  \]
\end{lemma}

\begin{proof}
  Set $w = v^\ell v^{r}$ and $w' = v^{\ell'}uv^{r'}$. We prove
  that $w \ksieq{i+1} w'$ using an \efgame argument: we prove that Duplicator
  has a winning strategy for the game in $k$ rounds for \siw{i+1} played on $w,w'$ with $w$
  as initial active~word. The proof goes by induction on $k$. We
  distinguish two cases depending on the value, 0 or 1, of the alternation counter $c$
  after Spoiler has played the first round.

  \medskip
  \noindent
  {\bf Case 1: $c=1$.} In this case, by definition of the game, it
  suffices to prove that $w' \ksieq{i} w$. From our hypothesis we
  already know that $u \ksieq{i} v$. Moreover, it follows from
  Lemma~\ref{lem:aperiodic} that $v^{\ell'} \ksieq{i} v^\ell$
  and $v^{r'} \ksieq{i} v^{r - 1}$. It then follows from
  Lemma~\ref{lem:efconcat} that $w' \ksieq{i} w$.

  \medskip
  \noindent
  {\bf Case 2: $c=0$.} By definition, this means that Spoiler played on some
  position $x$ in $w$. Therefore $x$ is inside a copy of the word $v$. Since
  $w$ contains more than $2^{k+1}$ copies of $v$, by symmetry we can assume
  that there are at least $2^{k}$ copies of $v$ to the right of $x$. We now
  define a position $x'$ inside $w'$ that will serve as Duplicator's
  answer. We choose $x'$ so that it belongs to a copy of $v$ inside $w'$ and
  is at the same relative position inside this copy as $x$ is in its own copy
  of~$v$. Therefore, to fully define $x'$, it only remains to define the copy
  of $v$ in which we choose~$x'$. Let $n$ be the number of copies of $v$ to
  the \emph{left} of $x$ in $w$, that is, $x$ belongs to the
  $(n+1)$-th copy of $v$   starting from the left of $w$. If $n <
  2^{k-1}-1$, then $x'$ is chosen inside the $(n+1)$-th copy of $v$
  starting from the left of $w'$. Otherwise, $x'$ is chosen inside the
  $2^{k-1}$-th copy of $v$ starting from the left of $w'$. Observe
  that these copies always exist, since $\ell'\gmo 2^k$.

  Set $w=w_pvw_q$ and $w'=w'_pvw'_q$, with the two distinguished $v$ factors 
  being the copies containing the positions $x,x'$. By definition of
  the game, it suffices to prove that $w_p \sieq{k-1}{i+1}
  w'_p$ and $w_q \sieq{k-1}{i+1} w'_q$ to conclude that Duplicator
  can play for the remaining $k-1$ rounds. If $n < 2^{k-1}-1$, then by 
  definition, $w_p=w'_p$, therefore it is immediate that
  $w_p \sieq{k-1}{i+1} w'_p$. Otherwise, both $w_p$ and
  $w_p'$ are concatenations of at least $2^{k-1}-1$ copies of
  $v$. Therefore $w_p \sieq{k-1}{i+1} w'_p$ follows
  Lemma~\ref{lem:aperiodic}. Finally observe that by definition $w_q =
  v^{\ell_1}v^{r}$ and $w'_q=v^{\ell'_1}uv^{r'}$ with $\ell_1 + r \gmo
  2^k$ and $\ell'_1,r' \gmo 2^{k-1}$. Therefore, it is immediate
  by induction on $k$ that $w_q \sieq{k-1}{i+1} w'_q$. \qed
\end{proof}

\subsection{Simon's Facorization Forests Theorem}

In this appendix, we briefly recall the definition of factorization
forests and state the associated theorem. Proofs and 
more detailed presentations can be found in~\cite{kfacto,bfacto}

Let $M$ be a finite monoid and $\alpha: A^* \rightarrow M$ a
morphism. An \emph{$\alpha$-factorization forest} is an ordered
unranked tree with nodes labeled by words in $A^*$ and such that for
any inner node $x$ with label $w$, if $x_1,\dots,x_n$ are its children
listed from left to right with labels $w_1,\dots,w_n$, then
$w=w_1\cdots w_n$. Moreover, all nodes $x$ in the forest must be of
the three following kinds:

\begin{itemize}
\item \emph{leaf nodes} which are labeled by either a single letter or
  the empty word.
\item \emph{binary nodes} which have exactly two children.
\item \emph{idempotent nodes} which have an arbitrary number of
  children whose labels $w_1,\dots,w_n$ verify $\alpha(w_1) = \cdots =
  \alpha(w_n) = e$ for some idempotent $e \in M$.
\end{itemize}

If $w \in A^*$, an \emph{$\alpha$-factorization forest for $w$} is an
$\alpha$-factorization forest whose root is labeled by $w$.

\begin{theorem}[Factorization Forest Theorem of
  Simon~\cite{simonfacto,kfacto}] \label{thm:facto}
  For all $w \in A^*$, there exists an $\alpha$-factorization forest for
  $w$ of height smaller than $3|M|-1$.
\end{theorem}






\section{Appendix to Section~\ref{sec:comput}: Proving the Algorithm}
\label{app:algo}
In this appendix, we prove Proposition~\ref{prop:compu}, that it is
the correctness and completeness of our algorithm which computes sets
of compatible \dchains. Recall that our algorithm works by
fixpoint. Given as input a morphism $\alpha: A^* \rightarrow M$ into a
finite monoid $M$ and a natural $n \in \nat$, it applies iteratively
the procedure $Sat_n$, starting from the application $C \mapsto
\fI_n[C]$, where $\fI_n[C]$ is the set of trivial sets of compatible
\dchains of length $n$ for $\alpha,C$. The fixpoint is a collection of
sets indexed by subalphabets $B$, denoted by $Sat^{*}_n[B](C \mapsto
\fI_n[C])$. 

We have to show that when the algorithm reaches its fixpoint, the
computed set $\downclos Sat^{*}_n[B](C \mapsto \fI_n[C])$ consists
exactly of all compatible sets of \dchains of length~$n$. This is
formulated in Proposition~\ref{prop:compu}, which we restate. In
addition, it states that for every length $n$, one can compute a rank
$\ell(n)$ that suffices to capture all sets of compatible sets of
\dchains of length $n$. In the following, we~let
\[\ell(n)= 3|M| \cdot 2^{|A|}\cdot n\cdot 2^{2^{2|M|^n}}.\]

\adjustc{prop:compu}
\begin{proposition}
  Let $n \gmo 1$, $B \subseteq A$ and $\ell\geq\ell(n)$. Then $$\fCtwon[\alpha,B] = 
  \fCgen{2}{\ell}{n}[\alpha,B] = \downclos Sat^{*}_n[B](C \mapsto \fI_n[C]).$$
\end{proposition}
\restorec

We proceed by induction on $n$. Observe that when $n = 1$ all three
sets are by definition equal to $\fI_n[B]$, therefore the result is
immediate. Assume now that $n \geq 2$. Using our induction hypothesis
we have the following fact.

\begin{fact} \label{fct:inducomp}
  Let $B \subseteq A$, then $\fCtwolen{n-1}[\alpha,B] =
  \fCgen{2}{\ell(n-1)}{n-1}[\alpha,B]$. Moreover, it follows that $\Cstwolen{n-1}[\alpha,B] =
  \Csgen{2}{\ell(n-1)}{n-1}[\alpha,B]$.
\end{fact}


For all $B \subseteq A$, we prove the following inclusions:
$\fCtwon[\alpha,B] \subseteq \fCgen2{\ell}{n}[\alpha,B] \subseteq
\downclos Sat^{*}_n[B](C \mapsto \fI_n[C]) \subseteq
\fCtwon[\alpha,B]$. Observe that $\fCtwon[\alpha,B] \subseteq
\fCgen2{\ell}{n}[\alpha,B]$ is immediate by definition. Therefore, we
have two inclusions to prove:
\begin{itemize}
\item $\downclos Sat^{*}_n[B](C \mapsto \fI_n[C]) \subseteq
  \fCtwon[\alpha,B]$, this corresponds to correctness of the
  algorithm: all computed sets are indeed sets of compatible
  \dchains. 
\item $\fCgen2{\ell}{n}[\alpha,B] \subseteq \downclos Sat^{*}_n[B](C
  \mapsto \fI_n[C])$, this corresponds to completeness of the
  algorithm: all sets of compatible \dchains are computed.
\end{itemize}
\noindent
We give each proof its own subsection. Note that
Fact~\ref{fct:inducomp} (i.e., induction on $n$) is
only used in the completeness proof.

\subsection{Correctness of the Algorithm}
In this subsection, we prove that for all $B \subseteq A$, $\downclos
Sat^{*}_n[B](C \mapsto \fI_n[C]) \subseteq \fCtwon[\alpha,B]$. This is
a consequence of the following proposition.

\begin{proposition} \label{prop:correc}
  Set $B \subseteq A$, for all $k \in \nat$, $Sat^{*}_n[B](C \mapsto
  \fI_n[C]) \subseteq \fCgen2{k}{n}[\alpha,B]$.
\end{proposition}

Before proving Proposition~\ref{prop:correc}, we explain how it is
used to prove correctness. By definition, for all $B$,
$\fCtwon[\alpha,B] = \bigcap_{k \in \nat}
\fCgen2{k}{n}[\alpha,B]$. Therefore, it is immediate from the
proposition that $Sat^{*}_n[B](C \mapsto \fI_n[C]) \subseteq
\fCtwon[\alpha,B]$. Moreover, by definition, $\downclos
\fCtwon[\alpha,B] = \fCtwon[\alpha,B]$. We conclude that
$\downclos Sat^{*}_n[B](C \mapsto
\fI_n[C]) \subseteq \fCtwon[\alpha,B]$ which 
terminates the correctness proof. It now remains to prove
Proposition~\ref{prop:correc}.

Let $k \in \nat$, $B \subseteq A$ and $\Rs \in Sat^{*}_n[B](C \mapsto
\fI_n[C])$. We need to prove that $\Rs \in \fCgen2{k}{n}[\alpha,B]$. By
definition, $\Rs \in Sat^{j}_n[B](C \mapsto \fI_n[C])$ for some
$j$. We proceed by induction on $j$. If $j = 0$, this is immediate
since $\Rs \in \fI_n[B] \subseteq \fCgen2{k}{n}[\alpha,B]$.

Assume now that $j > 0$. For all $D \subseteq A$, we set $\fT_D =
Sat^{j-1}_n[D](C \mapsto \fI_n[C])$. By induction hypothesis, for every
$D\subseteq A$, every element of
$\fT_D$ belongs to $\fCgen2{k}{n}[\alpha,D]$.
Since $\Rs \in Sat^{j}_n[B](C \mapsto
\fI_n[C])$, by definition we have $\Rs \in
\fT_B \cup \fM_B \cup \fO_B$ with
\begin{eqnarray*}
  \fM_B & = & \bigcup_{C \cup D = B} (\fT_C \cdot \fT_D) \\
  \fO_B & = & \{\Ts^\omega \cdot (1_M,\Cstwolen{n-1}[\alpha,B]) \cdot
  \Ts^{\omega} \mid \Ts \in \fT_B\}
\end{eqnarray*}
If $\Rs \in \fT_B$, it is immediate by induction that $\Rs \in
\fCgen2{k}{n}[\alpha,B]$ and we are finished. Assume now that $\Rs \in
\fM_B$. This means that there exist $C,D$ such that $C \cup D =B$,
$\Ts_C \in \fT_C$ and $\Ts_D \in \fT_D$ such that $\Rs = \Ts_C \cdot
\Ts_D$. By induction hypothesis, we have $\Ts_C \in
\fCgen2{k}{n}[\alpha,C]$ and $\Ts_D \in \fCgen2{k}{n}[\alpha,D]$. It is
then immediate by Fact~\ref{fct:setcomp} that $\Rs = \Ts_C \cdot \Ts_D 
\in \fCgen2{k}{n}[\alpha,B]$.

It remains to treat the case when $\Rs \in \fO_B$. In that case, we
get $\Ts \in \fT_B$ such that $\Rs = \Ts^\omega \cdot
(1_M,\Cstwolen{n-1}[\alpha,B]) \cdot \Ts^{\omega}$. In the following, we
write $h = \omega \times 2^{2k}$ (with $\omega$ as
$\omega(2^{M^n})$). Note that by definition of the number $\omega$, we have
$\Ts^\omega = \Ts^h$, and in particular, $\Rs = \Ts^h \cdot
(1_M,\Cstwolen{n-1}[\alpha,B]) \cdot \Ts^{h}$. Observe first that by
induction hypothesis, we know that $\Ts \in \fCgen2{k}{n}[\alpha,B]$. In
particular, this means that all \chains in $\Ts$ have the same first
element. We denote by $t_1$ this element. By definition of
$\fCgen2{k}{n}[\alpha,B]$, we get $u \in A^*$ such that $\content{u} =
B$, $\alpha(u) = t_1$ and for all \chains $(t_1,\dots,t_n) \in \Ts$
there exist $u_2,\dots,u_n \in A^*$ satisfying $u \ksieq{2} u_2
\ksieq{2} \cdots \ksieq{2} u_n$ and for all $j$, $t_j=\alpha(u_j)$ and
$\content{u_j}=B$. 

\smallskip We now prove that $\Rs \in \fCgen2{k}{n}[\alpha,B]$. Set $w =
u^{2h}$ and $r_1 = \alpha(w) = t_1^{\omega}$, by definition $\content{w} =
B$. Observe that since $\Rs = \Ts^\omega \cdot (1_M,\Cstwolen{n-1}[\alpha,B])
\cdot \Ts^{\omega}$, every \chain in \Rs has $r_1$ as first element. We now
prove that for any \chain $(r_1,\dots,r_n) \in \Rs$, there exist
$w_2,\dots,w_n \in A^*$ satisfying $w \ksieq{2} w_2 \ksieq{2} \cdots \ksieq{2}
w_n$ and for all $j$, $r_j = \alpha(w_j)$ and $\content{w_j}=B$. By
definition, this will mean that $\Rs \in \fCgen2{k}{n}[\alpha,B]$. Set
$(r_1,\dots,r_n) \in \Rs$. By hypothesis, $(r_1,\dots,r_n) =
(t'_1t''_1,t'_2s_2t''_2,\dots,t'_ns_nt''_n)$ with $(t'_1,\dots,t'_n),$
$(t''_1,\dots,t''_n) \in \Ts^{h}$ and $(s_2,\dots,s_n) \in
\Cstwolen{n-1}[\alpha,B]$. In particular, $t'_1=t''_1 = t_1^{h} =
t_1^\omega$. Since $\Ts\in\fCtwon[\alpha,B]$, we have
$\Ts^h\in\fCtwon[\alpha,B]$, so we get
$w'_2,\dots,w'_n,w''_2,\dots,w''_n \in A^*$ such that for all $j$,
$\content{w'_j} = \content{w''_j} = B$, $\alpha(w'_j) = t'_j$ and
$\alpha(w''_j) = t''_j$ and we have:
\begin{eqnarray*}
  u^{h} \ksieq{2} w'_2 \ksieq{2} \cdots \ksieq{2} w'_n \\
  u^{h} \ksieq{2} w''_2 \ksieq{2} \cdots \ksieq{2} w''_n
\end{eqnarray*}
On the other hand, using the fact that $(s_2,\dots,s_n) \in
\Cstwolen{n-1}[\alpha,B]$, we get words $v_2,\dots,v_n \in A^*$, mapped to
$s_2,\dots,s_n$ by $\alpha$ and all having alphabet $B$, such that $v_2 \ksieq{2}
\cdots \ksieq{2} v_n$. For all $j \geq 2$, set $w_j =
w'_jv^{}_jw''_j$. Observe that for any $j \geq 2$, $\content{w_j} = B$ 
and $\alpha(w_j) = s_j$. Therefore it remains to prove that $w
\ksieq{2} w_2 \ksieq{2} \cdots \ksieq{2} w_n$ to terminate the
proof. That $w_2 \ksieq{2} \cdots \ksieq{2} w_n$ is immediate by
Lemma~\ref{lem:efconcat}. Recall that $w = u^{2h}$, therefore the
last inequality is a consequence of the following lemma.

\begin{lemma} \label{lem:efcorrec}
  $u^{h}u^{h} \ksieq{2} w'_2v^{}_2w''_2$
\end{lemma}

\begin{proof}
  By Lemma~\ref{lem:efconcat}, we have $u^{h} v_2 u^{h} \ksieq{2} 
  w'_2v^{}_2w''_2$. Therefore, it suffices to prove that $u^{h}u^{h}
  \ksieq{2} u^{h}v_2u^{h}$ to conclude. Recall that by definition
  $\content{v_2} = \content{u} = B$, therefore, it is straightforward to
  see that 

  \begin{equation} \label{eq:inegcorrec}
    v_2 \sieq{k}{1} u^{2^k}
  \end{equation}

  Moreover, we chose $h = \omega  \times   2^{2k}$. Therefore, it is immediate from Lemma~\ref{lem:siprop}
  and~\eqref{eq:inegcorrec} that $u^{h}u^{h} \ksieq{2}
  u^{h}v_2u^{h}$. \qed 
\end{proof}

\subsection{Completeness of the Algorithm}
We need to prove that for all $B \subseteq A$, we have
$\fCgen2{\ell}{n}[\alpha,B] \subseteq \downclos Sat^{*}_n[B](C \mapsto 
\fI_n[C])$ for $\ell \geq \ell(n)$, where $\ell(n)= 3|M| \cdot 2^{|A|}\cdot n \cdot
2^{2^{2|M|^n}}$. We do this by proving a slightly more general
proposition by induction. To state this proposition, we need more
terminology.

\medskip
\noindent
{\bf Generated Compatible Sets.} Set $k \in \nat$, $w \in A^*$ and $B
= \content{w}$. We set $\Gs_n^k(w) \in 2^{M^n}$ as the following set
of \chains of length $n$: $(t_1,\dots,t_{n}) \in \Gs_n^k(w)$ iff
$t_1 = \alpha(w)$ and there exists $w_2,\dots,w_{n} \in A^*$ satisfying
\begin{itemize}
\item for all $j$, $\alpha(w_j)=t_j$.
\item $w \ksieq{2} w_2 \ksieq{2} \cdots \ksieq{2} w_{n}$.
\end{itemize}
\noindent
Observe that the last item implies that all $w_j$ have the same
alphabet $\content{w}$. Therefore, by definition, any $\Gs_n^k(w)$ is a compatible set of
\dchains of length $n$: $\Gs_n^k(w) \in
\fCgen{2}{k}{n}[\alpha,\content{w}]$. Moreover, any compatible set of \dchains of
length $n$, $\Ts \in \fCgen{2}{k}{n}[\alpha,B]$ is a subset of $\Gs_n^k(w)$ for
some $w$ of alphabet $B$. We finish the definition with a decomposition lemma
that will be useful in the proof.

\begin{lemma} \label{lem:gendecomp}
  Let $w_1,\dots,w_{m+1} \in A^*$ and $k \in \nat$ with $k > m$, then: 
  \[
  \Gs_n^k(w_1\cdots w_{m+1}) \subseteq \Gs_n^{k-m}(w_1) \cdots
  \Gs_n^{k-m}(w_{m+1})
  \]
\end{lemma}

\begin{proof}
  Let $(s_1,\dots,s_n) \in \Gs_n^k(w_1\cdots w_{m+1})$. By definition,
  there exists $u_1,\dots,u_n$ such that $u_1=w_1 \cdots w_{m+1}$, for
  all $i$, $\alpha(u_i) = s_i$ and $u_1 \sieq{k}{2} \cdots \sieq{k}{2}
  u_n$. Using a simple \efgame argument, we obtain that all words $u_i$
  can be decomposed as $u_i = u_{i,1} \cdots u_{i,m+1}$ with $u_{1,1} =
  w_1,\dots,u_{1,m+1} = w_{m+1}$ and for all $j$: $u_{1,j} \sieq{k-m}{2}
  \cdots \sieq{k-m}{2} u_{n,j}$. For all $i,j$, set $s_{i,j} =
  \alpha(u_{i,j})$. By definition, for all $j$, $(s_{1,j},\dots,s_{n,j})
  \in \Gs_n^{k-m}(w_j)$. Moreover, we have 
  \[
  (s_1,\dots,s_n) = (s_{1,1},\dots,s_{n,1}) \cdots (s_{1,m+1},\dots,s_{n,m+1}).
  \]
  Therefore, we have $(s_1,\dots,s_n) \in \Gs_n^{k-m}(w_1) \cdots
  \Gs_n^{k-m}(w_{m+1})$ which terminates the proof. \qed
\end{proof}

We can now state our inductive proposition and prove that
$\fCgen2{\ell}{n}[\alpha,B] \subseteq \downclos Sat^{*}_n[B](C \mapsto
\fI_n[C])$. Set $\beta : A^* \rightarrow M \times 2^A$ defined as 
$\beta(w) =  (\alpha(w),\content{w})$.

\begin{proposition} \label{prop:comp}
  Let $B \subseteq A$, $j \in \nat$ and $w \in A^*$ that admits a
  $\beta$-factorization forest of height $h$ and such that $\content{w}
  = B$. Set $k \geq  h \cdot 2^{2^{2|M|^n}} + \ell(n-1)$, then $\Gs_n^k(w) \in
  \downclos Sat^{*}_n[B](C \mapsto \fI_n[C])$.
\end{proposition}

Before proving Proposition~\ref{prop:comp}, we explain how to use it
to terminate our completeness proof. Set $\Ts \in
\fCgen2{\ell}{n}[\alpha,B]$, by definition, this means that there exists
$w \in A^*$ such that $\content{w} = B$ and $\Ts \subseteq
\Gs_n^{\ell}(w)$. By Theorem~\ref{thm:facto}, we know that $w$ admits a
$\beta$-factorization forest of height at most $3|M|2^{|A|}$.
Therefore, by choice of $\ell$, we can apply
Proposition~\ref{prop:comp} and we obtain $\Gs_n^k(w) \in 
\downclos Sat^{*}_n[B](C \mapsto \fI_n[C])$. By definition of
$\downclos$ it is then immediate that $\Ts \in \downclos
Sat^{*}_n[B](C \mapsto \fI_n[C])$ which terminates the proof.

It remains to prove Proposition~\ref{prop:comp}. Note that this is
where we use Fact~\ref{fct:inducomp} (i.e. induction on $n$).
Set $w \in A^*$ that admits a $\beta$-factorization forest of height
$h$, $B = \content{w}$ and $k \geq h \times 2^{2^{3|M|^n}} +
\ell(n-1)$. We need to prove that $\Gs_n^k(w) \in \downclos
Sat^{*}_n[B](C \mapsto \fI_n[C])$, i.e., to construct $\Ts \in
Sat^{*}_n[B](C \mapsto \fI_n[C])$ such that $\Gs_n^k(w) \subseteq
\Ts$. The proof is by induction on the height $h$ of the factorization
forest of $w$. It works by applying the proposition inductively to
the factors given by this factorization forest. In particular, we
will use Lemma~\ref{lem:gendecomp} to decompose $\Gs_n^k(w)$
according to this factorization forest. Then, once the factors have
been treated by induction, we will use the definition of the
procedure $Sat_n$ (i.e. Operations~\eqref{eq:mul} and~\eqref{eq:oper})
to conclude. In particular, we will use the following fact several
times.

\begin{fact} \label{fct:semig}
  $\downclos Sat^{*}_n[B](C \mapsto \fI_n[C])$ is subsemigroup of $2^{M^n}$.
\end{fact}

\begin{proof}
  We prove that $Sat^{*}_n[B](C \mapsto \fI_n[C])$ is subsemigroup of
  $2^{M^n}$, the result is then immediate by definition of $\downclos$.
  Set $\Ss_1,\Ss_2 \in Sat^{*}_n[B](C \mapsto \fI_n[C])$. By definition
  of $Sat_n$ (see Operation~\eqref{eq:mul}), we have $\Ss_1 \cdot \Ss_2
  \in Sat_n[B](B \mapsto Sat^{*}_n[B](C \mapsto \fI_n[C])) =
  Sat^{*}_n[B](C \mapsto \fI_n[C])$.\qed
\end{proof}

We now start the induction. We distinguish three cases depending on
the nature of the topmost node in the $\beta$-factorization forest of
$w$.

\medskip
\noindent
{\bf Case 1: the topmost node is a leaf.} In that case, $h = 1$ and
$w$ is a single letter word $a \in A$. In particular $B = \content{w}
= \{a\}$. Observe that $k \geq 2$, therefore, one can verify that
$\Gs_n^k(a) = \{(\alpha(a),\dots,\alpha(a))\}$. It follows that
$\Gs_n^k(a) \in \fI_n[B]$ which terminates the proof for this  
case.

\medskip
\noindent
{\bf Case 2: the topmost node is a binary node.} We use induction on
$h$ and Operation~\eqref{eq:mul} in the definition of $Sat_n$. By
hypothesis $w = w_1 \cdot w_2$ with $w_1,w_2$ words admitting 
$\beta$-factorization forests of heights $h_1,h_2 \leq
h-1$. Set $B_1 = \content{w_1}$ and $B_2 = \content{w_2}$, by definition,
we have $B = B_1 \cup B_2$. Moreover, observe that
\[
k - 1 \geq (h-1) \cdot 2^{2^{2|M|^n}} + \ell(n-1).
\]
Therefore, we can apply our induction hypothesis to $w_1,w_2$ and we
obtain $\Ts_1 \in Sat^{*}_n[B_1](C \mapsto \fI_n[C])$ and $\Ts_2 \in
Sat^{*}_n[B_2](C \mapsto \fI_n[C])$ such that $\Gs_n^{k-1}(w_1)
\subseteq \Ts_1$ and $\Gs_n^{k-1}(w_2) \subseteq \Ts_2$. By
Operation~\eqref{eq:mul} in the definition of $Sat$, it is immediate
that $\Ts_1 \cdot \Ts_2 \in Sat^{*}_n[B](C \mapsto
\fI_n[C])$. Moreover, by Lemma~\ref{lem:gendecomp}, $\Gs_n^k(w)
\subseteq \Gs_n^{k-1}(w_1) \cdot \Gs_n^{k-1}(w_2) \subseteq \Ts_1
\cdot \Ts_2$. It follows that $\Gs_n^k(w) \in \downclos Sat^{*}_n[B](C
\mapsto \fI_n[C])$ which terminates this case.

\medskip
\noindent
{\bf Case 3: the topmost node is an idempotent node.} This is the 
most difficult case. We use induction on $h$,
Operation~\eqref{eq:oper} in the definition of $Sat_n$ and
Fact~\ref{fct:semig}. Note that this is also where
Fact~\ref{fct:inducomp} (i.e. induction on $n$ in the general proof of
Proposition~\ref{prop:compu}) is used. We begin by summarizing our
hypothesis: $w$ admits what we call an $(e,B)$-decomposition.

\medskip
\noindent
{\bf $(e,B)$-Decompositions.} Set $\widetilde{k} = (h-1) \cdot
2^{2^{2|M|^n}} + \ell(n-1)$, $e \in M$ an idempotent and $u \in A^*$. We
say that $u$ admits an \emph{$(e,B)$-decomposition} $u_1,\dots,u_m$
if
\begin{enumerate}[label=$\alph*)$,ref=\alph*]
\item\label{item:1} $u = u_1 \cdots u_m$,
\item\label{item:2} for all $j$, $\content{u_j} = B$ and
  $\alpha(u_j) = e$ and
\item\label{item:3} for all $j$, $\Gs_n^{\widetilde{k}}(w_j) \in \downclos Sat^{*}_n[B](C \mapsto
  \fI_n[C])$.
\end{enumerate}
Note that $\ref{item:2})$ means that $\beta(u_j)$ is a constant idempotent, where we recall that $\beta : A^* \rightarrow M
\times 2^A$ is the morphism defined by $\beta(w) =  (\alpha(w),\content{w})$.

\begin{fact} \label{fct:ebdecomp}
  $w$ admits an $(e,B)$-decomposition for some idempotent $e \in M$.
\end{fact}

\begin{proof}
  By hypothesis of Case~3, there exists a decomposition $w_1,\dots,w_m$
  of $w$ that satisfies points $a)$ and $b)$. Moreover, for all $j$,
  $w_j$ admits a $\beta$-factorization forest of height $h_j \leq
  h-1$. Therefore point $c)$ is obtained by induction hypothesis on the height
  $h$. \qed
\end{proof}

For the remainder of this case, we assume that the idempotent $e \in
M$ and the $(e,B)$-decomposition $w_1,\dots,w_m$ of $w$ are fixed. We
finish the definition, with the following useful fact, which follows
from Fact~\ref{fct:inducomp}.

\begin{fact} \label{fct:correc}
  Assume that $u$ admits an \emph{$(e,B)$-decomposition}
  $u_1,\dots,u_m$ and let $i \leq j \leq m$. Then, $\Gs_n^{\ell(n-1)}(u_{i}
  \cdots u_{j}) \subseteq (e,\Cstwolen{n-1}[B])$.
\end{fact}

\begin{proof}
  Let $(s_1,\dots,s_n) \in \Gs_n^{\ell(n-1)}(u_i\cdots u_j)$. Since $\alpha(u_{i}
  \cdots u_{j}) = e$, we have $s_1 =e$. Moreover, it is immediate from
  Fact~\ref{fct:inducomp} that $(s_2,\dots,s_n) \in
  \Cs_{n-1}^{2}[\alpha,B]$. We conclude that $(s_1,\dots,s_n) \in
  (e,\Cstwolen{n-1}[B])$. \qed
\end{proof}

Recall that we want to prove that $\Gs_n^k(w) \in \downclos
Sat^{*}_n[B](C \mapsto \fI_n[C])$. In general, the number of factors
$m$ in the $(e,B)$-decomposition of $w$ can be arbitrarily large. In
particular, it is possible that $k - (m-1) < \widetilde{k}$. This
means that we cannot simply use Lemma~\ref{lem:gendecomp} as we did in
the previous case to conclude that $\Gs_n^k(w) \subseteq
\Gs_n^{\widetilde{k}}(w_{1}) \cdots
\Gs_n^{\widetilde{k}}(w_{m})$. However, we will partition $w_1,\dots,w_m$ 
as a bounded number of subdecompositions that we can treat using
Operation~\eqref{eq:oper} in the definition of $Sat_n$. The partition
is given by induction on a parameter of the $(e,B)$-decomposition
$w_1,\dots,w_m$ that we define now.

\medskip
\noindent
{\bf Index of an $(e,B)$-decomposition.} Set $k_n = 2^{|M|^n}$ (the
size of the monoid $2^{M^n}$). Let $u \in A^*$ that admits an
$(e,B)$-decomposition $u_1,\dots,u_{m}$ and let $j \in \nat$ such that $1
\leq j \leq m - k_n$ (i.e. $j$ is the index of one of the first $m -
k_n$ factors in the decomposition). The \emph{$k_n$-sequence} occurring
at $j$ is the sequence $\Gs_n^{\tilde{k}}(w_j),\dots,
\Gs_n^{\tilde{k}}(w_{j+k_n}) \in \downclos Sat^{*}_n[B](C \mapsto
\fI_n[C])$. The \emph{index} of $u_1,\dots,u_{m}$ is the number of
$k_n$-sequences that occur in $u_1,\dots,u_{m}$. Observe that by
definition, there are at most $(k_n)^{k_n+1}$
$k_n$-sequences. Therefore the index of the decomposition is bounded
by $(k_n)^{k_n+1}$. We proceed by induction on the index of the
decomposition and state this induction in the following lemma.

\begin{lemma} \label{lem:inducase3}
  Let $u \in A^*$ admitting an $(e,B)$-decomposition $u_1,\dots,u_m$ of
  index $g$ and set $\widehat{k} \geq 2g + 2(k_n+1) + \widetilde{k} +
  \ell(n-1)$. Then $\Gs_n^{\widehat{k}}(u) \in \downclos Sat^{*}_n[B](C \mapsto \fI_n[C])$.
\end{lemma}

Before proving this lemma, we use it to conclude Case~3. We know that our
$(e,B)$-decomposition $w_1,\dots,w_m$ has an index $g \leq
(k_n)^{k_n+1}$. Therefore, it suffices to prove that $k \geq
2(k_n)^{k_n+1} + 2(k_n+1) + \widetilde{k} + \ell(n-1)$ to conclude that
$\Gs_n^k(w) \in \downclos Sat^{*}_n[B](C \mapsto \fI_n[C])$ using
Lemma~\ref{lem:inducase3}. One can verify that $2^{2^{2|M|^n}} \geq
2(k_n)^{k_n+1} + 2(k_n+1)$ as soon as $k_n>2$. It is then immediate that
\[
k = h \cdot 2^{2^{2|M|^n}} + \ell(n-1) \geq 2^{2^{2|M|^n}} + (h-1) \cdot
2^{2^{2|M|^n}} + \ell(n-1) = 2^{2^{2|M|^n}} + \widetilde{k} + \ell(n-1)
\]

\begin{proof}[of Lemma~\ref{lem:inducase3}]
  The proof goes by induction on the index $g$. We distinguish two cases
  depending on whether there exists a $k_n$-sequence that occurs at two
  different positions in the $(e,B)$-decomposition.

  Assume first that this is not the case, i.e., all $k_n$-sequences occurring at
  positions $1 \leq j \leq m - k_n$ are different. Since there are exactly $g$
  $k_n$-sequences occurring in the decomposition, a simple pigeon-hole
  principle argument yields that $m \leq g + k_n$. We use our choice of
  $\widehat{k}$ to conclude with a similar argument to the one we used in
  Case~2. By Lemma~\ref{lem:gendecomp}, we have: 
  \[
  \Gs_n^{\widehat{k}}(u) \subseteq
  \Gs_n^{\widehat{k}-(m-1))}(u_1) \cdots
  \Gs_n^{\widehat{k}-(m-1)}(u_m) 
  \]
  Observe that by hypothesis of this case, $\widehat{k} - (m-1) \geq 
  \widetilde{k}$. Therefore, by definition of $(e,B)$-decompositions,
  for all $j$, $\Gs_n^{\widehat{k}-(m-1))}(u_j) \in \downclos
  Sat^{*}_n[B](C \mapsto \fI_n[C])$. It is then immediate from
  Fact~\ref{fct:semig} that $\Gs_n^{\widehat{k}-(m-1))}(u_1) \cdots
  \Gs_n^{\widehat{k}-(m-1)}(u_m) \in \downclos
  Sat^{*}_n[B](C \mapsto \fI_n[C])$. We conclude that
  $\Gs_n^{\widehat{k}}(u) \in \downclos
  Sat^{*}_n[B](C \mapsto \fI_n[C])$ which terminates this case.

  \medskip
  Assume now that there exist $j,j' \in \nat$ such that $1 \leq j<j'
  \leq m - (k_n-1)$, and the
  $k_n$-sequences occurring at $j$ and $j'$ are the same. For the
  remainder of the proof, we set $\Rs_1,\dots,\Rs_{k_n+1}$ as this
  common $k_n$-sequence. Moreover, we assume that $j$ and $j'$ are
  chosen minimal and maximal respectively, i.e. there exists no
  $j'' < j$ or $j'' > j'$ such that $\Rs_1,\dots,\Rs_{k_n+1}$ occur
  at $j''$. By definition of a $k_n$-sequence, recall that we have
  $\Rs_1,\dots,\Rs_{k_n+1} \in \downclos Sat^{*}_n[B](C \mapsto
  \fI_n[C])$. Set
  \begin{align*}
    v_1 &= u_1 \cdots u_{j-1},\\
    v_2 &= u_{j} \cdots u_{j'+k_n}\\
    v_3 &= u_{j'+k_n+1} \cdots u_{m}.
  \end{align*}
  By Lemma~\ref{lem:gendecomp}, we know that
  \[
  \Gs_n^{\widehat{k}}(u) \subseteq
  \Gs_n^{\widehat{k}-2}(v_1) \cdot
  \Gs_n^{\widehat{k}-2}(v_2) \cdot
  \Gs_n^{\widehat{k}-2}(v_3) 
  \]
  We prove that for $i=1,2,3$, $\Gs_n^{\widehat{k}-2}(v_i) \in \downclos
  Sat^{*}_n[B](C \mapsto \fI_n[C])$. By Fact~\ref{fct:semig}, it will
  then be immediate that $\Gs_n^{\widehat{k}}(u) \in \downclos
  Sat^{*}_n[B](C \mapsto \fI_n[C])$ which terminates the proof. Observe
  that by choice of $j,j'$, $u_1,\dots,u_{j-1}$ and $u_{j'+k_n+1},\dots,
  u_{m}$ are $(e,B)$-decompositions of index smaller than $g$ (the
  $k_n$-sequence $\Rs_1,\dots,\Rs_{k_n+1}$ does not occur in these
  decompositions). Therefore, it is immediate by induction hypothesis on
  $g$ that $\Gs_n^{\widehat{k}-2}(v_1), \Gs_n^{\widehat{k}-2}(v_3) \in \downclos
  Sat^{*}_n[B](C \mapsto \fI_n[C])$.

  \medskip
  It remains to prove that $\Gs_n^{\widehat{k}-2}(v_2) \in \downclos
  Sat^{*}_n[B](C \mapsto \fI_n[C])$. If $j' \lmo j + k_n$, then $v_2$
  admits an $(e,B)$-decomposition of length smaller than $2(k_n+1)$ and
  we can conclude using Lemma~\ref{lem:gendecomp} as in the previous
  case. Therefore, assume that $j' > j + k_n$ and set $v = u_{j+k_n+1}
  \cdots u_{j'-1}$ and observe that by definition $v_2 = u_j \cdots
  u_{j+k_n} \cdot v \cdot u_{j'} \cdots u_{j'+k_n}$. Moreover,
  $\widehat{k} -2 - 2(k_n+1) \geq \widetilde{k}$, using
  Lemma~\ref{lem:gendecomp} we get that  
  \[
  \Gs_n^{\widehat{k}-2}(v_2) \subseteq \Gs_n^{\widetilde{k}}(u_j) \cdots
  \Gs_n^{\widetilde{k}}(u_{j+k_n}) \cdot \Gs_n^{\widetilde{k}}(v) \cdot
  \Gs_n^{\widetilde{k}}(u_{j'}) \cdots \Gs_n^{\widetilde{k}}(u_{j'+k_n})
  \]
  By definition $\Rs_1,\dots,\Rs_{k_n+1}$ is the $k_n$-sequence occurring
  at both $j$ and $j'$. Therefore, it follows that 
  \begin{equation}
    \Gs_n^{\widehat{k}-2}(v_2) \subseteq \Rs_1 \cdots \Rs_{k_n+1} \cdot
    \Gs_n^{\widetilde{k}}(v) \cdot \Rs_1 \cdots \Rs_{k_n+1} \label{eq:inc2}
  \end{equation}
  Intuitively, we want to find an idempotent in the sequence $\Rs_1
  \cdots \Rs_{k_n+1}$ in order to apply
  Operation~\eqref{eq:oper}. Observe that since the $\Rs_j$ are elements
  of the monoid $2^{M^n}$ and $k_n = 2^{|M|^n}$, the sequence $\Rs_1 \cdots
  \Rs_{k_n+1}$ must contain a ''loop.'' By this we mean that there
  exists $j_1 < j_2$ such that $\Rs_{1} \cdots \Rs_{j_1} = \Rs_1 \cdots
  \Rs_{j_2}$. Set $\Ss_1 = \Rs_{1} \cdots \Rs_{j_1}$, $\Ss_2 = \Rs_{j_1
    + 1} \cdots \Rs_{j_2}$ and $\Ss_3 = \Rs_{j_2 + 1} \cdots
  \Rs_{k_n+1}$. By definition of $\Ss_1,\Ss_2,\Ss_3$, we have $\Rs_1
  \cdots \Rs_{k_n+1} = \Ss_1 \cdot (\Ss_2)^\omega \cdot
  \Ss_3$. Note that by Fact~\ref{fct:semig}, we have $\Ss_1,\Ss_2,\Ss_3
  \in \downclos Sat^{*}_n[B](C \mapsto \fI_n[C])$. By replacing this
  in~\eqref{eq:inc2}, we get
  \[
  \Gs_n^{\widehat{k}-2}(u_2) \subseteq \Ss_1 \cdot
  (\Ss_2)^\omega \cdot \Ss_3 \cdot \Gs_n^{\tilde{k}}(v) \cdot \Ss_1 \cdot
  (\Ss_2)^\omega \cdot \Ss_3
  \]
  Moreover, observe that $\tilde{k} \geq \ell(n-1)$, therefore, using
  Fact~\ref{fct:correc}, we get that $\Ss_3 \cdot
  \Gs_n^{\tilde{k}}(v) \cdot \Ss_1 \subseteq
  (e,\Cstwolen{n-1}[B])$. Moreover, since all \chains in $\Ss_2$ have $e$
  as first element (see Fact~\ref{fct:correc}), it is immediate that
  $(\Ss_2)^\omega \cdot (e,\Cstwolen{n-1}[B]) \cdot (\Ss_2)^\omega =
  (\Ss_2)^\omega \cdot (1_M,\Cstwolen{n-1}[B]) \cdot (\Ss_2)^\omega$. This
  yields
  \[
  \Gs_n^{\widehat{k}-2}(u_2) \subseteq \Ss_1 \cdot
  (\Ss_2)^\omega \cdot (1_M,\Cstwolen{n-1}[B]) \cdot
  (\Ss_2)^\omega \cdot \Ss_3.
  \]
  Since $\Ss_2 \in \downclos Sat^{*}_n[B](C \mapsto \fI_n[C])$, it is
  immediate by Operation~\eqref{eq:oper} in the definition of $Sat_n$
  that $(\Ss_2)^\omega \cdot (1_M,\Cstwolen{n-1}[B]) \cdot (\Ss_2)^\omega
  \in \downclos Sat^{*}_n[B](C \mapsto \fI_n[C])$. It then follows from
  Fact~\ref{fct:semig} that $\Ss_1 \cdot
  (\Ss_2)^\omega \cdot (1_M,\Cstwolen{n-1}[B]) \cdot
  (\Ss_2)^\omega \cdot \Ss_3 \in  \downclos Sat^{*}_n[B](C \mapsto
  \fI_n[C])$ and therefore that $\Gs_n^{\widehat{k}-2}(u_2) \in
  \downclos Sat^{*}_n[B](C \mapsto \fI_n[C])$ which terminates the proof.\qed
\end{proof}

\section{Proof of Theorem~\ref{thm:caracsig}: Characterization of \siw{i}}
\label{app:sig}
In this appendix, we prove Theorem~\ref{thm:caracsig}, i.e., our
characterization for \siw{i}. For this whole appendix, we assume that
the level $i$ in the quantifier alternation hierarchy is fixed.

\adjustc{thm:caracsig}
\begin{theorem}
  Let $L$ be a regular language and $\alpha: A^* \rightarrow M$ be its
  syntactic morphism. For all $i \geq 1$, $L$ is definable in \siw{i}
  iff $\alpha$ satisfies: 
  \begin{equation}
    s^{\omega} \lmo s^{\omega}ts^{\omega} \quad \text{for all $(t,s) \in \Cslev{i-1}[\alpha]$} \tag{\ref{eq:sig}}
  \end{equation}
\end{theorem}
\restorec
There are two directions. We give each one its own subsection.

\subsection{Equation~\eqref{eq:sig} is necessary}

We prove that the syntactic morphism of any \siw{i}-definable language
satisfies~\eqref{eq:sig}. We state this in the following proposition.

\begin{proposition} \label{prop:signec}
  Let $L$ be a \siw{i}-definable language and let $\alpha : A^* 
  \rightarrow M$ be its syntactic morphism. Then $\alpha$
  satisfies~\eqref{eq:sig}.
\end{proposition}

\begin{proof}
  By hypothesis, $L$ is defined by some \siw{i}-formula $\varphi$. Let $k$ be
  its quantifier rank. Set $(t,s) \in \Cs_{i-1}[\alpha]$, we need to prove
  that $s^{\omega} \lmo s^{\omega}ts^{\omega}$. Since $(t,s) \in
  \Cs_{i-1}[\alpha]$, by definition, there exist $v,u$ such that $\alpha(v) =
  t$, $\alpha(u) = s$ and $v \ksieq{i-1} u$.  By Lemma~\ref{lem:siprop}, we
  immediately obtain
  \[
  u^{2^k\omega}\cdot u^{2^k\omega} \ksieq{i}u^{2^k\omega} \cdot v \cdot
  u^{2^k\omega}.
  \]
  It then follows from Lemma~\ref{lem:efconcat} that for any $w_1,w_2 \in A^*$
  we have:
  \begin{equation} \label{eq:efcorrec} w_1 \cdot u^{2^k\omega}\cdot
    u^{2^k\omega} \cdot w_2 ~~\ksieq{i}~~ w_1 \cdot u^{2^k\omega} \cdot v
    \cdot u^{2^k\omega} \cdot w_2.
  \end{equation}
  By definition, this means that $w_1 \cdot u^{2^k\omega} \cdot w_2 \in L$
  implies that $w_1 \cdot u^{2^k\omega} v u^{2^k\omega} \cdot w_2 \in
  L$. Which, by definition of the syntactic preorder, means that $s^{\omega}
  \lmo s^{\omega}ts^{\omega}$.\qed
\end{proof}

\subsection{Equation~\eqref{eq:sig} is sufficient}

It remains to prove that whenever $\alpha$ satisfies~\eqref{eq:sig},
$L$ is definable in \siw{i}. This is a consequence of the following
proposition.

\begin{proposition} \label{prop:signec}
  Let $L$ be a regular language such that its syntactic morphism $\alpha
  : A^* \rightarrow M$ satisfies~\eqref{eq:sig}. Then there exists $k
  \in \nat$ such that for all $u,v \in A^*$:
  \[
  u \ksieq{i} v \Rightarrow \alpha(u) \lmo \alpha(v)
  \]
\end{proposition}
Assume for now that Proposition~\ref{prop:signec} holds and let
$\alpha$ satisfy~\eqref{eq:sig}. Let then $u,v\in A^*$ with $u\in L$
and $u \ksieq{i} v$. By Proposition~\ref{prop:signec}, we deduce that
$\alpha(u)\lmo\alpha(v)$ which, by definition of the preorder $\lmo$,
implies that $v\in L$. Therefore, $\ksieq i$ saturates $L$, so $L$ is
definable in \siw{i}.

It remains to prove Proposition~\ref{prop:signec}. We begin by
choosing $k$. The choice depends on the following lemma. Recall that
$\Csgen{i}{k}{2}[\alpha]$ is the set of chains of length 2 belonging
to $\Csik[\alpha]$.

\begin{fact} \label{fct:boundk}
  For any morphism $\alpha: A^* \rightarrow M$ into a finite monoid $M$,
  there exists $k_i \in \nat$ such that for all $k \geq k_i$,
  $\Csgen{i}{k}{2}[\alpha] = \Csitwo[\alpha]$.
\end{fact}

\begin{proof}
  This is because for all $k < k'$, $\Csgen{i}{k'}{2}[\alpha] \subseteq
  \Csgen{i}{k}{2}[\alpha] \subseteq M^2$. Since $M^2$ is a finite set, there
  exists an index $k_i$ such that for all $k \leq k_i$,
  $\Csgen{i}{k}{2}[\alpha] = \Csgen{i}{k_i}{2}[\alpha]$. It is then
  immediate by definition that $\Csgen{i}{k_i}{2}[\alpha] = \Csitwo[\alpha]$.\qed
\end{proof}

Observe that while proving proving the existence $k_i$ is easy, the
proof is non-constructive and computing $k_i$ from $i,\alpha$ is a
difficult problem. In particular, having $k_i$ allows us to compute all
\ichains of length $2$ via a brute-force algorithm. When $i=2$, we
proved in Proposition~\ref{prop:compu} that it suffices to take $k_2 =
3|M| \cdot 2^{|A|} \cdot 2 \cdot  2^{2^{2|M|^2}}$.

We can now prove Proposition~\ref{prop:signec}. Set $k_{i-1}$ as defined
in Fact~\ref{fct:boundk} for $i-1$. This means that $(s,t)$ is a
\qchain{i-1} for $\alpha$ iff there exists $u,v \in A^*$ such that
$\alpha(u) = s$, $\alpha(v)=t$ and $u \sieq{k_{i-1}}{i-1} v$. We prove
that Proposition~\ref{prop:signec} holds for $k = 6|M|+k_{i-1}$. This
follows from the next lemma.

\begin{lemma} \label{lem:factosig}
  Let $h \in \nat$ and $u,v \in A^*$, such that $u$ admits an
  $\alpha$-factorization forest of height smaller than $h$. Then  
  \[
  u \sieq{2h+k_{i-1}}{i} v \Rightarrow \alpha(u) \lmo \alpha(v)
  \]
\end{lemma}

Observe that by Theorem~\ref{thm:facto} all words admit an
$\alpha$-factorization forest of height less than $3|M|$. Therefore,
Proposition~\ref{prop:signec} is an immediate consequence of
Lemma~\ref{lem:factosig}. It remains to prove the lemma.

\begin{proof}[of Lemma~\ref{lem:factosig}]
  We distinguish three cases depending on the nature of the topmost node
  in the $\alpha$-factorization forest of $u$. If the topmost node is a
  leaf then $u$ is a single letter word. Moreover, since $2h + k_{i-1} =
  2 + k_{i-1} \geq 2$, we have $u \sieq{2}{i} v$, therefore, $v = u$ and
  $\alpha(u) = \alpha(v)$.

  If the topmost node is a binary node then $u = u_1\cdot u_2$ with
  $u_1,u_2$ admitting $\alpha$-factorization forests of height $h_1,h_2
  \leq h-1$. Using a simple \efgame argument, we get that $v = v_1 \cdot
  v_2$ with $u_1 \sieq{2h+k_{i-1}-1}{i} v_1$ and $u_2
  \sieq{2h+k_{i-1}-1}{i} v_2$. Since $2h +k_{i-1} -1 \geq 2(h-1) +k_{i-1}$,
  we can use our induction hypothesis which yields that $\alpha(u_1)
  \lmo \alpha(v_1)$ and $\alpha(u_2) \lmo \alpha(v_2)$. By combining the
  two we obtain that $\alpha(u) = \alpha(u_1) \cdot \alpha(u_2) \lmo
  \alpha(v_1) \cdot \alpha(v_2) = \alpha(v)$.

  If the topmost node is an idempotent node for some idempotent $e$,
  then $u=u_1\cdot u' \cdot u_2$ such that $\alpha(u_1) = \alpha(u_2) =
  \alpha(u') = e$ and $u_1,u_2$ admit $\alpha$-factorization forests of
  height $h_1,h_2 \leq h-1$. By using a simple \efgame argument we get
  that $v = v_1 \cdot v' \cdot v_2$ such that $u_1
  \sieq{2h+k_{i-1}-2}{i} v_1$, $u' \sieq{2h+k_{i-1}-2}{i} v'$ and $u_2
  \sieq{2h+k_{i-1}-2}{i} v_2$. Applying the induction hypothesis as in
  the previous case, we get that $e = \alpha(u_1) \lmo \alpha(v_1)$ and
  $e = \alpha(u_2) \lmo \alpha(v_2)$. However, we cannot apply induction
  on $u'$ since the height of its $\alpha$-factorization forest has not
  decreased. We use Equation~\eqref{eq:sig} instead. We know that $u'
  \sieq{2h+k_{i-1}-2}{i} v'$, therefore, by choice of $k_i$, we have
  $(\alpha(v'),\alpha(u')) \in \Cs_{i-1}[\alpha]$. Recall that by
  hypothesis of this case, $\alpha(u') = e$. Therefore, by
  Equation~\eqref{eq:sig}, we get that:
  \[
  \alpha(u) = e \lmo e \cdot \alpha(v') \cdot e \lmo \alpha(v_1) \cdot \alpha(v')
  \cdot \alpha(v_2) = \alpha(v)
  \]
  \noindent
  which terminates the proof.\qed
\end{proof}

\section{Analyzing \dChains: \Chain Trees}
\label{app:ctrees}
In this appendix, we define \chain trees. \Chain trees are our main
tool in the proof of the difficult 'if' direction of
Theorem~\ref{thm:caracbc}. The main goal of the notion is to analyze
how \dchains are constructed. In particular we are interested in a
specific property of the set of \dchains that we define now.

\medskip
\noindent
{\bf Alternation.} Let $M$ be a finite monoid. 
We say that a \chain $(s_1,\dots,s_n) \in
M^*$ has \emph{alternation} $\ell$ if there are exactly $\ell$
indices $i$ such that $s_i \neq s_{i+1}$. We say that a set of
\chains $\Ss$ has \emph{bounded alternation} if there exists a bound
$\ell \in \nat$ such that all \chains in $\Ss$ have alternation at
most $\ell$.

We will see in Appendix~\ref{app:bc} that $\Cstwo[\alpha]$ having
bounded alternation is another characterization of $\bswd$. The
difficult direction of Theorem~\ref{thm:caracbc} will then be reduced
to proving that if $\Cstwo[\alpha]$ has \emph{unbounded alternation}
then one of the two equations in the characterization is contradicted.
Therefore, we will need a way to analyze how \dchains with high
alternation are built. In particular, we will need to extract a
property from the set of \dchains that decides which equation is
contradicted. This is what \chain trees are for. Intuitively, a \chain
tree is associated to a single \dchain and represents a computation of
our algorithm (see Section~\ref{sec:comput}) that yields this \dchain.

As we explained in the main paper, one can find connections between
our proof and that of the characterization of boolean combination of
open sets of trees~\cite{bpopen}. In~\cite{bpopen} as well, the
authors consider a notion of  ``\chains'' which corresponds to open
sets of trees and need to analyze how they are built. This is achieved
with an object called ``Strategy Tree''. Though strategy trees and
\chain trees share the same purpose, i.e., analyzing how \chains are
built, there is no connection between the notions themselves since
they deal with completely different objects.

We organize the appendix in three subsections. We first define the 
general notion of \chain trees. In the second subsection, we define
the main tool we use to analyze \chain trees: context values. In
particular, we prove that we can use context values to generate
$B$-schemas. Finally, in the last subsection, we define a strict
subset of \chain trees: the \emph{locally optimal \chain trees} and
prove that it suffices to consider only such  \chain trees (i.e., that
for any \dchain there exists a locally optimal \chain tree that
``computes'' it). 

\subsection{Definition}

Set $\alpha: A^* \rightarrow M$ a morphism into a finite monoid $M$. We
associate to $\alpha$ a set $\ct[\alpha]$ of \emph{\chain trees}. As we
explained, a \chain tree is associated to a single \dchain for $\alpha$ and
represents a way to compute this \dchain using our algorithm. Note that our
algorithm works with sets of compatible sets of \dchains, while \chain trees
are for single \dchains. This difference will be reflected in the definition.
For all $n \in \nat$ we define $\ell_n = \omega(2^{M^n})$.

\medskip
\noindent
{\bf \Chain Trees.} Set $n \in \nat$. A \emph{\chain tree} $T$ \emph{of
  level} $n$ for $\alpha$ is an ordered unranked tree that may have
two types of (unlabeled) inner nodes: product nodes and operation
nodes, and two types of leaves, labeled with a \dchain of length $n$:
initial leaves and operation leaves. Moreover, to each node $x$ in
the tree, we associate an alphabet $\content{x} \subseteq A$ and a
value $\val{x} \in M^n$ by induction on the structure of the tree.

Intuitively, each type of node corresponds to a part of the algorithm that
computes \dchains. Initial leaves correspond to the initial trivial compatible
sets from which the algorithm starts, product nodes correspond to the product~\eqref{eq:mul},
finally operation nodes and leaves can only be used together and correspond to
the application of~\eqref{eq:oper}. We now give a precise definition of each
type of node.

\medskip
\noindent
{\it Initial Leaves.} An initial leaf $x$ is labeled  with a constant
\dchain $(s,\cdots,s) \in \Cstwon[\alpha,B]$ for some $B \subseteq A$.
We set $\content{x} = B$ and $\val{x} = (s,\cdots,s)$.

\medskip
\noindent
{\it Operation Leaves.} An operation leaf $x$ is labeled with an arbitrary
\dchain $\bar{s} \in \Cstwon[\alpha,B]$ for some $B \subseteq A$.
We set $\content{x} = B$ and $\val{x} = \bar{s}$. Note that we
will set constraints on the parents of operation leaves. In
particular, these parents are always operation nodes. We will see
this in details when defining operation nodes.

\medskip
\noindent
{\it Product Nodes.} A product node $x$ is unlabeled. It can have an
arbitrary number of children $x_1,\dots,x_m$ which are all initial
leaves, product nodes or operation nodes. In particular, we set
$\content{x} = \content{x_1} \cup \cdots \cup \content{x_m}$ and
$\val{x} = \val{x_1} \cdots \val{x_m}$.

\medskip
\noindent
{\it Operation Nodes.} An operation node $x$ has exactly $2 \ell_n + 
1$ children sharing the same alphabet $B$. The $(\ell_n + 1)$-th
child, called the \emph{central child} of $x$, has to be an operation 
leaf. The other children, called the \emph{context children} of $x$,
are either operation nodes, product nodes or initial leaves and the
set of their values must be \emph{compatible for $\alpha,B$} (i.e. it
must belong to $\fCtwon[\alpha,B]$). Finally, we set a restriction on
the value of the central child. Since the values of the context
children of $x$ form a compatible set of \dchains, they all share the
same first component, that we call $t$. We require the first component
of the value of the central child to be $t^{\ell_n}$. This means that
the central child is an operation node labeled with
$(t^{\ell_n},s_1,\dots,s_{n-1}) \in \Cstwon[\alpha,B]$. Finally, we
set $\content{x} = B$ and $\val{x} = \val{x_1} \cdots
\val{x_{2\ell_n+1}}$.



\medskip
This terminates the definition of \chain trees. The alphabet and value
of a \chain tree $T$, \content{T} and \val{T}, are the alphabet and
value of its root. We give an example of a \chain tree in
Figure~\ref{fig:ctree}. Moreover, the following fact is immediate by
definition.



\begin{fact} \label{fct:value}
  Let $T$ be a \chain tree and let $x_1,\dots,x_m$ be its leaves listed
  from left to right. Then $\val{T} = \val{x_1} \cdots \val{x_m}$.
\end{fact}

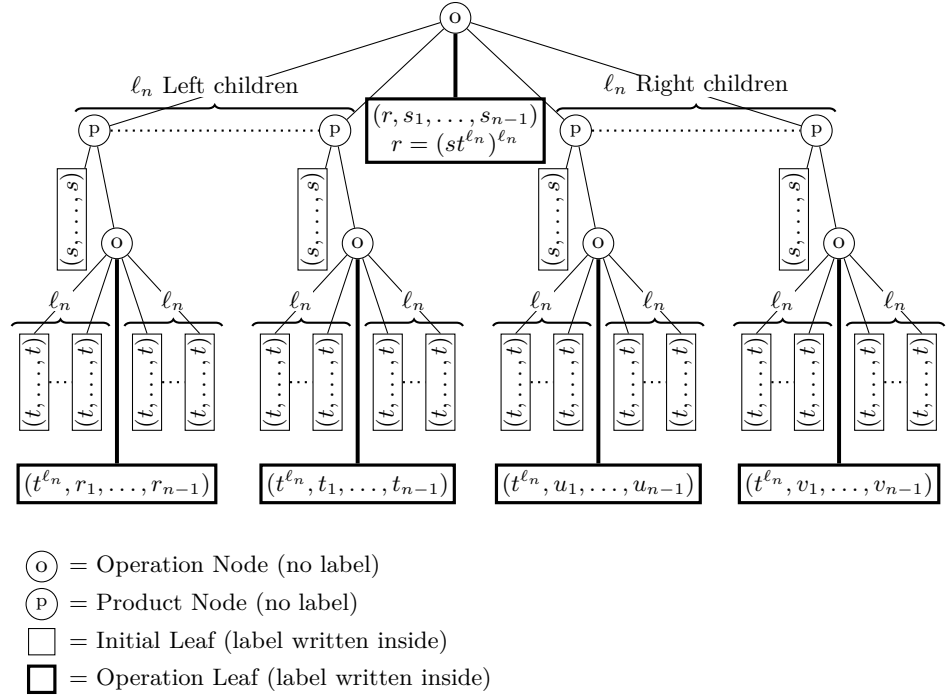
\begin{figure}[h]
  \begin{center}
    \begin{tikzpicture}
      \node[nod] (x) at (0.0,0.0) {o};

      \node[nod] (x1) at (-4.8,-1.5) {\scriptsize p};
      \node[nod] (x2) at (-1.6,-1.5) {\scriptsize p};

      \node[nod] (x3) at (1.6,-1.5) {\scriptsize p};
      \node[nod] (x4) at (4.8,-1.5) {\scriptsize p};

      \draw (x) to (x1);
      \draw (x) to (x2);
      \draw (x) to (x3);
      \draw (x) to (x4);
      \draw[thick,dotted] (x1) to (x2);
      \draw[thick,dotted] (x3) to (x4);

      \node[nor,align=center,very thick,fill=white] (la) at (0.0,-1.5)
      {\small $(r,s_1,\dots,s_{n-1})$\\$r = (st^{\ell_n})^{\ell_n}$};
      \draw[siar,-] (la) to (x);

      \draw[thick,decoration={brace,mirror},decorate] ($(x2.north
      east)+(0.1,0.1)$) to coordinate (l1)
      ($(x1.north west)+(-0.1,0.1)$);

      \node[anchor=south,inner sep=0pt,fill=white] at ($(l1)+(0.0,0.2)$) {$\ell_n$ Left children};

      \draw[thick,decoration={brace,mirror},decorate] ($(x4.north
      east)+(0.1,0.1)$) to coordinate (l2)
      ($(x3.north west)+(-0.1,0.1)$);

      \node[anchor=south,inner sep=0pt,fill=white] at ($(l2)+(0.0,0.2)$) {$\ell_n$ Right children};

      \begin{scope}[xshift=-4.8cm]

        \node[nol,anchor=east] (y1) at (-0.3,-2.0) {\small $(s,\dots,s)$};
        \node[nod] (y2) at (+0.3,-3.0) {o};

        \node[nol,anchor=east] (z1) at (-0.8,-4.2) {\small $(t,\dots,t)$};
        \node[nol,anchor=east] (z2) at (-0.1,-4.2) {\small $(t,\dots,t)$};
        \node[nol,anchor=east] (z3) at (+0.7,-4.2) {\small $(t,\dots,t)$};
        \node[nol,anchor=east] (z4) at (+1.4,-4.2) {\small $(t,\dots,t)$};

        \draw (x1) to (y1);
        \draw (x1) to (y2);

        \draw (y2) to (z1.east);
        \draw (y2) to (z2.east);
        \draw (y2) to (z3.east);
        \draw (y2) to (z4.east);

        \draw[thick,dotted] (z1) to (z2);
        \draw[thick,dotted] (z3) to (z4);

        \draw[thick,decoration={brace,mirror},decorate] ($(z2.south
        east)+(0.1,0.1)$) to coordinate (l1) ($(z1.north east)+(-0.1,0.1)$);

        \node[anchor=south,inner sep=0pt,fill=white] at ($(l1)+(0.0,0.2)$) {$\ell_n$};

        \draw[thick,decoration={brace,mirror},decorate] ($(z4.south
        east)+(0.1,0.1)$) to coordinate (l2)
        ($(z3.north east)+(-0.1,0.1)$);

        \node[anchor=south,inner sep=0pt,fill=white] at ($(l2)+(0.0,0.2)$) {$\ell_n$};

        \node[nor,align=center,very thick] (lb) at (0.3,-6.2) {\small $(t^{\ell_n},r_1,\dots,r_{n-1})$};

        \draw[siar,-] (lb) to (y2);

      \end{scope}

      \begin{scope}[xshift=-1.6cm]

        \node[nol,anchor=east] (y1) at (-0.3,-2.0) {\small $(s,\dots,s)$};
        \node[nod] (y2) at (+0.3,-3.0) {o};

        \node[nol,anchor=east] (z1) at (-0.8,-4.2) {\small $(t,\dots,t)$};
        \node[nol,anchor=east] (z2) at (-0.1,-4.2) {\small $(t,\dots,t)$};
        \node[nol,anchor=east] (z3) at (+0.7,-4.2) {\small $(t,\dots,t)$};
        \node[nol,anchor=east] (z4) at (+1.4,-4.2) {\small $(t,\dots,t)$};

        \draw (x2) to (y1);
        \draw (x2) to (y2);

        \draw (y2) to (z1.east);
        \draw (y2) to (z2.east);
        \draw (y2) to (z3.east);
        \draw (y2) to (z4.east);

        \draw[thick,dotted] (z1) to (z2);
        \draw[thick,dotted] (z3) to (z4);

        \draw[thick,decoration={brace,mirror},decorate] ($(z2.south
        east)+(0.1,0.1)$) to coordinate (l1) ($(z1.north east)+(-0.1,0.1)$);

        \node[anchor=south,inner sep=0pt,fill=white] at ($(l1)+(0.0,0.2)$) {$\ell_n$};

        \draw[thick,decoration={brace,mirror},decorate] ($(z4.south
        east)+(0.1,0.1)$) to coordinate (l2)
        ($(z3.north east)+(-0.1,0.1)$);

        \node[anchor=south,inner sep=0pt,fill=white] at ($(l2)+(0.0,0.2)$)
        {$\ell_n$};

        \node[nor,align=center,very thick] (lb) at (0.3,-6.2) {\small $(t^{\ell_n},t_1,\dots,t_{n-1})$};

        \draw[siar,-] (lb) to (y2);

      \end{scope}

      \begin{scope}[xshift=1.6cm]

        \node[nol,anchor=east] (y1) at (-0.3,-2.0) {\small $(s,\dots,s)$};
        \node[nod] (y2) at (+0.3,-3.0) {o};

        \node[nol,anchor=east] (z1) at (-0.8,-4.2) {\small $(t,\dots,t)$};
        \node[nol,anchor=east] (z2) at (-0.1,-4.2) {\small $(t,\dots,t)$};
        \node[nol,anchor=east] (z3) at (+0.7,-4.2) {\small $(t,\dots,t)$};
        \node[nol,anchor=east] (z4) at (+1.4,-4.2) {\small $(t,\dots,t)$};

        \draw (x3) to (y1);
        \draw (x3) to (y2);

        \draw (y2) to (z1.east);
        \draw (y2) to (z2.east);
        \draw (y2) to (z3.east);
        \draw (y2) to (z4.east);

        \draw[thick,dotted] (z1) to (z2);
        \draw[thick,dotted] (z3) to (z4);

        \draw[thick,decoration={brace,mirror},decorate] ($(z2.south
        east)+(0.1,0.1)$) to coordinate (l1) ($(z1.north east)+(-0.1,0.1)$);

        \node[anchor=south,inner sep=0pt,fill=white] at ($(l1)+(0.0,0.2)$) {$\ell_n$};

        \draw[thick,decoration={brace,mirror},decorate] ($(z4.south
        east)+(0.1,0.1)$) to coordinate (l2)
        ($(z3.north east)+(-0.1,0.1)$);

        \node[anchor=south,inner sep=0pt,fill=white] at ($(l2)+(0.0,0.2)$) {$\ell_n$};

        \node[nor,align=center,very thick] (lb) at (0.3,-6.2) {\small $(t^{\ell_n},u_1,\dots,u_{n-1})$};

        \draw[siar,-] (lb) to (y2);

      \end{scope}

      \begin{scope}[xshift=4.8cm]

        \node[nol,anchor=east] (y1) at (-0.3,-2.0) {\small $(s,\dots,s)$};
        \node[nod] (y2) at (+0.3,-3.0) {o};

        \node[nol,anchor=east] (z1) at (-0.8,-4.2) {\small $(t,\dots,t)$};
        \node[nol,anchor=east] (z2) at (-0.1,-4.2) {\small $(t,\dots,t)$};
        \node[nol,anchor=east] (z3) at (+0.7,-4.2) {\small $(t,\dots,t)$};
        \node[nol,anchor=east] (z4) at (+1.4,-4.2) {\small $(t,\dots,t)$};

        \draw (x4) to (y1);
        \draw (x4) to (y2);

        \draw (y2) to (z1.east);
        \draw (y2) to (z2.east);
        \draw (y2) to (z3.east);
        \draw (y2) to (z4.east);

        \draw[thick,dotted] (z1) to (z2);
        \draw[thick,dotted] (z3) to (z4);

        \draw[thick,decoration={brace,mirror},decorate] ($(z2.south
        east)+(0.1,0.1)$) to coordinate (l1) ($(z1.north east)+(-0.1,0.1)$);

        \node[anchor=south,inner sep=0pt,fill=white] at ($(l1)+(0.0,0.2)$) {$\ell_n$};

        \draw[thick,decoration={brace,mirror},decorate] ($(z4.south
        east)+(0.1,0.1)$) to coordinate (l2)
        ($(z3.north east)+(-0.1,0.1)$);

        \node[anchor=south,inner sep=0pt,fill=white] at ($(l2)+(0.0,0.2)$)
        {$\ell_n$};

        \node[nor,align=center,very thick] (lb) at (0.3,-6.2) {\small $(t^{\ell_n},v_1,\dots,v_{n-1})$};

        \draw[siar,-] (lb) to (y2);

      \end{scope}

      \node[anchor=west] at (-5.25,-7.3) {$=$ Operation Node (no label)};
      \node[anchor=west] at (-5.25,-7.8) {$=$ Product Node (no label)};
      \node[anchor=west] at (-5.25,-8.3) {$=$ Initial Leaf (label written inside)};
      \node[anchor=west] at (-5.25,-8.8) {$=$ Operation Leaf (label written inside)};

      \node[nod] at (-5.5,-7.3) {o};
      \node[nod] at (-5.5,-7.8) {\scriptsize p};
      \node[nol] at (-5.5,-8.3) {};
      \node[nol,very thick] at (-5.5,-8.8) {};
    \end{tikzpicture}
  \end{center}
  \caption{An example of \chain tree of level $n$}
  \label{fig:ctree}
\end{figure}

We denote by $\ct_n[\alpha,B]$ the set of all \chain trees of level $n$ and
alphabet $B$ associated to $\alpha$ and by $\ct[\alpha]$ the set of all \chain
trees associated to $\alpha$. If $\cs$ is a set of \chain trees, we define
$\val{\cs} = \{\val{T} \mid T \in \cs\}$. We now state ``correctness'' and
``completeness'' of \chain trees, i.e., a \chain is a \dchain iff it is the
value of some \chain tree. We prove this as a consequence of the validity of
our algorithm for computing \dchains, stated in
Proposition~\ref{prop:compu}. 


\begin{proposition} \label{prop:ctree}
  $\Cstwon[\alpha,B] = \val{\ct_n[\alpha,B]}$.
\end{proposition}

\begin{proof}
  That $\val{\ct_n[\alpha,B]} \subseteq \Cstwon[\alpha,B]$ is immediate
  by definition and Fact~\ref{fct:chaincomp}. We concentrate on the other
  inclusion. Since Proposition~\ref{prop:compu} deals with sets of
  compatible \dchains rather than just \dchains, we prove a slightly
  stronger result. Two \chain trees are said \emph{compatible} if they
  have the same structure, the same alphabet and differ only by the
  labels of their operation leaves. For all $T \in \ct[\alpha]$, we set
  $id(T) \subseteq \ct[\alpha]$  as the set of all \chain trees that are
  compatible with $T$.

  \begin{lemma} \label{lem:ctrees}
    Let $B \subseteq A$. Then $$Sat^*_n[B](C
    \mapsto \fI_n[C]) ~~\subseteq~~ \downclos \{\val{id(T)} \mid \content{T} = B\}$$
  \end{lemma}

  By Proposition~\ref{prop:compu}, if $\bar{s}$ is a \dchain of length
  $n$ for $\alpha,B$, there exists $\Ss \in Sat^*_n[B](C \mapsto
  \fI_n[C])$ such that $\bar{s} \in \Ss$. Therefore the inclusion
  $\Cstwon[\alpha,B] \subseteq \val{\ct_n[\alpha,B]}$ is an immediate
  consequence of Lemma~\ref{lem:ctrees}. It remains to prove
  Lemma~\ref{lem:ctrees}.

  Let $\Ts \in Sat^*_n[B](C \mapsto \fI_n[C])$. We
  need to construct $T \in \ct[\alpha]$ such that $\Ts \subseteq
  \val{id(T)}$. By definition, $\Ts \in Sat^j_n[B](C \mapsto \fI_n[C])$
  for some $j \in \nat$. We proceed by induction on $j$. Assume first
  that $j=0$. Then $\Ts = \{(s,\dots,s)\} \in \fI_n[B]$. By definition
  this means that $\Ts = \{\val{T}\}$ where $T$ is the \chain tree
  composed of a single initial leaf with label $(s,\dots,s)$ and
  alphabet $B$. Assume now that $j \geq 1$. For all $D \subseteq A$, we
  set $\fT_D = Sat^{j-1}_n[D](C \mapsto \fI_n[C])$. By definition, we
  have $\Ts \in \fT_B \cup \fM_B \cup \fO_B$ with
  \begin{eqnarray*}
    \fM_B & = & \cup_{C \cup D = B} (\fT_C \cdot \fT_D) \\
    \fO_B & = & \{\Ss^\omega \cdot (1_M,\Cstwolen{n-1}[\alpha,B]) \cdot
    \Ss^{\omega} \mid \Ss \in \fT_B\}
  \end{eqnarray*}
  If $\Ts \in \fT_B$, the result is immediate by induction
  hypothesis. Assume now that $\Ts \in \fM_B$. By definition, this means
  that there exist $C,D$ such that $C \cup D = B$ and $\Ts_C,\Ts_D$ in
  $\fT_C,\fT_D$ such that $\Ts = \Ts_C \cdot \Ts_B$. Using our induction
  hypothesis, we get $T_C,T_D$ such that $\Ts_C \subseteq \val{id(T_C)}$
  and $\Ts_D \subseteq  \val{id(T_D)}$. Consider $T$ the \chain tree
  whose topmost node is a product with $T_C,T_D$ as children. It is
  immediate by definition that $\Ts \subseteq \val{id(T_C)} \cdot
  \val{id(T_D)} = \val{id(T)}$.

  It remains to treat the case when $\Ts \in \fO_B$. By definition, we
  get $\Ss \in \fT_B$ such that $\Ts = \Ss^\omega \cdot (1_M,
  \Cstwolen{n-1}[\alpha,B]) \cdot \Ss^{\omega}$. Note that since $\Ss,\Ts
  \in Sat^*_n[B](C \mapsto \fI_n[C])$, by Proposition~\ref{prop:compu},
  $\Ss,\Ts \in \fCtwolen{n}[\alpha,B]$. We denote by $s$ the first element
  common to all \chains in $\Ss$. Note that since $\ell_n =
  \omega(2^{M^n})$, the first element common to all \chains in \Ts is
  $s^{\ell_n}$. Set $\Rs_{s^{\ell_n}}$ as the set of all \dchains of
  length $n$ for $\alpha,B$ that have $s^{\ell_n}$ as first element. By
  definition $\Ts \subseteq \Rs_{s^{\ell_n}}$. Moreover,
  \[
  \Ts = \Ss^{\ell_n} \cdot \Ts \cdot \Ss^{\ell_n} \subseteq \Ss^{\ell_n} \cdot
  \Rs_{s^{\ell_n}} \cdot \Ss^{\ell_n}
  \]
  By induction hypothesis there exists a \chain tree $T_{\Ss}$ of
  alphabet $B$ such that $\Ss \subseteq \val{id(T_{\Ss})}$. Let $T$ be
  the \chain tree whose topmost node is an operation node whose context
  children are all copies of $T_{\Ss}$ and whose central child is the
  operation leaf labeled with some arbitrary chosen \dchain in
  $\Rs_{s^{\ell_n}}$. Observe that by definition, 
  $\val{id(T)} = (\val{id(T_{\Ss})})^{\ell_n} \cdot \Rs_{s^{\ell_n}} \cdot
  (\val{id(T_{\Ss})})^{\ell_n}$. Therefore, since $\Ss \subseteq
  \val{id(T_{\Ss})}$, we have $\Ts \subseteq \val{id(T)}$ which
  terminates the proof. Note that the tree we obtained is particular: all
  subtrees rooted at context children of an operation node are identical. \qed
\end{proof}

\medskip
\noindent
{\bf Alternation and Recursive Alternation of a \Chain Tree.} The
\emph{alternation} of a \chain tree is the alternation of its
value. We say that $\ct[\alpha]$ has \emph{unbounded alternation} if  
the set $\val{\ct[\alpha]}$ has unbounded alternation. Note that by
Proposition~\ref{prop:ctree}, $\Cstwo[\alpha]$ has unbounded
alternation iff $\ct[\alpha]$ has unbounded alternation.

In the proof we will be interested in another property of \chain 
trees: \emph{recursive alternation}. Recursive alternation corresponds
to the maximal alternation of labels of operation leaves in the tree.
More precisely, if $T$ is a \chain tree, its \emph{recursive
  alternation} is the largest natural $j$ such that there exists an
\emph{operation leaf} in $T$ whose label has alternation $j$. An
important idea in the proof will be to separate the case when we can
find a set of \chain trees with unbounded alternation but bounded
recursive alternation from the converse one. However, in order to make
this work, we will need to add one more condition to our
trees. Intuitively, we need to know that if a tree has high recursive
alternation, this is necessary, i.e, the tree cannot be modified into 
a tree that has low alternation while keeping the same value. This is
what we do with local optimality.

\subsection{Locally Optimal \Chain Trees}

For all $n,B$, we define a strict subset of $\ct_n[\alpha,B]$ called 
the set of \emph{Locally Optimal \Chain Trees}. We then prove that
we can assume without loss of generality that all \chain trees we
consider are locally optimal.

We first define local optimality as a property of a single node $x$ 
in a \chain tree $T$. We will generalize the notion to a whole tree
by saying that it is locally optimal if and only if all its
\emph{operation leaves} are locally optimal. Given a node $x$, local
optimality of $x$ depends on two parameters of $x$: its value
$\val{x}$ and a new parameter called its \emph{context value},
$\cval{x}$, that we define now. 

\medskip
\noindent
{\bf Context Value of a Node.} Let $T$ be a \chain tree of level $n$
and let $x_1,\dots,x_m$ be the leaves of $T$ sorted in prefix
order. Recall that by Fact~\ref{fct:value}, $\val{T} = \val{x_1}
\cdots \val{x_m}$. To every node~$x$ of $T$, we associate a
pair $\cval{x} \in (\Cstwon[\alpha])^2$ called the \emph{context 
  value} of $x$. Set $x_{i},\dots,x_j$ as the leaves of the subtree
rooted at $x$ (in prefix order). We set $\cval{x} = (\val{x_1}
\cdots \val{x_{i-1}},\val{x_{j+1}} \cdots \val{x_m})$. Note that since
for all $i$, $\val{x_i} \in \Cstwon[\alpha]$, $\cval{x}$ is indeed a
pair in $(\Cstwon[\alpha])^2$. By definition and Fact~\ref{fct:value}
one can verify the two following facts:


\begin{fact} \label{fct:value2}
  Let $x$ be a node in a \chain tree $T$ and set $\cval{x}
  =(\bar{s},\bar{s}')$. Then $\val{T} = \bar{s} \cdot \val{x}
  \cdot \bar{s}'$.
\end{fact}
\begin{fact} \label{fct:value3}
  Let $x$ be an inner node in a \chain tree $T$ and set $\cval{x} =
  (\bar{s},\bar{s}')$.  Set $z_1,\dots,z_k$ as the children of $x$ with
  context values $\cval{z_i} = (\bar{q_i},\bar{q_i}')$. Then, for all
  $i$, $\bar{q_i} =  \bar{s} \cdot \val{z_1} \cdots \val{z_{i-1}}$ and
  $\bar{q_i}' = \val{z_{i+1}} \cdots \val{x_{k}} \cdot \bar{s}'$.
\end{fact}

In many cases, we will work with context values that are constant,
i.e. $\cval{x} = ((s,\dots,s),(s',\dots,s'))$. In these cases, if
$(t_1,\dots,t_n)$ is a \chain, it will be convenient to simply write
$s \cdot (t_1,\dots,t_n) \cdot s'$ for $(s,\dots,s) \cdot (t_1,
\dots,t_n) \cdot (s',\dots,s')$.

\medskip
\noindent
{\bf Local Optimality.} Set $(s,s') \in M^2$ and $T$ a \chain
tree. Let $x$ be any node in $T$, $(t_1,\cdots,t_{n}) = \val{x}$
and $((s_1,\cdots,s_{n}),(s'_1,\cdots,s'_{n})) = \cval{x}$. We say
that $x$ is \emph{locally optimal for $(s,s')$} if for all $i < n$
such that $t_i \neq t_{i+1}$ the following condition holds:
\[
s \cdot s_{i+1} \cdot t_{i} \cdot s'_{i+1} \cdot s' \neq s \cdot
s_{i+1} \cdot t_{i+1} \cdot s'_{i+1} \cdot s'
\]
Intuitively this means that for all $i$, changing $t_i$ to $t_{i+1}$
in the value of $x$ is necessary to get alternation at position $i$
in the value of the tree (see Fact~\ref{fct:value2}). We say that a
\chain tree $T$ is \emph{locally optimal for} $(s,s')$ if all its
{\bf operation leaves} are locally optimal for $(s,s')$. We say that
$T$ is locally optimal iff it is locally optimal for $(1_M,1_M)$.
This means that locally optimality of a \chain tree only depends on
the context values and labels of operation leaves in the tree. The
following fact is immediate from the definitions:

\begin{fact} \label{fct:optimal}
  Let $(s,s')\in M^2$. 
  Assume that $T$ is locally optimal for
  $(s,s')$. Then $T$ is locally optimal (i.e. locally optimal for
  $(1_M,1_M)$). 
\end{fact}

We finish with our main proposition, which states that for any \chain tree,
there exists a locally optimal one with the same value. In particular,
this means that we will always be able to assume that our \chain trees
are locally optimal.

\begin{proposition} \label{prop:optimal}
  Let $T \in \ct_n[\alpha,B]$ and $(s,s') \in M^2$. There exists $T' \in
  \ct_n[\alpha,B]$ which is locally optimal for $(s,s')$ and such that
  $s \cdot \val{T} \cdot s' = s \cdot \val{T'} \cdot s'$.
\end{proposition}

\begin{proof}
  Set $T \in \ct_n[\alpha,B]$, we explain how to construct $T'$. For all
  $i < n$, we define the $i$-alternation of $T$ as the number of
  operation leaves $x$ in $T$ such that $\val{x} = (t_1,\cdots,t_{n})$
  with $t_i \neq t_{i+1}$. Finally, we define the \emph{index of $T$} as
  the sequence of its $i$-alternations ordered with increasing $i$.

  We can now describe the construction. Assume that $T$ is not locally
  optimal for $(s,s')$. We explain how to construct a second \chain tree
  $T'$ such that
  \begin{enumerate}
  \item $s \cdot \val{T} \cdot s' = s \cdot \val{T'} \cdot s'$.
  \item $T'$ has strictly smaller index than $T$.
  \end{enumerate}
  It then suffices to apply this operation recursively to $T$ until we
  get the desired tree. We now explain the construction. Since $T$ is
  not locally optimal for $(s,s')$, there exists an operation leaf $x$
  of $T$ that is not locally optimal for $(s,s')$. Let
  $(t_1,\dots,t_{n}) = \val{x}$ and
  $((s_1,\dots,s_{n}),(s'_1,\dots,s'_{n})) = \cval{x}$. By choice of
  $x$, there exists $i < n$ such that $t_i \neq t_{i+1}$ and
  $ss_{i+1}t_{i}s'_{i+1}s' = ss_{i+1}t_{i+1}s'_{i+1}s'$. We set $T'$ as 
  the \chain tree obtained from $T$ by replacing the label of $x$ with
  $(t_1,\dots,t_i,t_i,t_{i+2},\dots,t_n)$. By choice of $i$ and
  Fact~\ref{fct:value2}, it is immediate that $s \cdot \val{T} \cdot s' =
  s \cdot \val{T'} \cdot s'$. Moreover, for any $j < i$, $T,T'$ have the
  same $j$-alternation and $T'$ has by definition strictly smaller
  $i$-alternation than $T$. It follows that $T'$ has strictly smaller
  index than $T$ which terminates the proof. \qed
\end{proof} 

\section{Proof of Theorem~\ref{thm:caracbc}: Characterization of \bswd}
\label{app:bc}
This appendix is devoted to the proof of Theorem~\ref{thm:caracbc},
i.e., the decidable characterization of \bswd. We actually prove a
more general theorem that includes a second characterization in terms
of alternation of $\Cstwo[\alpha]$, which will be needed as an
intermediary step when proving the difficult 'if' direction of
Theorem~\ref{thm:caracbc}.

\begin{theorem} \label{thm:caracbc2}
  Let $L$ be a regular language and let $\alpha: A^* \rightarrow M$ be its
  syntactic morphism. The three following properties are equivalent:

  \begin{enumerate}
  \item $L$ is definable in \bswd.
  \item $\Cstwo[\alpha]$ has bounded alternation.
  \item $M$ satisfies the following equations:

    \[
    \begin{array}{rcl}
      s_1^{\omega}s_3^{\omega} & = & s_1^{\omega}s_2s_3^{\omega} \\
      s_3^{\omega}s_1^{\omega} & = & s_3^{\omega}s_2s_1^{\omega}
    \end{array} \quad \text{for $(s_1,s_2,s_3) \in \Cstwo[\alpha]$} \tag{\ref{eq:bcs1}}
    \]

    \[
    \begin{array}{c}
      (s_2t_2)^{\omega}s_1(t'_2s'_2)^{\omega} = (s_2t_2)^{\omega}s_2t_1s'_2(t'_2s'_2)^{\omega} \\
      \text{for $(s_1,s_2,s'_2)$ and $(t_1,t_2,t'_2)$ $B$-schemas for some
        $B \subseteq A$}
    \end{array}\tag{\ref{eq:bcs2}}
    \]
  \end{enumerate}
\end{theorem}

Observe that Theorem~\ref{thm:caracbc} is exactly the equivalence
between Items~1 and~3 in Theorem~\ref{thm:caracbc2}. Therefore it
suffices to prove Theorem~\ref{thm:caracbc2}. Intuitively, Item~2
seems harder to decide than Item~3, since it requires computing a
description of the whole set $\Cstwo[\alpha]$ rather than just the
\dchains and sets of compatible \dchains of length $2$ and
$3$. However, it will serve as a convenient intermediary for proving
Item~3.

We now turn to the proof of Theorem~\ref{thm:caracbc2}. We prove that
1~$\Rightarrow$~3~$\Rightarrow$~2~$\Rightarrow$~1. In this appendix,
we give full proofs for the two ''easy'' directions: 1~$\Rightarrow$~3 
and 2~$\Rightarrow$~1. For the direction 3~$\Rightarrow$~2, we use
\chain trees to reduce the proof to two propositions. We then give
each proposition its own Appendix: Appendix~\ref{app:depth} and
Appendix~\ref{app:width}.

\subsection{1~$\Rightarrow$~3}
\label{sec:1-rightarrow-3-1}

We prove the direction 1~$\Rightarrow$~3 in Theorem~\ref{thm:caracbc2}
which is stated in the following lemma.

\begin{lemma} \label{lem:bcnec}
  Let $L$ be a regular language and $\alpha$ be its syntactic
  morphism. Assume that $L$ is definable in \bswd, then $\alpha$
  satisfies~\eqref{eq:bcs1} and~\eqref{eq:bcs2} .
\end{lemma}

The remainder of this subsection is devoted to proving
Lemma~\ref{lem:bcnec}. The proof is an \efgame argument. We begin by
defining the equivalence associated to \bswd. For any $k \in \nat$, we
write $w \kbceq{2} w'$ iff $w$ and $w'$ satisfy the same \bswd
formulas of quantifier rank~$k$. Therefore, a language if definable by
a $\bswd$ formula of rank~$k$ iff it is saturated by $\kbceq{2}$. One
can verify that $\kbceq{2}$ is an equivalence and that $w \kbceq{2}
w'$ iff $w \ksieq{2} w'$ and $w' \ksieq{2} w$.

We can now prove the lemma. By hypothesis there exists some $\bswd$
formula $\varphi$ that defines $L$, we set $k$ as the quantifier rank  
of this formula. 

\medskip
\noindent
{\bf Proving Equation~\eqref{eq:bcs1}.} Set $(s_1,s_2,s_3) \in
\Cstwo[\alpha]$, we prove that $s_1^{\omega}s_3^{\omega} =
s_1^{\omega}s_2s_3^{\omega}$ (the dual case is proved in the same
way). We prove that there exist $w_1,w_2,w_3 \in A^*$ such that
$\alpha(w_1) = s_1$, $\alpha(w_2)=s_2$, $\alpha(w_3)=s_3$ and for all
pair of words $u,v \in A^*$: 
\begin{equation}
  uw_1^{2^k\omega}w_3^{2^k\omega}v ~~~\kbceq{2}~~~
  uw_1^{2^k\omega}w_2w_3^{2^k\omega}v \label{eq:bcnec}
\end{equation}
Set $N=2^k\omega$. By definition of \kbceq{2}, \eqref{eq:bcnec} means that 
$u (w_1^{N}w_3^{N}) v$ and
$u (w_1^{N}w_2w_3^{N}) v$ cannot be distinguished by a \bswd
formula of quantifier rank $k$. Hence, by definition of $k$, we get
\[
u(w_1^{N}w_3^{N})v \in L \text{~~~iff~~~}
u(w_1^{N}w_2w_3^{N})v \in L
\]
Therefore, by definition of $w_1,w_2,w_3$, of $N$, and of the syntactic monoid
this will prove that $s_1^{\omega}s_3^{\omega} = s_1^{\omega}s_2s_3^{\omega}$.

Since $(s_1,s_2,s_3) \in \Cstwo[\alpha]$ by assumption, there
exist $w_1,w_2,w_3$ such that $w_1 \ksieq{2} w_2 \ksieq{2} w_3$ and
$\alpha(w_1) = s_1$, $\alpha(w_2)=s_2$, $\alpha(w_3)=s_3$. Set $u,v
\in A^*$. We need to prove that
\begin{align}
  \label{eq:1}
  u(w_1^{N}w_3^{N})v & ~~~\ksieq{2}~~~ u(w_1^{N}w_2w_3^{N})v\\
  \label{eq:2}
  u(w_1^{N}w_2w_3^{N})v & ~~~\ksieq{2}~~~ u(w_1^{N}w_3^{N})v
\end{align}
By definition of $w_1,w_2$, we have $w_1
\ksieq{2} w_2$. By Lemma~\ref{lem:aperiodic}, we obtain
$w_1^{N-1} \ksieq{2} w_1^{N}$. Therefore, using 
Lemma~\ref{lem:efconcat} we first get $w_1^N\ksieq{2}w_1^{N}w_2$, and
then that~\eqref{eq:1} holds.

The proof of $\eqref{eq:2}$ is similar: by definition, we have $w_2 \ksieq{2} w_3$, and by
Lemma~\ref{lem:aperiodic} we get $w_3^{N} \ksieq{2} w_3^{N-1}$. Using
Lemma~\ref{lem:efconcat} again, we conclude that $w_2w_3^{N}\ksieq{2}w_3^{N}$, and then
that~$\eqref{eq:2}$ holds.

\medskip
\noindent
{\bf Proving Equation~\eqref{eq:bcs2}.} It remains to prove that
$\alpha$ satisfies Equation~\eqref{eq:bcs2}. We begin with a lemma
on $B$-schemas.

\begin{lemma} \label{lem:schemprop}
  Assume that $(s_1,s_2,s'_2)$ is a $B$-schema. Then for all $k \in
  \nat$ there exist $w_1,w_2,w'_2 \in A^*$ such that:
  \begin{itemize}
  \item $\content{w_1} = \content{w_2} = \content{w'_2} = B$.
  \item $\alpha(w_1) = s_1, \alpha(w_2)=s_2$ and $\alpha(w'_2) = s'_2$.
  \item for all $u \in B^*$, $w_1 \ksieq{2} w_2uw'_2$.
  \end{itemize}
\end{lemma}

\begin{proof}
  This is proved using Lemma~\ref{lem:siprop}. Fix a $B$-schema
  $(s_1,s_2,s'_2)$ and $k \in \nat$. By definition, there exist $\Ts \in
  \fCtwo[\alpha,B]$ and $r_1,r'_1 \in M$ satisfying $s_1 = r_1r'_1$,
  $(r_1,s_2) = (t_1,t_2) \cdot (q,q_2)$ and $(r_1,s'_2) = (q,q'_2) \cdot
  (t'_1,t'_2)$ with $(t_1,t_2),(t'_1,t'_2) \in \Cstwo[\alpha,B]$ and
  $(q,q_2),(q,q'_2) \in \Ts^\omega = \Ts^{2^{2k}\omega}$. By definition
  of \dchains, we obtain words $v_1,v,v'_1,w_2,w'_2 \in A^*$ satisfying
  the following properties:
  \begin{enumerate}[label=$\alph*)$]
  \item $\content{v_1} = \content{v} = \content{v'_1} =\content{w_2}
    =\content{w'_2} = B$
  \item $\alpha(v_1) = t_1$,
    $\alpha(v'_1) = t'_1$, $\alpha(w_2) = t_2q_2$, $\alpha(w'_2) =
    q'_2t'_2$ and $\alpha(v^{2^{2k}\omega}) = q$.
  \item $v_1v^{2^{2k}\omega}
    \ksieq{2} w_2$ and $v^{2^{2k}\omega}v'_1 \ksieq{2} w'_2$.
  \end{enumerate}

  Set $w_1 = v_1v^{2^{2k}\omega}v^{2^{2k}\omega}v'_1$ and observe that
  by item $a)$, $\content{w_1} = \content{w_2} = \content{w_2} =
  B$. Moreover, by item $b)$, $\alpha(w_1) = t_1qqt'_1 = r_1r'_1= s_1$,
  $\alpha(w_2) = t_2q_2 = s_2$ and $\alpha(w'_2) = q'_2t'_2 =
  s'_2$. Finally, it is immediate using \efgame games that for any
  word $u \in B^*$, $u \ksieq{1} v^{2^k\omega}$. Therefore it follows
  from Lemma~\ref{lem:siprop} that $w_1 \ksieq{2}
  v_1v^{2^{2k}\omega}uv^{2^{2k}\omega}v'_1$. Using item $c)$, we then
  conclude that $w_1 \ksieq{2} w_2uw'_2$. \qed  
\end{proof}

We can now use Lemma~\ref{lem:schemprop} to prove that $\alpha$
satisfies Equation~\eqref{eq:bcs2}. Let $(s_1,s_2,s'_2)$ and
$(t_1,t_2,t'_2)$ be $B$-schemas. Let $w_1,w_2,w'_2 \in A^*$ of images
$s_1,s_2,s'_2$ and $v_1,v_2,v'_2 \in A^*$ of images $t_1,t_2,t'_2$
satisfying the conditions of Lemma~\ref{lem:schemprop}. We prove that
for any $u,v \in A^*$:
\begin{equation}
  u[(v_2w_2)^{N}v_1(w'_2v'_2)^{N}]v ~~~\kbceq{2}~~~
  u[(v_2w_2)^{N}v_2w_1v'_2(w'_2v'_2)^{N}]v  \label{eq:bcnec2}
\end{equation}
where again $N=2^k\omega$. By definition of the syntactic monoid and since $L$ is defined by a 
\bswd formula of rank $k$, Equation~\eqref{eq:bcs2} will
follow. Observe that the words $v_1,v_2,v'_2$ and $w_1,w_2,w'_2$ given by
Lemma~\ref{lem:schemprop} satisfy
\begin{align}
  \label{eq:3}
  v_1 &\ksieq{2} v_2w_1v'_2,\\
  \label{eq:4}
  w_1 &\ksieq{2} w_2v_1w'_2.
\end{align}
Using Lemma~\ref{lem:efconcat}, we may multiply \eqref{eq:3} by
$u(v_2w_2)^{N}$ on the left and by $(w'_2v'_2)^{N}v$ on the right:
\[
u(v_2w_2)^{N}v_1(w'_2v'_2)^{N}v ~~~\ksieq{2}~~~
u(v_2w_2)^{N}v_2w_1v'_2(w'_2v'_2)^{N}v.
\]
For the converse direction, from Lemma~\ref{lem:aperiodic}, we have 
$(v_2w_2)^{N} \ksieq{2} (v_2w_2)^{N-1}$
and $(w'_2v'_2)^{N} \ksieq{2}
(w'_2v'_2)^{N-1}$. Using   \eqref{eq:4} and Lemma~\ref{lem:efconcat} again, we
conclude that:
\[u(v_2w_2)^{N}v_2w_1v'_2(w'_2v'_2)^{N}v ~~~\ksieq{2}~~~
u(v_2w_2)^{N-1}v_2(w_2v_1w'_2)v'_2(w'_2v'_2)^{N-1}v\]
i.e.,
\[u(v_2w_2)^{N}v_2w_1v'_2(w'_2v'_2)^{N}v~~~\ksieq{2}~~~u(v_2w_2)^{N}v_1(w'_2v'_2)^{N}v.\]

\subsection{2~$\Rightarrow$~1}

We prove the direction 2~$\Rightarrow$~1 in Theorem~\ref{thm:caracbc2}
which is stated in the following lemma. 

\begin{lemma} \label{lem:bcinter}
  Let $L$ be a regular language and $\alpha$ its syntactic
  morphism. Assume that $\Cstwo[\alpha]$ has bounded alternation, then
  $L$ is definable in \bswd.
\end{lemma}

\begin{proof}
  Assume that $\Cstwo[\alpha]$ has bounded alternation. We prove that
  there exists $k \in \nat$ such that for all $w,w' \in A^*$, $w
  \kbceq{2} w' \Rightarrow \alpha(w) = \alpha(w')$. This proves that $L$
  is saturated with $\kbceq{2}$ and hence definable by a \bsw{2} formula
  of quantifier rank $k$.

  We proceed by contradiction. Assume that for all $k \in \nat$ there
  exists $w_k,w_k' \in A^*$ such that $w_k \kbceq{2} w_k'$ and
  $\alpha(w_k) \neq \alpha(w_k')$. Notice that since there are only
  finitely many pairs in $M^2$, there must exist a pair $(s,s') \in M^2$
  such that $s \neq s'$ and there exists arbitrarily large naturals $k$
  such that $\alpha(w_k)=s$ and $\alpha(w'_k) = s'$. We prove that
  $(s,s')^* \subseteq \Cstwo[\alpha]$ which contradicts that
  $\Cstwo[\alpha]$ has unbounded alternation (recall that $s \neq
  s'$). By definition for all $k \in \nat$ there exists $\ell \geq k$
  such that $\alpha(w_{\ell})=s$ and $\alpha(w'_{\ell}) = s'$, since $\ell
  \geq k$ and by definition of $\kbceq{2}$ this means that :
  \[
  w_{\ell} \ksieq{2} w'_{\ell} \ksieq{2} w_{\ell} \ksieq{2} w'_{\ell}
  \ksieq{2} w_{\ell} \ksieq{2} w'_{\ell} \ksieq{2} \cdots
  \]
  Hence for all $k,j$, $(s,s')^j \in \Cslevk2k[\alpha]$ and therefore,
  for all $j$  $(s,s')^j \in \Cstwo[\alpha]$ which terminates the proof.
  \qed
\end{proof}

\subsection{3~$\Rightarrow$~2}

This is the most difficult direction of Theorem~\ref{thm:caracbc2}. We
state it in the following proposition.

\begin{proposition} \label{prop:bcsuff}
  Let $L$ be a regular language, $\alpha : A^* \rightarrow M$ be its
  syntactic morphism. Assume that $\alpha$ satisfies~\eqref{eq:bcs1}
  and~\eqref{eq:bcs2}, then $\Cstwo[\alpha]$ has bounded alternation.
\end{proposition}

For the remaining of the section, we assume that $L,M$ and $\alpha$
are fixed as in the statement of the proposition. We prove the
contrapositive of Proposition~\ref{prop:bcsuff}: if $\Cstwo[\alpha]$
has unbounded alternation, then either Equation~\eqref{eq:bcs1} or
Equation~\eqref{eq:bcs2} must be contradicted. We use \chain trees to
separate this property into two properties that we will prove in
Appendix~\ref{app:depth} and Appendix~\ref{app:width}. Consider the
two following propositions 

\begin{proposition} \label{prop:width}
  Assume that there exists a set of locally optimal \chain trees $\cs
  \subseteq \ct[\alpha]$ with unbounded alternation but bounded
  recursive alternation. Then $\alpha$ does not satisfy
  Equation~\eqref{eq:bcs1}.
\end{proposition}

\begin{proposition} \label{prop:depth}
  Assume that there exists a set of locally optimal \chain
  trees $\cs \subseteq \ct[\alpha]$ with unbounded alternation and that
  all such sets have unbounded recursive alternation. Then $\alpha$ does
  not satisfy Equation~\eqref{eq:bcs2}.
\end{proposition}

Proposition~\ref{prop:width} and Proposition~\ref{prop:depth} are
proven in Appendix~\ref{app:width} and Appendix~\ref{app:depth}. We
finish this appendix by using them to conclude the proof of
Proposition~\ref{prop:bcsuff}.

If $\Cs[\alpha]$ has unbounded alternation. By
Proposition~\ref{prop:optimal}, we know that there exists a set of
locally optimal \chain trees $\cs \subseteq \ct[\alpha]$ with
unbounded alternation. If $\cs$ can be chosen with bounded recursive
alternation, there is a contradiction to Equation~\eqref{eq:bcs1} by
Proposition~\ref{prop:width}. Otherwise there is a contradiction to
Equation~\eqref{eq:bcs2} by Proposition~\ref{prop:depth} which
terminates the proof of Proposition~\ref{prop:bcsuff}.

\section{Proof of Proposition~\ref{prop:depth}}
\label{app:depth}
Recall that we fixed a morphism $\alpha: A^* \rightarrow M$ into a
finite monoid $M$. We prove Proposition~\ref{prop:depth}.

\adjustc{prop:depth}
\begin{proposition}
  Assume that there exists a set of locally optimal \chain
  trees $\cs \subseteq \ct[\alpha]$ with unbounded alternation and that
  all such sets have unbounded recursive alternation. Then $\alpha$ does
  not satisfy Equation~\eqref{eq:bcs2}.
\end{proposition}
\restorec

We define a new object that is specific to this case: the \emph{\Chain
  Graph}. The \chain graph describes a construction process for a subset
of the set of \dchains for $\alpha$. While this subset is potentially
strict, we will prove that under the hypothesis of
Proposition~\ref{prop:depth}, it is sufficient to derive a
contradiction to Equation~\eqref{eq:bcs2}.

\medskip
\noindent
{\bf \Chain Graph}. We define a graph $G[\alpha]=(V,E)$ whose edges
are labeled by subsets of the alphabet $A$. We call $G[\alpha]$ the
\emph{\chain graph} of $\alpha$. The set $V$ of nodes of $G[\alpha]$
is the set $V = M^2 \times M$. Let $((s,s'),u)$ and $((t,t'),v)$ be
nodes of $G[\alpha]$ and $B \subseteq A$, then $E$ contains an edge
labeled by $B$ from $((s,s'),u)$ to $((t,t'),v)$ iff there exists a
$B$-schema $(s_1,s_2,s'_2) \in M^3$ such that:

\begin{itemize}
\item $s \cdot s_1 \cdot s' = u$.
\item $s \cdot s_2= t$ and $s'_2 \cdot s' = t'$.
\end{itemize} 

Observe that the definition does not depend on $v$. We say that
$G[\alpha]$ is \emph{recursive} if it contains a cycle such that
\begin{itemize}
\item [$a)$] all edges in the cycle are labeled by the same alphabet $B
  \subseteq A$,
\item [$b)$] the cycle contains two nodes $((s,s'),u)$, $((t,t'),v)$
  such that $u \neq v$.
\end{itemize}
We now prove
Proposition~\ref{prop:depth} as a consequence of the two following 
propositions.

\begin{proposition} \label{prop:graphcont1}
  Assume that $G[\alpha]$ is recursive. Then $\alpha$ does not
  satisfy~\eqref{eq:bcs2}.
\end{proposition}

\begin{proposition} \label{prop:graphcont2}
  Assume that there exists a set of locally optimal \chain
  trees $\cs \subseteq \ct[\alpha]$ with unbounded alternation and that
  all such sets have unbounded recursive alternation. Then $G[\alpha]$
  is recursive.
\end{proposition}

Observe that Proposition~\ref{prop:depth} is an immediate consequence of
Propositions~\ref{prop:graphcont1} and~\ref{prop:graphcont2}. Before proving
them, note that the notion of \chain graph is inspired from the notion of
strategy graph in~\cite{bpopen}. This is because both notions are designed to
derive contradiction to similar equations. However, our proof remains fairly
different from the one of~\cite{bpopen}. The reason for this is that the main
difficulty here is proving Proposition~\ref{prop:graphcont2}, i.e., going from
\chain trees (which are unique to our setting) to a recursive \chain graph. On
the contrary, the much simpler proof of Proposition~\ref{prop:graphcont1} is
similar to the corresponding one in~\cite{bpopen}.

\subsection{Proof of Proposition~\ref{prop:graphcont1}}

\adjustc{prop:graphcont1}
\begin{proposition}
  Assume that $G[\alpha]$ is recursive then $\alpha$ does not satisfy~\eqref{eq:bcs2}.
\end{proposition}
\restorec

Assume that $G[\alpha]$ is recursive. By definition, we get $B
\subseteq A$, a cycle whose edges are all labeled with $B$ and two
consecutive nodes $((s,s'),u)$ and $((t,t'),v)$ in this cycle such
that $u \neq v$. Since there exists an edge $((s,s'),u)
\xrightarrow{B} ((t,t'),v)$, we obtain a $B$-schema $(s_1,s_2,s'_2)$
such that
\begin{align*}
  u&=s \cdot s_1 \cdot s',\\
  t&=s \cdot s_2, \\
  t'&=s'_2 \cdot s'.
\end{align*}
Moreover, one can verify that since $((s,s'),u)$ and
$((t,t'),v)$ are in the same cycle with all edges labeled by $B$, there
exists another $B$-schema $(t_1,t_2,t'_2)$ and $w,w' \in B^*$ such
that
\begin{align*}
  v&=  t \cdot t_1 \cdot t',\\
  s&=t \cdot t_2 \cdot \alpha(w),\\
  s'&=\alpha(w') \cdot t'_2 \cdot t'.
\end{align*}
By combining all these
definitions we get: 
\begin{eqnarray*}
  u & = & s(s_2t_2\alpha(w))^{\omega+1} s_1 (\alpha(w')t'_2s'_2)^{\omega+1} s' \\
  v & = & s(s_2t_2\alpha(w))^{\omega+1} s_2t_1 s'_2(\alpha(w')t'_2s'_2)^{\omega+1} s'
\end{eqnarray*}
Set $r_1 = \alpha(w) s_1 \alpha(w')$, $r_2 = \alpha(w)s_2$ and $r'_2 =
s'_2\alpha(w')$. One can verify that since $w,w' \in B^*$ and
$(s_1,s_2,s'_2)$ is a $B$-schema, $(r_1,r_2,r'_2)$ is a $B$-schema as
well. Moreover, by reformulating the equalities above we get:
\begin{eqnarray*}
  u & = & ss_2t_2(r_2t_2)^\omega r_1 (t'_2r'_2)^\omega t'_2s'_2 s' \\
  v & = & ss_2t_2(r_2t_2)^\omega r_2t_1 r'_2(t'_2r'_2)^\omega t'_2s'_2s'
\end{eqnarray*}
Therefore, Equation~\eqref{eq:bcs2} would require $u = v$. Since $u
\neq v$ by hypothesis, $\alpha$ does not satisfy~\eqref{eq:bcs2} and we 
are finished.

\subsection{Proof of Proposition~\ref{prop:graphcont2}}

\adjustc{prop:graphcont2}
\begin{proposition}
  Assume that there exists a set of locally optimal \chain
  trees $\cs \subseteq \ct[\alpha]$ with unbounded alternation and that
  all such sets have unbounded recursive alternation. Then $G[\alpha]$
  is recursive.
\end{proposition}
\restorec

In the remainder of the section, we assume that $\alpha$ satisfies the
hypothesis of Proposition~\ref{prop:graphcont2}. Set $B \subseteq A$
and let $((s,s'),u)$ be a node of $G[\alpha]$, we say that
$((s,s'),u)$ is \emph{$B$-alternating} if for all $n$, there exists 
$(s_1,\dots,s_n) \in \Cstwon[\alpha,B]$ such that the \chain
$(ss_1s',\dots,ss_ns')$ has alternation $n-1$ and $ss_1s' = u$. 

\begin{lemma} \label{lem:alt1}
  $G[\alpha]$ contains at least one $B$-alternating node for some $B$.
\end{lemma}

\begin{proof}
  This is because $\cs$ has unbounded alternation. It follows that there
  exists a least one $u \in M$ such that there are \dchains with
  arbitrary high alternation and $u$ as first element. By definition,
  the node $((1_M,1_M),u)$ is then $B$-alternating for some $B$. \qed
\end{proof}

For the remainder of the proof we define $B$ as a minimal alphabet
such that there exists a $B$-alternating node in $G[\alpha]$. By this
we mean that for any $C \subsetneq B$, there exists no
$C$-alternating node in $G[\alpha]$.

\begin{lemma} \label{lem:alt2}
  Let $((s,s'),u)$ be any $B$-alternating node of $G[\alpha]$. Then
  there exists a node $((t,t'),v)$ such that

  \begin{enumerate}
  \item $((t,t'),v)$ is $B$-alternating.
  \item $((s,s'),u) \xrightarrow{B} ((t,t'),v)$.
  \item $u \neq v$.
  \end{enumerate}
\end{lemma}

By definition $G[\alpha]$ has finitely many nodes. Therefore, since by
definition, there exists at least one $B$-alternating node, it is
immediate from Lemma~\ref{lem:alt2} that $G[\alpha]$ must contain a
cycle whose edges are all labeled by $B$. Moreover, by Item~3 in
Lemma~\ref{lem:alt2}, this cycle contains two nodes $((s,s'),u)$ and
$((t,t'),v)$ such that $u \neq v$. We conclude that $G[\alpha]$ is
recursive which terminates the proof of
Proposition~\ref{prop:graphcont2}. It remains to prove Lemma~\ref{lem:alt2}.  

\begin{proof} We proceed in three steps. We first use our hypothesis
  to construct a special set of \chain trees \crr of alphabet $B$.
  Then, we choose a \chain tree $T$ in \crr with large enough recursive
  alternation. Finally, we use $T$ to construct the desired node
  $((t,t'),v)$. We begin with the construction of \crr.

  \medskip
  \noindent
  {\bf Construction of \crr.} We construct a set \crr of \chain trees
  that satisfies the following properties: 
  \begin{enumerate}
  \item For all $ T \in \crr$, $\content{T} = B$.
  \item All \chains in $s \cdot val(\crr) \cdot s'$ have $u$ as first
    element.
  \item All trees in $\crr$ are locally optimal for $(s,s')$.
  \item $\crr$ has unbounded recursive alternation.
  \end{enumerate}
  We use the fact that $((s,s'),u)$ is $B$-alternating and the
  hypothesis in Proposition~\ref{prop:graphcont2}. Since $((s,s'),u)$ is
  $B$-alternating, we know that for any $n \in \nat$, there exists
  $(s_1,\dots,s_n) \in \Cstwo[\alpha,B]$ such that the \chain
  $(ss_1s',\dots,ss_ns')$ has alternation $n-1$ and $ss_1s' = u$. We 
  denote by $\Rs \subseteq \Cstwo[\alpha,B]$ the set of all these
  \dchains. Observe that by definition, $\Rs$ has unbounded
  alternation. It follows from Proposition~\ref{prop:ctree} that one 
  can construct a set of \chain trees $\crr'$ whose set of values
  is exactly $\Rs$. By definition, $\crr'$ satisfies Items~1 and~2 and
  $s \cdot val(\crr') \cdot s'$ has unbounded alternation.

  We now use Proposition~\ref{prop:optimal} to construct \crr from
  $\crr'$ which is locally optimal for $(s,s')$ and satisfies $s \cdot
  val(\crr') \cdot s' = s \cdot val(\crr) \cdot s'$. We now know that
  \crr satisfies properties~1 to~3. Observe that by definition \crr has
  unbounded alternation. By hypothesis of
  Proposition~\ref{prop:graphcont2}, it follows that \crr has also 
  unbounded recursive alternation and all items are satisfied.

  \medskip
  \noindent
  {\bf Choosing a \chain tree $T \in \crr$.} We now select a special
  \chain tree $T$ in \crr. We want $T$ to have large enough recursive
  alternation in order to use it to construct the node 
  $((t,t'),v)$. We define the needed recursive alternation in the
  following lemma.

  \begin{lemma} \label{lem:chooseK}
    There exists $K \in \nat$ such that for all $t_1,t_2 \in M$ and all $C
    \subseteq A$, $(t_1,t_2)^K \in \Cslev 2[\alpha,C] \Rightarrow (t_1,t_2)^*
    \subseteq \Cslev 2[\alpha,C]$.
  \end{lemma}

  \begin{proof}
    It suffices to take $K$ as the largest $k$ such that there exists
    $t_1,t_2 \in M$ and $C \subseteq A$ with $(t_1,t_2)^{k-1} \in
    \Cslev 2[\alpha,C]$ but $(t_1,t_2)^{k} \not\in \Cslev 2[\alpha,C]$. \qed
  \end{proof}

  Set $m = |M|^2 \cdot K$ with $K$ as defined in
  Lemma~\ref{lem:chooseK}. By hypothesis on \crr (see property~4) there
  exists a tree $T \in \crr$  with recursive alternation $m$. We set $n$
  as the level of $T$.

  \medskip
  \noindent
  {\bf Construction of the node $((t,t'),v)$.} Set $r$ as the first
  element in \val{T}. Recall that by choice of $T$ in \crr, $srs' = u$. 
  By definition of recursive alternation, $T$ must contain an operation
  leaf $x$ whose label $\val{x} = (t_1,\dots,t_{n})$ has alternation
  $m$. Set $((s_1,\dots,s_n),(s'_1,\dots,s'_n)) = \cval{x}$ and $C =
  \content{x}$. Note that since $\content{T} = B$, $C \subseteq
  B$. Recall that by Fact~\ref{fct:value2}, we have 
  \[
  s \cdot \val{T} \cdot s' = s \cdot (s_1,\dots,s_n) \cdot
  (t_1,\dots,t_{n}) \cdot (s'_1,\dots,s'_n) \cdot s'
  \]
  Note that $(t_1,\dots,t_{n}) \in \Cstwo[\alpha,C]$, $(s_1,\dots,s_n) \in
  \Cstwo[\alpha]$ and $(s'_1,\dots,s'_n) \in \Cstwo[\alpha]$. We know that
  $(t_1,\dots,t_{n})$ has alternation $m = |M|^2 \cdot K$. It follows
  from a pigeon-hole principle argument that there exists $q_1 \neq q_2
  \in M$ and a set $I \subseteq  \{1,\dots,n-1\}$ of size at least $K$
  such that for all $i \in I$, $t_i = q_1$ and $t_{i+1} = q_2$. Observe
  that by definition, the \chain $(q_1,q_2)^{K}$ is a subword of
  $(t_1,\dots,t_{n})$ and therefore a \dchain for $\alpha,C$. By choice
  of $K$ it follows that $(q_1,q_2)^{*} \subseteq \Cstwo[\alpha,C]$. Note
  that this means that the node $((1_M,1_M),q_1)$ is
  $C$-alternating. Therefore, by minimality of $B$, we have $C = B$.
  Choose some arbitrary $i \in I$, say the first element in $I$. Recall
  that $T \in \crr$ and therefore locally optimal for
  $(s,s')$. The following fact is immediate by definition of local
  optimality: 

  \begin{fact} \label{fct:ccont}
    $ss_{i+1}q_1s_{i+1}'s' \neq ss_{i+1}q_2s_{i+1}'s'$.
  \end{fact}

  We now define the node $((t,t'),v)$. It is immediate from
  Fact~\ref{fct:ccont} that either $ss_{i+1}q_1s_{i+1}'s' \neq u$ or
  $ss_{i+1}q_2s_{i+1}'s' \neq u$, we set $v \neq u$ as this element.
  Finally, we set $t = ss_{i+1}$ and $t' = s'_{i+1}s'$. Observe that by
  Fact~\ref{fct:ccont} $tq_1t' \neq tq_2t'$, therefore since
  $(q_1,q_2)^{*} \subseteq \Cstwo[\alpha,B]$ and by choice of $v$, we 
  know that $((t,t'),v)$ is $B$-alternating. It remains to prove that
  $((s,s'),u) \xrightarrow{B} ((t,t'),v)$. We already know that $u =
  srs'$, $t = ss_{i+1}$ and $t' = s'_{i+1}s'$. We need to prove that
  $(r,s_{i+1},s'_{i+1})$ is a $B$-schema.

  \medskip

  Using the definition of operation nodes, we prove that $r=s_1s'_1$
  and define $\Ts \in \fCtwotwo[\alpha,B]$ such that $(s_1,s_{i+1}) \in 
  \Cstwotwo[\alpha,B] \cdot \Ts^{\omega}$ and $(s'_1,s'_{i+1}) \in
  \Ts^{\omega} \cdot \Cstwotwo[\alpha,B]$ which terminates the
  proof. Set $y$ as the parent of $x$. By definition, $y$ is an
  operation node, set $x_1,\dots,x_{2\ell_n + 1}$ as the children of $y$
  ($x = x_{\ell_n+1}$). By definition, 
  \[
  \Rs = \{\val{x_1},\dots, \val{x_{\ell_n}},\val{x_{\ell_n+ 2}},\dots,
  \val{x_{2\ell_n+1}}\} \in \fCtwolen{n}[\alpha,B]
  \]
  Set $t$ has the common first value of all \chains in $\Rs$ and
  $(\bar{q},\bar{q}') = \cval{y}$. By Fact~\ref{fct:value3}, we have   
  \begin{equation}
    \bar{s} = \bar{q} \cdot \val{x_1} \cdots \val{x_{\ell_n}} \text{ and }
    \val{x_{\ell_n+2}} \cdots \val{x_{2\ell_n+1}} \cdot \bar{q}' =
    \bar{s}'
    \label{eq:fin}
  \end{equation}
  By Fact~\ref{fct:value2}, and definition of operation nodes, $r = s_1
  t^{\ell_n} s'_1$. It follows that $r=s_1s'_1$.

  Since $T$ has alphabet $B$, we have $\bar{q} \in \Cs_{2,n}[\alpha,C]$
  for some $C \subseteq B$. Using~\eqref{eq:fin} and the definition of
  $\ell_n$ as $\omega(2^{M^n})$, we get that $\bar{s} \in
  \Cs_{2,n}[\alpha,C] \cdot \Rs^\omega$. Moreover, since $Rs^{\omega}
  \subseteq \Cs_{2,n}[\alpha,B]$, $\bar{s} \in \Cs_{2,n}[\alpha,B]$.
  Using a symetrical argument, we get that $\bar{s}' \in \Rs^\omega
  \cdot \Cs_{2,n}[\alpha,B]$.

  Finally, set \Ts as the set of \chains of length $2$ obtained from
  \chains in \Rs by keeping only the values at component $1$ and
  $i+1$. Since \dchains are closed under subwords, it is immediate from
  $\Rs \in \fCtwolen{n}[\alpha,B]$ that $\Ts \in \fCtwotwo[\alpha,B]$.
  Moreover, by definition, we have $(s_1,s_{i+1}) \in
  \Cstwotwo[\alpha,B] \cdot \Ts^{\omega}$ and $(s'_1,s'_{i+1}) \in 
  \Ts^{\omega} \cdot \Cstwotwo[\alpha,B]$. We conclude that
  $(r,s_{i+1},s'_{i+1})$ is a $B$-schema  which terminates the proof. \qed
\end{proof}

\section{Proof of Proposition~\ref{prop:width}}
\label{app:width}
\newcommand\alt[2]{\ensuremath{\textsf{alt}(#1,#2)}\xspace}
Recall that we fixed the morphism $\alpha: A^* \rightarrow M$. We prove Proposition~\ref{prop:width}.

\adjustc{prop:width}
\begin{proposition}
  Assume that there exists a set of locally optimal \chain trees $\cs
  \subseteq \ct[\alpha]$ with unbounded alternation but bounded
  recursive alternation. Then $\alpha$ does not satisfy
  Equation~\eqref{eq:bcs1}.
\end{proposition}
\restorec

As for the previous section, we will use a new object that is specific
to this case: \emph{\chain matrices}.

\medskip
\noindent
{\bf \Chain Matrices.} Let $n \in \nat$. A \emph{\chain matrix of
  length $n$} is a rectangular matrix with $n$ columns and such that
rows belong to $\Cstwolen{n}[\alpha]$. If \mat is a \chain matrix, we
will denote by $\mat_{i,j}$ the entry at row $i$ (starting from the
top) and column $j$ (starting from the left) in \mat. If \mat is a
\chain matrix of length $n$ and with $m$ rows, we call the \chain
$\bigl((\mat_{1,1} \cdots \mat_{m,1}), \dots, (\mat_{1,n} \cdots
\mat_{m,n})\bigr)$, the \emph{value} of \mat. By Fact~\ref{fct:chaincomp},
the value of a \chain matrix is a \dchain. We give an example with
$3$ rows in Figure~\ref{fig:valmat}.

\begin{figure}[h]
  \begin{center}
    \begin{tikzpicture}
      \node[anchor=mid] (s1) at (0.0,0) {$s_1$};
      \node[anchor=mid] (s2) at (1.0,0) {$s_2$};
      \node[anchor=mid] (s3) at (2.0,0) {$s_3$};
      \node[anchor=mid] (s4) at (3.0,0) {$s_4$};
      \node[anchor=mid] (s5) at (4.0,0) {$\cdots$};
      \node[anchor=mid] (sn) at (5.0,0) {$s_{n}$};

      \node[anchor=mid] (t1) at (0.0,-0.5) {$t_1$};
      \node[anchor=mid] (t2) at (1.0,-0.5) {$t_2$};
      \node[anchor=mid] (t3) at (2.0,-0.5) {$t_3$};
      \node[anchor=mid] (t4) at (3.0,-0.5) {$t_4$};
      \node[anchor=mid] (t5) at (4.0,-0.5) {$\cdots$};
      \node[anchor=mid] (tn) at (5.0,-0.5) {$t_{n}$};

      \node[anchor=mid] (r1) at (0.0,-1.0) {$r_1$};
      \node[anchor=mid] (r2) at (1.0,-1.0) {$r_2$};
      \node[anchor=mid] (r3) at (2.0,-1.0) {$r_3$};
      \node[anchor=mid] (r4) at (3.0,-1.0) {$r_4$};
      \node[anchor=mid] (r5) at (4.0,-1.0) {$\cdots$};
      \node[anchor=mid] (rn) at (5.0,-1.0) {$r_{n}$};

      \draw (-0.5,0.25) to (-0.5,-1.25);
      \draw (0.5,0.25) to (0.5,-1.25);
      \draw (1.5,0.25) to (1.5,-1.25);
      \draw (2.5,0.25) to (2.5,-1.25);
      \draw (3.5,0.25) to (3.5,-1.25);
      \draw (4.5,0.25) to (4.5,-1.25);
      \draw (5.5,0.25) to (5.5,-1.25);

      \draw (-0.5,0.25) to (5.5,0.25);
      \draw (-0.5,-0.25) to (5.5,-0.25);
      \draw (-0.5,-0.75) to (5.5,-0.75);
      \draw (-0.5,-1.25) to (5.5,-1.25);


      \node[anchor=mid] (p1) at (-0.55,-2.2) {$($};
      \node[anchor=mid] (v1) at (0.0,-2.2) {$s_1t_1r_1,$};
      \node[anchor=mid] (v2) at (1.0,-2.2) {$s_2t_2r_2,$};
      \node[anchor=mid] (v3) at (2.0,-2.2) {$s_3t_3r_3,$};
      \node[anchor=mid] (v4) at (3.0,-2.2) {$s_4t_4r_4,$};
      \node[anchor=mid] (v5) at (4.0,-2.2) {$\dots$};
      \node[anchor=mid] (vn) at (5.0,-2.2) {$,s_{n}t_{n}r_{n}$};
      \node[anchor=mid] (p2) at (5.8,-2.2) {$)$};
      \node[anchor=mid] (text) at (-2.0,-2.2) {Value};
      \draw[ar] (text) to (p1);

      \draw[->] ($(r1)-(0.0,0.4)$) to (v1);
      \draw[->] ($(r2)-(0.0,0.4)$) to (v2);
      \draw[->] ($(r3)-(0.0,0.4)$) to (v3);
      \draw[->] ($(r4)-(0.0,0.4)$) to (v4);
      \draw[->] ($(rn)-(0.0,0.4)$) to (vn);
    \end{tikzpicture}
  \end{center}
  \caption{Value of \chain matrix with $3$ rows}
  \label{fig:valmat}
\end{figure}
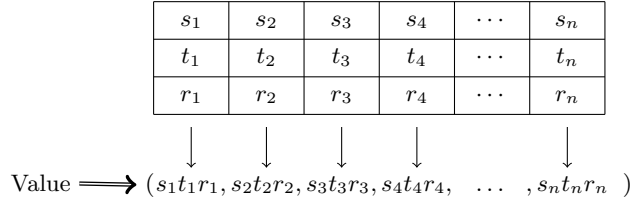

Given a \chain matrix, \mat, the \emph{alternation} of \mat is the
alternation of its value. Finally, the \emph{local alternation} of a
\chain matrix, \mat, is the largest natural $m$ such that \mat has a
row with alternation $m$. We now prove the two following propositions.

\begin{proposition} \label{prop:matinit}
  Assume that there exists a set of locally optimal \chain trees $\cs
  \subseteq \ct[\alpha]$ with unbounded alternation and recursive
  alternation bounded by $K \in \nat$. Then there exist \chain
  matrices with arbitrarily large alternation and local alternation
  bounded by $K$.   
\end{proposition}

\begin{proposition} \label{prop:contradend}
  Assume that there exist \chain matrices with arbitrarily large alternation
  and local alternation bounded by $K\in \nat$. Then $\alpha$ does not
  satisfy~\eqref{eq:bcs2}. 
\end{proposition}

Proposition~\ref{prop:width} is an immediate consequence of
Proposition~\ref{prop:matinit} and~\ref{prop:contradend}. Note that
\chain matrices are reused from~\cite{bpopen} (where they are called
''strategy matrices''). Moreover, in this case going from \chain trees
to \chains matrices (i.e. proving Proposition~\ref{prop:matinit}) is
simple and the main difficulty is proving
Proposition~\ref{prop:contradend}. This means that while our presentation
is slightly different from that of~\cite{bpopen}, the arguments
themselves are essentially the same. We give a full proof for
the sake of completeness. We begin by proving Proposition~\ref{prop:matinit}.

\begin{proof}[of Proposition~\ref{prop:matinit}]
  We prove that for all $n \in \nat$, there exists a \chain matrix \mat
  of alternation $n$ and local alternation bounded by $K$. By definition
  of $\cs$ there exists a tree $T \in \cs$ whose value has alternation
  $n$ and has recursive alternation bounded by $K$. Set $x_1,\dots,x_m$
  as leaves of $T$ listed from left to right. 
  By Fact~\ref{fct:value}, $\val{T} = \val{x_1} \cdots
  \val{x_m}$. Observe that by definition, for all $i$, $\val{x_i}$ has
  alternation bounded by $K$. Therefore it suffices to set \mat as the
  $m$ rows matrix where row $i$ is filled with $\val{x_i}$.\qed
\end{proof}

It now remains to prove Proposition~\ref{prop:contradend}. We proceed
as follows. Assuming there exists a \chain matrix \mat with local
alternation bounded by $K$ and very large alternation, we refine \mat
in several steps to ultimately obtain what we call a
\emph{contradiction matrix}. There are two types of contradiction
matrices, \emph{increasing} and \emph{decreasing}, both are
\chain matrices of length $6$ and with the following entries:

\begin{center}
  \begin{tikzpicture}

    \node (m1) at (0,0) {$\begin{array}{|c|c|c|c|c|c|}
        \hline
        u_1 & v_1 & f & f & f & f\\
        \hline
        e & e & u_2 & v_2 & f & f\\
        \hline
        e & e & e & e & u_3 & v_3\\
        \hline
      \end{array}$};
    \node (l1) at(0,-1.2) {Increasing Contradiction Matrix};

    \node (m2) at (6,0) {$\begin{array}{|c|c|c|c|c|c|}
        \hline
        f & f & f & f & u_3 & v_3\\
        \hline
        f & f & u_2 & v_2 & e & e\\
        \hline
        u_1 & v_1 & e & e & e & e\\
        \hline
      \end{array}$};
    \node (l2) at(6,-1.2) {Decreasing Contradiction Matrix};
  \end{tikzpicture}
\end{center}

\noindent
such that $e,f$ are idempotents and $fu_2e \neq fv_2e$. As the
name suggests, the existence of a contradiction matrix contradicts
Equation~\eqref{eq:bcs1}. This is what we state in the following
lemma.

\begin{lemma} \label{lem:contradmat}
  If there exists a contradiction matrix, $\alpha$ does not
  satisfy~\eqref{eq:bcs1}.
\end{lemma}

\begin{proof}
  Assume that we have an increasing contradiction matrix (the other case
  is treated in a symmetrical way). Since $fu_2e \neq fv_2e$, either
  $fu_2e \neq fe$ or $fv_2e \neq fe$. By symmetry assume it is the
  former. Since $e,f$ are idempotents, this means that $f^\omega u_2
  e^\omega \neq f^\omega e^\omega$. However by definition of \chain
  matrices $(e,u_2,v_2,f) \in \Cstwo[\alpha]$ and therefore  $(e,u_2,f)
  \in \Cstwo[\alpha]$ which contradicts Equation~\eqref{eq:bcs1}. Note 
  that we only used one half of Equation~\eqref{eq:bcs1}, the other half 
  is used in the decreasing case.\qed
\end{proof}

By Lemma~\ref{lem:contradmat}, it suffices to prove the existence of a
contradiction matrix to conclude the proof of
Proposition~\ref{prop:contradend}. This is what we do in the remainder
of this Appendix. By hypothesis, we know that there exist \chain
matrices with arbitrarily large alternation and local alternation
bounded by $K \in \nat$. For the remainder of the section, we assume
that this hypothesis holds. We use several steps to prove that we can
choose our \chain matrices with increasingly strong properties until
we get a contradiction matrix. We use two intermediaries that we 
call \emph{Tame \Chain Matrices} and \emph{Monotonous \Chain
  Matrices}. We divide the proof in three subsections, one for each
step.

\subsection{Tame \Chain Matrices}

Let \mat be a \chain matrix of \emph{even length} $2\ell$ and let $j \leq
\ell$. The \emph{set of alternating rows for $j$}, denoted by
\alt{\mat}{j}, is the set $\{i \mid \mat_{i,2j-1} \neq
\mat_{i,2j}\}$. Let 
$(s_1,\dots,s_{2\ell})$ be the value of \mat. We say that \mat is \emph{tame} if
\begin{itemize}
\item [$a)$] for all $j \leq \ell$, $s_{2j-1} \neq s_{2j}$,
\item [$b)$] for all $j \leq \ell$,
  \alt{\mat}{j} is a singleton and
\item [$c)$] if $j \neq j'$ then $\alt{\mat}{j}
  \neq \alt{\mat}{j'}$.
\end{itemize}
We represent a tame \chain matrix
of length $6$ in Figure~\ref{fig:tamemat}. Observe that the definition
only considers the relationship between odd columns and the next even
column. Moreover, observe that a tame \chain matrix of length $2\ell$
has by definition alternation at least $\ell$.

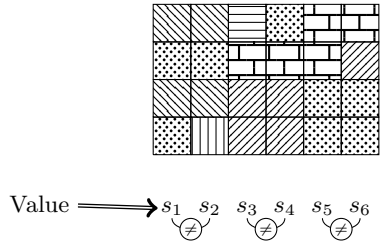
\begin{figure}[h]
  \begin{center}
    \begin{tikzpicture}
      \draw[pattern=north west lines,draw] (0.0,1.5) to (0.5,1.5) to
      (0.5,2.0) to (0.0,2.0) to (0.0,1.5); 
      \draw[pattern=crosshatch dots,draw] (0.0,1.0) to (0.5,1.0) to
      (0.5,1.5) to (0.0,1.5) to (0.0,1.0); 
      \draw[pattern=north west lines,draw] (0.0,0.5) to (0.5,0.5) to
      (0.5,1.0) to (0.0,1.0) to (0.0,0.5); 
      \draw[pattern=crosshatch dots,draw] (0.0,0.0) to (0.5,0.0) to
      (0.5,0.5) to (0.0,0.5) to (0.0,0.0); 

      \draw[pattern=north west lines,draw] (0.5,1.5) to (1.0,1.5) to
      (1.0,2.0) to (0.5,2.0) to (0.5,1.5); 
      \draw[pattern=crosshatch dots,draw] (0.5,1.0) to (1.0,1.0) to
      (1.0,1.5) to (0.5,1.5) to (0.5,1.0); 
      \draw[pattern=north west lines,draw] (0.5,0.5) to (1.0,0.5) to
      (1.0,1.0) to (0.5,1.0) to (0.5,0.5); 
      \draw[pattern=vertical lines,draw] (0.5,0.0) to (1.0,0.0) to (1.0,0.5)
      to (0.5,0.5) to (0.5,0.0);

      \draw[pattern=horizontal lines,draw] (1.0,1.5) to (1.5,1.5) to
      (1.5,2.0) to (1.0,2.0) to (1.0,1.5); 
      \draw[pattern=bricks,draw] (1.0,1.0) to (1.5,1.0) to
      (1.5,1.5) to (1.0,1.5) to (1.0,1.0); 
      \draw[pattern=north east lines,draw] (1.0,0.5) to (1.5,0.5) to
      (1.5,1.0) to (1.0,1.0) to (1.0,0.5); 
      \draw[pattern=north east lines,draw] (1.0,0.0) to (1.5,0.0) to
      (1.5,0.5) to (1.0,0.5) to (1.0,0.0); 

      \draw[pattern=crosshatch dots,draw] (1.5,1.5) to (2.0,1.5) to
      (2.0,2.0) to (1.5,2.0) to (1.5,1.5); 
      \draw[pattern=bricks,draw] (1.5,1.0) to (2.0,1.0) to
      (2.0,1.5) to (1.5,1.5) to (1.5,1.0); 
      \draw[pattern=north east lines,draw] (1.5,0.5) to (2.0,0.5) to
      (2.0,1.0) to (1.5,1.0) to (1.5,0.5); 
      \draw[pattern=north east lines,draw] (1.5,0.0) to (2.0,0.0) to
      (2.0,0.5) to (1.5,0.5) to (1.5,0.0);

      \draw[pattern=bricks,draw] (2.0,1.5) to (2.5,1.5) to
      (2.5,2.0) to (2.0,2.0) to (2.0,1.5); 
      \draw[pattern=bricks,draw] (2.0,1.0) to (2.5,1.0) to
      (2.5,1.5) to (2.0,1.5) to (2.0,1.0); 
      \draw[pattern=crosshatch dots,draw] (2.0,0.5) to (2.5,0.5) to
      (2.5,1.0) to (2.0,1.0) to (2.0,0.5); 
      \draw[pattern=crosshatch dots,draw] (2.0,0.0) to (2.5,0.0) to
      (2.5,0.5) to (2.0,0.5) to (2.0,0.0); 

      \draw[pattern=bricks,draw] (2.5,1.5) to (3.0,1.5) to
      (3.0,2.0) to (2.5,2.0) to (2.5,1.5); 
      \draw[pattern=north east lines,draw] (2.5,1.0) to (3.0,1.0) to
      (3.0,1.5) to (2.5,1.5) to (2.5,1.0); 
      \draw[pattern=crosshatch dots,draw] (2.5,0.5) to (3.0,0.5) to
      (3.0,1.0) to (2.5,1.0) to (2.5,0.5); 
      \draw[pattern=crosshatch dots,draw] (2.5,0.0) to (3.0,0.0) to
      (3.0,0.5) to (2.5,0.5) to (2.5,0.0); 

      \node[anchor=mid,inner sep=1pt] (s1) at (0.25,-0.7) {$s_1$};
      \node[anchor=mid,inner sep=1pt] (s2) at (0.75,-0.7) {$s_2$};
      \node[anchor=mid,inner sep=1pt] (s3) at (1.25,-0.7) {$s_3$};
      \node[anchor=mid,inner sep=1pt] (s4) at (1.75,-0.7) {$s_4$};
      \node[anchor=mid,inner sep=1pt] (s5) at (2.25,-0.7) {$s_5$};
      \node[anchor=mid,inner sep=1pt] (s6) at (2.75,-0.7) {$s_6$};

      \node[anchor=mid east] (text) at (-1.0,-0.7) {Value};
      \draw[ar] (text) to (s1);

      \draw (s1.south) to [in=-90,out=-90]
      node[sloped,draw,fill=white,circle,inner sep=0.5pt] {\tiny $\neq$}
      (s2.south);
      \draw (s3.south) to [in=-90,out=-90]
      node[sloped,draw,fill=white,circle,inner sep=0.5pt] {\tiny $\neq$}
      (s4.south);
      \draw (s5.south) to [in=-90,out=-90]
      node[sloped,draw,fill=white,circle,inner sep=0.5pt] {\tiny $\neq$}
      (s6.south);

    \end{tikzpicture}
  \end{center}
  \caption{A tame \chain matrix of length $6$}
  \label{fig:tamemat}
\end{figure}

\begin{lemma} \label{lem:tame}
  There exists tame \chain matrices of arbitrarily large length.
\end{lemma}

\begin{proof}
  Set $n \in \nat$, we explain how to construct a tame \chain matrix of
  length $2n$. By hypothesis, there exists a \chain matrix \mat with
  local alternation at most~$K$ and alternation greater than $2nK$. Set
  $m$ as the number of rows of \mat. We explain how to modify \mat to
  obtain a matrix satisfying $a)$, $b)$ and $c)$. Recall that \dchains
  are closed under subwords, therefore removing columns from \mat yields
  a \chain matrix. Since \mat has alternation $2nK$, it is simple to see
  that by removing columns one can obtain a \chain matrix of length
  $2nK$ that satisfies $a)$. We denote by \mnat this matrix. We now
  proceed in two steps: first, we modify the entries in \mnat to get a
  matrix \pat of length $2nK$ satisfying both $a)$ and $b)$. Then we use
  our bound on local alternation to remove columns and enforce $c)$ in
  the resulting matrix.

  \medskip
  \noindent
  {\bf Construction of \pat.} Let $j \leq nK$ such that \alt{\mnat}{j} 
  is of size at least $2$. We modify the matrix to reduce
  the size of \alt{\mnat}{j} while preserving $a)$. One can then repeat
  the operation to get the desired matrix. Let $i \in \alt{\mnat}{j}$.
  Set $s_{1} = \mnat_{1,2j-1} \cdots \mnat_{i-1,2j-1}$ and $s_2 =
  \mnat_{i+1,2j-1} \cdots \mnat_{m,2j-1}$. We distinguish two~cases.

  First, if $s_1\mnat_{i,2j-1}s_2 \neq s_1\mnat_{i,2j}s_2$, then for all $i' \neq i$,
  we replace entry $\mnat_{i',2j}$ with entry $\mnat_{i',2j-1}$. One can verify
  that this yields a \chain matrix of length $2nK$, local alternation bounded by
  $K$. Moreover, it still satisfies $a)$, since $s_1\mnat_{i,2j-1}s_2 \neq
  s_1\mnat_{i,2j}s_2$. Finally, \alt{\mnat}{j} is now a singleton, namely $\{i\}$.

  In the second case, we have $s_1\mnat_{i,2j-1}s_2 = s_1\mnat_{i,2j}s_2$. In that
  case, we replace $\mnat_{i,2j-1}$ with $\mnat_{i,2j}$. One can
  verify that this yields a \chain matrix of length $2nK$, local
  alternation bounded by $K$. Moreover, it still satisfies $a)$ since we did not
  change the value on the whole. Finally,
  the size of \alt{\mnat}{j} has decreased by~$1$.

  \medskip
  \noindent
  {\bf Construction of the tame matrix.} We now have a \chain matrix
  \pat of length $2nK$, with local alternation bounded by $K$ and
  satisfying both $a)$ and $b)$. Since $a)$ and $b)$ are satisfied, for
  all $j \leq nK$ there exists exactly one row $i$ such that
  $\mnat_{i,2j-1} \neq \mnat_{i,2j}$. Moreover, since each row has
  alternation at most $K$, a single row $i$ has this property for at
  most $K$ indices $j$. Therefore, it suffices to remove at most
  $n(K-1)$ pairs of odd-even columns to get a matrix that satisfies
  $c)$. Since the original matrix had length $2nK$, this leaves a matrix
  of length at least $2n$ and we are finished. \qed 
\end{proof}

\subsection{Monotonous \Chain Matrices}

Let \mat be a tame \chain matrix of length $2n$ and let $x_1,\dots,x_n$
be naturals such that for all $j$, $\alt{\mat}{j} = \{x_j\}$. We say
that \mat is a \emph{monotonous \chain matrix} if it has exactly $n$
rows and $1 = x_{1} < x_2 < \cdots < x_{n} = n$ (in which case the
matrix is said \emph{increasing}) or $n = x_{1} > x_2 > \cdots > x_{n} = 1$
(in which case we say the matrix is \emph{decreasing}). We give a 
representation of the increasing case in Figure~\ref{fig:inctame}.

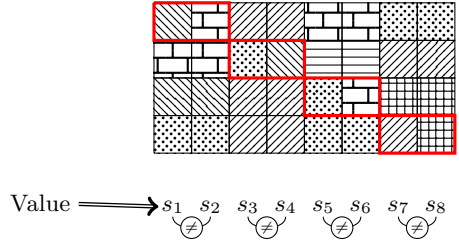
\begin{figure}[h]
  \begin{center}
    \begin{tikzpicture}
      \draw[pattern=north west lines,draw] (0.0,1.5) to (0.5,1.5) to
      (0.5,2.0) to (0.0,2.0) to (0.0,1.5); 
      \draw[pattern=bricks,draw] (0.0,1.0) to (0.5,1.0) to
      (0.5,1.5) to (0.0,1.5) to (0.0,1.0); 
      \draw[pattern=north west lines,draw] (0.0,0.5) to (0.5,0.5) to
      (0.5,1.0) to (0.0,1.0) to (0.0,0.5); 
      \draw[pattern=crosshatch dots,draw] (0.0,0.0) to (0.5,0.0) to
      (0.5,0.5) to (0.0,0.5) to (0.0,0.0); 

      \draw[pattern=bricks,draw] (0.5,1.5) to (1.0,1.5) to
      (1.0,2.0) to (0.5,2.0) to (0.5,1.5); 
      \draw[pattern=bricks,draw] (0.5,1.0) to (1.0,1.0) to
      (1.0,1.5) to (0.5,1.5) to (0.5,1.0); 
      \draw[pattern=north west lines,draw] (0.5,0.5) to (1.0,0.5) to
      (1.0,1.0) to (0.5,1.0) to (0.5,0.5); 
      \draw[pattern=crosshatch dots,draw] (0.5,0.0) to (1.0,0.0) to (1.0,0.5)
      to (0.5,0.5) to (0.5,0.0);

      \draw[pattern=north east lines,draw] (1.0,1.5) to (1.5,1.5) to
      (1.5,2.0) to (1.0,2.0) to (1.0,1.5); 
      \draw[pattern=crosshatch dots,draw] (1.0,1.0) to (1.5,1.0) to
      (1.5,1.5) to (1.0,1.5) to (1.0,1.0); 
      \draw[pattern=north east lines,draw] (1.0,0.5) to (1.5,0.5) to
      (1.5,1.0) to (1.0,1.0) to (1.0,0.5); 
      \draw[pattern=north east lines,draw] (1.0,0.0) to (1.5,0.0) to
      (1.5,0.5) to (1.0,0.5) to (1.0,0.0); 

      \draw[pattern=north east lines,draw] (1.5,1.5) to (2.0,1.5) to
      (2.0,2.0) to (1.5,2.0) to (1.5,1.5); 
      \draw[pattern=north west lines,draw] (1.5,1.0) to (2.0,1.0) to
      (2.0,1.5) to (1.5,1.5) to (1.5,1.0); 
      \draw[pattern=north east lines,draw] (1.5,0.5) to (2.0,0.5) to
      (2.0,1.0) to (1.5,1.0) to (1.5,0.5); 
      \draw[pattern=north east lines,draw] (1.5,0.0) to (2.0,0.0) to
      (2.0,0.5) to (1.5,0.5) to (1.5,0.0);

      \draw[pattern=bricks,draw] (2.0,1.5) to (2.5,1.5) to
      (2.5,2.0) to (2.0,2.0) to (2.0,1.5); 
      \draw[pattern=horizontal lines,draw] (2.0,1.0) to (2.5,1.0) to
      (2.5,1.5) to (2.0,1.5) to (2.0,1.0); 
      \draw[pattern=crosshatch dots,draw] (2.0,0.5) to (2.5,0.5) to
      (2.5,1.0) to (2.0,1.0) to (2.0,0.5); 
      \draw[pattern=crosshatch dots,draw] (2.0,0.0) to (2.5,0.0) to
      (2.5,0.5) to (2.0,0.5) to (2.0,0.0); 

      \draw[pattern=bricks,draw] (2.5,1.5) to (3.0,1.5) to
      (3.0,2.0) to (2.5,2.0) to (2.5,1.5); 
      \draw[pattern=horizontal lines,draw] (2.5,1.0) to (3.0,1.0) to
      (3.0,1.5) to (2.5,1.5) to (2.5,1.0); 
      \draw[pattern=bricks,draw] (2.5,0.5) to (3.0,0.5) to
      (3.0,1.0) to (2.5,1.0) to (2.5,0.5); 
      \draw[pattern=crosshatch dots,draw] (2.5,0.0) to (3.0,0.0) to
      (3.0,0.5) to (2.5,0.5) to (2.5,0.0); 

      \draw[pattern=crosshatch dots,draw] (3.0,1.5) to (3.5,1.5) to
      (3.5,2.0) to (3.0,2.0) to (3.0,1.5); 
      \draw[pattern=north east lines,draw] (3.0,1.0) to (3.5,1.0) to
      (3.5,1.5) to (3.0,1.5) to (3.0,1.0); 
      \draw[pattern=grid,draw] (3.0,0.5) to (3.5,0.5) to
      (3.5,1.0) to (3.0,1.0) to (3.0,0.5); 
      \draw[pattern=north east lines,draw] (3.0,0.0) to (3.5,0.0) to
      (3.5,0.5) to (3.0,0.5) to (3.0,0.0); 

      \draw[pattern=crosshatch dots,draw] (3.5,1.5) to (4.0,1.5) to
      (4.0,2.0) to (3.5,2.0) to (3.5,1.5); 
      \draw[pattern=north east lines,draw] (3.5,1.0) to (4.0,1.0) to
      (4.0,1.5) to (3.5,1.5) to (3.5,1.0); 
      \draw[pattern=grid,draw] (3.5,0.5) to (4.0,0.5) to
      (4.0,1.0) to (3.5,1.0) to (3.5,0.5); 
      \draw[pattern=grid,draw] (3.5,0.0) to (4.0,0.0) to
      (4.0,0.5) to (3.5,0.5) to (3.5,0.0); 

      \draw[very thick,red] (0.0,2.0) to (1.0,2.0) to
      (1.0,1.5) to (0.0,1.5) to (0.0,2.0);

      \draw[very thick,red] (1.0,1.5) to (2.0,1.5) to
      (2.0,1.0) to (1.0,1.0) to (1.0,1.5);

      \draw[very thick,red] (2.0,1.0) to (3.0,1.0) to
      (3.0,0.5) to (2.0,0.5) to (2.0,1.0);

      \draw[very thick,red] (3.0,0.5) to (4.0,0.5) to
      (4.0,0.0) to (3.0,0.0) to (3.0,0.5);

      \node[anchor=mid,inner sep=1pt] (s1) at (0.25,-0.7) {$s_1$};
      \node[anchor=mid,inner sep=1pt] (s2) at (0.75,-0.7) {$s_2$};
      \node[anchor=mid,inner sep=1pt] (s3) at (1.25,-0.7) {$s_3$};
      \node[anchor=mid,inner sep=1pt] (s4) at (1.75,-0.7) {$s_4$};
      \node[anchor=mid,inner sep=1pt] (s5) at (2.25,-0.7) {$s_5$};
      \node[anchor=mid,inner sep=1pt] (s6) at (2.75,-0.7) {$s_6$};
      \node[anchor=mid,inner sep=1pt] (s7) at (3.25,-0.7) {$s_7$};
      \node[anchor=mid,inner sep=1pt] (s8) at (3.75,-0.7) {$s_8$};

      \node[anchor=mid east] (text) at (-1.0,-0.7) {Value};
      \draw[ar] (text) to (s1);

      \draw (s1.south) to [in=-90,out=-90]
      node[sloped,draw,fill=white,circle,inner sep=0.5pt] {\tiny $\neq$}
      (s2.south);
      \draw (s3.south) to [in=-90,out=-90]
      node[sloped,draw,fill=white,circle,inner sep=0.5pt] {\tiny $\neq$}
      (s4.south);
      \draw (s5.south) to [in=-90,out=-90]
      node[sloped,draw,fill=white,circle,inner sep=0.5pt] {\tiny $\neq$}
      (s6.south);
      \draw (s7.south) to [in=-90,out=-90]
      node[sloped,draw,fill=white,circle,inner sep=0.5pt] {\tiny $\neq$}
      (s8.south);
    \end{tikzpicture}
  \end{center}
  \caption{A monotonous \chain matrix (increasing)}
  \label{fig:inctame}
\end{figure}

\begin{lemma} \label{lem:mono}
  There exists monotonous \chain matrices of arbitrarily large length.
\end{lemma}

\begin{proof}
  Set $n \in \nat$, we explain how to construct a tame \chain matrix of
  length $2n$. By Lemma~\ref{lem:tame}, there exists a tame \chain matrix
  \mat of length $2n^2$. Set $x_1,\dots,x_{n^2}$ the indices such that
  for all $j$, $\alt{\mat}{j} = \{x_j\}$. Note that by tameness, $x_j
  \neq x_{j'}$ for $j \neq j'$. Since the sequence $x_1,\dots,x_{n^2}$
  is of length $n^2$, we can extract, using Erd\"os-Szekeres theorem, a
  monotonous sequence of length $n$, $x_{j_1} < \cdots < x_{j_n}$ or
  $x_{j_1} > \cdots > x_{j_n}$ with $j_1 < \cdots < j_n$.  By symmetry
  we assume it is the former and construct an increasing \chain matrix
  of length $n$.

  Let \pat be the matrix of length $2n$ obtained from \mat, by keeping
  only the pairs of columns $2j-1,2j$ for $j \in \{j_1,\dots,j_n\}$. Set
  $x'_1,\dots,x'_{n}$ the indices such that for all $j$, $\alt{\pat}{j}
  = \{x'_j\}$. By definition, $x'_{1} < \cdots < x'_{n}$. We now want
  $\pat$ to have exactly $n$ rows. Note that the rows  that do not belong to
  $x'_{1} < \cdots < x'_{n}$ are constant chains. We simply merge these rows
  with others. For example, if row $i$ is labeled with the 
  constant \chain $(s,\dots,s)$, let $(s_1,\dots,s_{2n})$ be the label of
  row $i+1$. We remove row $i$ and replace row $i+1$ by the \dchain
  $(ss_1,\dots,ss_{2n})$. Repeating the operation yields the desired increasing
  monotonous \chain~matrix. \qed
\end{proof}

\subsection{Construction of the Contradiction Matrix}  

We can now use Lemma~\ref{lem:mono} to construct a contradiction
matrix and end the proof of Proposition~\ref{prop:width}. We state
this in the following proposition.

\begin{proposition} \label{prop:contmat}
  There exists a contradiction matrix.
\end{proposition}

The remainder of this appendix is devoted to the proof of
Proposition~\ref{prop:contmat}. The result follows from a Ramsey
argument. We use Lemma~\ref{lem:mono} to choose a monotonous matrix of
sufficiently large length. Then, we use Ramsey's Theorem (for
hypergraphs with edges of size $3$) to extract the desired
contradiction matrix.

We first define the length of the monotonous \chain matrix that we need to
pick. By Ramsey's Theorem, for every $m \in \nat$ there exists a number
$\varphi(m)$ such that for any complete 3-hypergraph with hyperedges colored
over the monoid~$M$, there exists a complete sub-hypergraph of size $m$ in
which all edges share the same color. We choose $n = \varphi(\varphi(4)+1)$. By Lemma~\ref{lem:mono}, there exists a monotonous
\chain matrix \mat of length $2n$. Since it is monotonous, \mat has $n$ rows.

By symmetry, we assume that \mat is increasing and use it to construct an
increasing contradiction matrix.  We use our choice of $n$ to extract a
contradiction matrix from \mat. We proceed in two steps using Ramsey's Theorem
each time. In the first step we treat all entries above the diagonal in \mat
and in the second step all entries below the diagonal. We state the first step
in the next lemma.

\begin{lemma} \label{lem:matlemma}
  There exists an increasing monotonous matrix \mnat of length $2 \cdot
  \varphi(4)$ such that all cells above the diagonal contain the same
  idempotent $f \in M$.
\end{lemma}

\begin{proof}
  This is proved by applying Ramsey's Theorem to \mat. Consider the
  complete 3-hypergraph whose nodes are $\{0,\ldots,n\}$. We label the
  hyperedge $\{i_1,i_2,i_3\}$ where $i_1 < i_2 < i_3$  by the value obtained by
  multiplying in the monoid $M$, the cells that appear in rows
  $i_1+1,\ldots,i_2$ in column $2i_3-1$. Observe that since $i_1 < i_2 < i_3$, by monotonicity, these entries are the
  same as in column $2i_3$. More formally, the label of the hyperedge $\{i_1,i_2,i_3\}$ is therefore
  \[
  \mat_{i_1+1,2i_3-1} \cdots \mat_{i_2,2i_3-1} = \mat_{i_1+1,2i_3} \cdots \mat_{i_2,2i_3}.
  \]
  By choice of $n$, we can apply Ramsey's Theorem to this coloring. We get a
  subset of $\varphi(4)+1$ vertices, say $K = \{k_1,\ldots,k_{\varphi(4)+1}\} \subseteq
  \{0,\ldots,n\}$, such that all hyperedges connecting nodes in $K$ have the same
  color, say $f \in M$. For $i_1<i_2<i_3<i_4$ in~$K$, note that the color of the hyperedge
  $\{i_1,i_3,i_4\}$ is by definition the product of the colors of the hyperedges 
  $\{i_1,i_2,i_4\}$ and $\{i_2,i_3,i_4\}$. Therefore, the common color $f$ needs to be an idempotent
  (i.e. $f = ff$). We now extract the desired matrix \mnat from \mat according
  to the subset $K$. The main idea is that the new row $i$ in \mnat will be the
  merging of rows $k_{i}+1$ to $k_{i+1}$ in \mat and the new pair of columns
  $2j-1,2j$ will correspond to the pair $2k_{j+1}-1,2k_{j+1}$ in \mat.

  We first merge rows. For all $i \geq 1$, we ''merge'' all rows from
  $k_{i}+1$ to $k_{i+1}$ into a single row. More precisely, this means
  that we replace the rows $k_{i}+1$ to $k_{i+1}$ by a single row
  containing the \dchain 
  \[
  (\mat_{k_{i}+1,1} \cdots \mat_{k_{i+1},1},\ldots,\mat_{k_{i}+1,2n} \cdots \mat_{k_{i+1},2n})
  \]

  Moreover, we remove the top and bottom rows, i.e. row $1$ to $k_1$ and
  rows $k_{\varphi(4)+1}$ to $\varphi(4)+1$. Then we remove all columns from $1$ to
  $2k_2 -2$, all columns from $2k_{\varphi(4)+1} + 1$ to $2n$, and for all $i
  \geq 2$, all columns from $2k_i+1$ to $2k_{i+1}-2$. One can verify
  that these two operations applied together preserve
  monotonicity. Observe that the resulting matrix \mnat has exactly
  $2\cdot \varphi(4)$ columns. Moreover, the cell $i,2j$ in the new
  matrix contains entry $\mat_{k_{i}+1,2k_{j+1}} \cdots
  \mat_{k_{i+1},2k_{j+1}}$. In particular if $j > i$, by definition of
  the set $K$, this entry is $f$, which means \mnat satisfies the
  conditions of the lemma.\qed
\end{proof}

It remains to apply Ramsey's Theorem a second time to the matrix \mnat
obtained from Lemma~\ref{lem:matlemma} to treat the cells below the
diagonal and get the contradiction matrix. We state this in the
following last lemma.

\begin{lemma}
  There exists an increasing monotonous matrix \pat of length $6$ such
  that all cells above the diagonal contain the same idempotent $f \in
  M$ and all cells below the diagonal contain the same idempotent $e \in
  M$ (i.e. \pat is an increasing contradiction matrix). 
\end{lemma}

\begin{proof}
  The argument is identical to the one of Lemma~\ref{lem:matlemma}. This
  time we apply it to the matrix \mnat of length $2 \cdot \varphi(4)$ for the
  cells below the diagonal.\qed
\end{proof}


\begin{thebibliography}{10}

\bibitem{MR1709911}
J.~Almeida.
\newblock Some algorithmic problems for pseudovarieties.
\newblock {\em Publ. Math. Debrecen}, 54:531--552, 1999.
\newblock {P}roc. of {Automata and Formal Languages, VIII}.

\bibitem{AK2010}
J.~Almeida and O.~Kl\'{\i}ma.
\newblock New decidable upper bound of the 2nd level in the
  {S}traubing-{T}h{\'e}rien concatenation hierarchy of star-free languages.
\newblock {\em {DMTCS}}, 2010.

\bibitem{arfi87}
M.~Arfi.
\newblock Polynomial operations on rational languages.
\newblock In {\em STACS'87}, 1987.

\bibitem{bfacto}
M.~Bojanczyk.
\newblock Factorization forests.
\newblock In {\em DLT'09}, pages 1--17, 2009.

\bibitem{bpopen}
M.~Bojanczyk and T.~Place.
\newblock Regular languages of infinite trees that are boolean combinations of
  open sets.
\newblock In {\em ICALP'12}, pages 104--115, 2012.

\bibitem{BroKnaStrict}
J.~Brzozowski and R.~Knast.
\newblock The dot-depth hierarchy of star-free languages is infinite.
\newblock {\em J. Comp. Syst. Sci.}, 16(1):37--55, 1978.

\bibitem{BrzoDot}
R.~S. Cohen and J.~Brzozowski.
\newblock Dot-depth of star-free events.
\newblock {\em J. Comp. Syst. Sci.}, 5:1--16, 1971.

\bibitem{martens}
W.~Czerwinski, W.~Martens, and T.~Masopust.
\newblock Efficient separability of regular languages by subsequences and
  suffixes.
\newblock In {\em ICALP'13}, pages 150--161, 2013.

\bibitem{kfacto}
M.~Kufleitner.
\newblock The height of factorization forests.
\newblock In {\em MFCS'08}, 2008.

\bibitem{mnpfo}
R.~McNaughton and S.~Papert.
\newblock {\em Counter-Free Automata}.
\newblock {MIT} Press, 1971.

\bibitem{pinbridges}
J.-E. Pin.
\newblock {Bridges for concatenation hierarchies}.
\newblock In {\em {ICALP'98}}, 1998.

\bibitem{Pin-ThemeVar2011}
J.-E. Pin.
\newblock Theme and variations on the concatenation product.
\newblock In {\em 4th Int. Conf. on Algebraic Informatics}, pages 44--64.
  Springer, 2011.

\bibitem{pin-straubing:upper}
J.-E. Pin and H.~Straubing.
\newblock {Monoids of upper triangular boolean matrices}.
\newblock In {\em {Semigroups. Structure and Universal Algebraic Problems}},
  volume~39 of {\em Colloquia Mathematica Societatis Janos Bolyal}, pages
  259--272. North-Holland, 1985.

\bibitem{pwdelta}
J.-E. Pin and P.~Weil.
\newblock Polynomial closure and unambiguous product.
\newblock {\em Theory of Computing Systems}, 30(4):383--422, 1997.

\bibitem{pinweilVD}
J.-E. Pin and P.~Weil.
\newblock {The wreath product principle for ordered semigroups}.
\newblock {\em Communications in Algebra}, 30:5677--5713, 2002.

\bibitem{pvzmfcs13}
T.~Place, L.~{\swap{Rooijen}{van }}, and M.~Zeitoun.
\newblock Separating regular languages by piecewise testable and unambiguous
  languages.
\newblock In {\em MFCS'13}, pages 729--740, 2013.

\bibitem{pvzltt}
T.~Place, L.~van Rooijen, and M.~Zeitoun.
\newblock Separating regular languages by locally testable and locally
  threshold testable languages.
\newblock In {\em FSTTCS'13}, LIPIcs, 2013.

\bibitem{pz:qalt:2014}
T.~Place and M.~Zeitoun.
\newblock Going higher in the first-order quantifier alternation hierarchy on
  words.
\newblock {\em Arxiv}, 2014.

\bibitem{pzfo}
T.~Place and M.~Zeitoun.
\newblock Separating regular languages with first-order logic.
\newblock In {\em CSL-LICS'14}, 2014.

\bibitem{sfo}
M.~P. Sch{\"u}tzenberger.
\newblock On finite monoids having only trivial subgroups.
\newblock {\em Information and Control}, 8:190--194, 1965.

\bibitem{simon75}
I.~Simon.
\newblock Piecewise testable events.
\newblock In {\em 2nd GI Conference on Automata Theory and Formal Languages},
  pages 214--222, 1975.

\bibitem{simonfacto}
I.~Simon.
\newblock Factorization forests of finite height.
\newblock {\em {TCS}}, 72(1):65--94, 1990.

\bibitem{StrauConcat}
H.~Straubing.
\newblock A generalization of the {S}ch{\"u}tzenberger product of finite
  monoids.
\newblock {\em {{TCS}}}, 1981.

\bibitem{StrauVD}
H.~Straubing.
\newblock Finite semigroup varieties of the form {V {\textasteriskcentered} D}.
\newblock {\em J. Pure App. Algebra}, 36:53--94, 1985.

\bibitem{StrauDD2}
H.~Straubing.
\newblock Semigroups and languages of dot-depth two.
\newblock {\em {TCS}}, 1988.

\bibitem{bookstraub}
H.~Straubing.
\newblock {\em Finite Automata, Formal Logic and Circuit Complexity}.
\newblock 1994.

\bibitem{Tesson02diamondsare}
P.~Tesson and D.~Therien.
\newblock Diamonds are forever: The variety {DA}.
\newblock In {\em Semigroups, Algorithms, Automata and Languages}, pages
  475--500. World Scientific, 2002.

\bibitem{TheConcat}
D.~Th{\'e}rien.
\newblock Classification of finite monoids: the language approach.
\newblock {\em {{TCS}}}, 1981.

\bibitem{twfodeux}
D.~Th\'{e}rien and T.~Wilke.
\newblock Over words, two variables are as powerful as one quantifier
  alternation.
\newblock In {\em STOC'98}, pages 234--240. ACM, 1998.

\bibitem{Thom82}
W.~Thomas.
\newblock Classifying regular events in symbolic logic.
\newblock {\em J. Comp. Syst. Sci.}, 1982.

\bibitem{ThomStrict}
W.~Thomas.
\newblock A concatenation game and the dot-depth hierarchy.
\newblock In {\em Computation Theory and Logic}, pages 415--426. 1987.

\end{thebibliography}
\end{document}